\RequirePackage[final]{graphicx}
\documentclass[UKenglish,cleveref, autoref]{lmcs}
\pdfoutput=1

\usepackage[utf8]{inputenc}

\usepackage{lastpage}
\lmcsdoi{17}{1}{1}
\lmcsheading{}{\pageref{LastPage}}{}{}%
{Jan.~23,~2020}{Jan.~21,~2021}{}

\usepackage{graphicx}
\usepackage[draft,nomargin,inline]{fixme} 

\fxsetup{theme=color,mode=multiuser}
\FXRegisterAuthor{MB}{aMB}{MB}
\newlength{\spacelen}
\newcommand{\tabspace}[1]
        {\settowidth{\spacelen}{$#1$}
         \hspace*{\spacelen} }

\newcommand{\negspace}[1]
        {\settowidth{\spacelen}{$#1$}
         \hspace*{-\spacelen}}

\def\ms#1{\null\ifmmode\mathord{\mathcode`-="702D\it #1\mathcode`\-="2200}%
        \else$\mathord{\mathcode`-="702D\it #1\mathcode`\-="2200}$\fi}

\newcommand{\cws}[3]
        {\\[-#1pt] $$ {#3} $$ \\[-#2pt]}

\newlength{\widtharg} \newlength{\heightarg}
\newcommand{\cala}
        {{\mathcal A}}

\newcommand{\calb}
        {{\mathcal B}}
 
\newcommand{\cale}
        {{\mathcal E}}

\newcommand{\calp}
        {{\mathcal P}}

\newcommand{\snil}
	{\underline 1}

\newcommand{\narrow}[1]
        {\arrow{#1}\!\!\!\!\!\!\!\!\!\!\!\!/\,\,\,\,\,\,\,\,\,\,}

\newcommand{\sarrow}[1]
        {\, \auxsarrow\limits^{#1} \,}
\newcommand{\auxsarrow}
        {\mathop{\longrightarrow}}

\newcommand{\arrow}[1]
        {\, \auxarrow\limits^{#1} \,}
\newcommand{\notarrow}[1]
        {\, \notauxarrow\limits^{#1} \,}
\newcommand{\arrowd}[2]
        {\, {\auxarrow\limits^{#1}}_{#2} \,}
\newcommand{\narrowd}[2]
        {\, {\notauxarrow\limits^{#1}}_{#2} \,}
\newcommand{\auxarrow}
        {\mathop{\longrightarrow}}
\newcommand{\notauxarrow}
        {\mathop{\not \!\! \longrightarrow}}

\newcommand{\warrow}[1]
        {\, \auxwarrow\limits^{#1} \,}
\newcommand{\auxwarrow}
        {\mathop{\Longrightarrow}}

\newcommand{\nil}
        {\underline 0}

\newcommand{\tch}
        {+^t \,}

\newcommand{\pco}[1]
        {\mathop{\Vert_{#1}}}

\newcommand{\lme}[1]
        {\mathop{\lfloor \;\!\!\! \lfloor_{#1}}}

\newcommand{\sme}[1]
        {\mathop{|_{#1}}}

\newcommand{\notexists}
        {\not \!\! \exists}

\newcommand{\infr}[2]
        {\renewcommand{\arraystretch}{1.5}
        \begin{array}{c}
        #1\\
        \hline
        #2
        \end{array}}

\input{epsf}

\begin{document}

\title{Axiomatizing Maximal Progress and Discrete Time}

\author[M.~Bravetti]{Mario Bravetti}
\address{University of Bologna, Italy / INRIA FoCUS Team, France}
\email{mario.bravetti@unibo.it}{}{}
 







%
\thanks{Research partly supported by the
H2020-MSCA-RISE project ID 778233 ``Behavioural Application Program Interfaces
(BEHAPI)''}

\begin{abstract}




\noindent
Milner's complete proof system for observational congruence 
is crucially based on the possibility to equate $\tau$ divergent 
expressions to non-divergent ones by means of the axiom 
$\ms{rec} X. (\tau.X + E) = \ms{rec} X. \tau. E$. 
In the presence of a notion of priority, where, e.g., 
actions of type $\delta$ have a lower priority than silent 
$\tau$ actions, this
axiom is no longer sound. Such a form of priority is, however, common in 
timed process algebra, where, due to the interpretation of $\delta$ as a time delay, it naturally arises from the 
maximal progress assumption.
We here present our solution, based on introducing an auxiliary operator 
$\ms{pri}(E)$ defining a ``priority scope'',  to the long time open problem of axiomatizing priority using standard
observational congruence: 
we provide 
a complete axiomatization 
for a basic process algebra with priority and (unguarded) recursion.
We also show that,
when the setting is extended by considering static operators of a discrete time calculus,
an axiomatization that is complete over (a characterization of) finite-state terms can be developed by re-using 
techniques devised in the context of a cooperation with Prof. Jos Baeten. 


%
\end{abstract}


\maketitle

\section{Introduction}


%
%

The necessity of extending the expressiveness of classical
process algebras, so to make them suitable for the specification and
analysis of real case studies, led to the definition of several timed
calculi (see, e.g.,~\cite{Yi,NS,Han,Hil,HenR,Her,Ber,mtcs00,tcs,tcs3,Bra,BMM,ABDGHW,mtcs02,jlamp18}). Even if expressing different ``kind''
of times, e.g.\ discrete time or stochastic time,
often such extensions led to the necessity
of introducing a notion of priority among actions. A quite common technical design choice is to exploit priority for 
enacting the so-called {\it maximal progress assumption}~\cite{Yi,NS,HenR}:
the possibility of executing internal transitions prevents the 
execution of timed
transitions, thus expressing that the system cannot wait if it has something
internal to do. 

Technical problems arising from expressing maximal progress were
studied in the context of prioritized process algebras, see, 
e.g.,~\cite{CH,NCCC,CLN} and~\cite{CLN2} for a survey. 
One of the open questions in this context (see~\cite{CLN2}) was 
finding a complete axiomatization for observational congruence in the presence
of (unguarded) recursion.

In this paper we present our solution to such a problem. We proceed as follows. First, in Section~\ref{SectAPFSB},
we do it in the context of a {\it basic calculus}, presenting the results in \cite{icalp00} and their full technical machinery. In particular, we consider the algebra of
finite-state agents (made up of choice, prefix and recursion only),
used by Milner in~\cite{Mil89} for
axiomatizing observational congruence in presence of recursion, and we
extend it with $\delta$ prefixing, where
$\delta$ actions have lower {\it priority} than internal $\tau$ actions.
Such a calculus can be interpreted in this way: 
$\delta$ actions
represent ``generic'' time delays, standard actions are executed in zero time,
and the priority of $\tau$ actions over $\delta$ actions derives from the
maximal progress assumption. 
As in~\cite{CH,NCCC,CLN,CLN2,HL} we assume that visible actions
never have pre-emptive power over lower priority $\delta$ actions, because 
we see visible actions as indicating only the potential for execution.
The presence of such a priority mechanism
makes the standard Milner's complete proof system
for observational congruence (with $\delta$ actions being treated as visible actions) no longer sound. In particular this happens
for the fundamental $Ung$ axiom
$\ms{rec} X. (\tau.X + E) = \ms{rec} X. \tau. E$ (a $\delta$ action 
performable by $E$
is pre-empted in the left-hand term but not in the right-hand term),
which makes it possible to equate $\tau$ divergent
expressions to non-divergent ones, so to remove unguarded recursion.
Such a problem was previously faced in~\cite{HL}, where an alternative axiom was found 
in the context of a $\tau$-divergent sensitive equivalence that is finer than observational congruence/weak bisimulation.
Our solution, which does not require a modification of the equivalence, is based on the idea of introducing an auxiliary operator
$\ms{pri}(E)$ that cuts behaviours in $E$ starting with an unprioritized $\delta$ action: it allows us to express the result of applying priority within a certain {\it scope}.
By suitably modifying the axiom above and by
introducing some new axioms, our technique provides
a complete axiomatization for Milner's standard observational congruence 
over a basic calculus with this simple kind of
priority and (unguarded) recursion.


Then, in Section~\ref{SectDRT}, we consider a full timed calculus, including also static operators like parallel composition. 
In particular, we interpret unprioritized $\delta$ actions of the basic calculus as representing time delays in the context of discrete time  (see~\cite{HenR} and the references therein), where a $\delta$ action is called a ``tick''. Ticks take a fixed (unspecified)
amount of time, 
which is the same for all processes,
and are assumed to synchronize over all system processes.

We first show that, under this ``timed'' interpretation for $\delta$ actions, we can extend the basic calculus of Section~\ref{SectAPFSB} with {\it static} operators, like CSP~\cite{Hoa} parallel composition and hiding, preserving the congruence property of standard observational congruence.
We then consider a full discrete time calculus, which, besides including recursion and the above ``timed'' CSP parallel composition and hiding operators, is endowed with a timed variant, taken from~\cite{HenR}, of the prefix and choice operators: a ``timed prefix'' $a^t.P$, which allows time to pass via $\delta$ actions, while waiting for the action $a$ to be performed; and a ``timed choice'' $P+^t Q$, which, similarly, allows time to pass via $\delta$ actions performed (in synchronization) by $P$ and $Q$, while waiting for the choice to be resolved by a standard action. Considering such a timed variant of the operators of the basic calculus is convenient from a modeling viewpoint and leads to {\it time-determinism}: any process of the discrete time calculus can perform at most one $\delta$ transition. Such a property is natural in discrete time (see~\cite{HenR}) in that the mere passage of time is expected to modify system behavior in a deterministic way.

Even if we consider the discrete time calculus as a ``specification level calculus'' that does not include the standard prefix and choice operators of the basic calculus of Section~\ref{SectAPFSB}, the latter, also due to presence of recursion, plays a fundamental role in its axiomatization.
The idea is that, since all the operators of the full discrete time calculus have a sort of static behaviour (even $a^t.P$ and $P+^t Q$), we can use 
terms of the basic calculus of Section~\ref{SectAPFSB} (which essentially can be seen as a process algebraic representation of finite-state labeled transition systems) to express ``normal forms'' of processes. The 
complete axiomatization, developed in Section~\ref{SectAPFSB}, enforced over normal forms, would then yield completeness for  
finite-state processes of the discrete time calculus.

As a matter of fact, the introduction of the ``timed choice'' $P+^t Q$ operator, makes standard observational congruence,
with $\delta$ actions being treated as standard visible actions, no longer a congruence. Similarly as in~\cite{CLM} we need to extend the root condition 
of the equivalence (the first step, leading to weak bisimulation), so that the (matched) execution of a $\delta$ action does not lead to leaving the root condition, i.e.\ to weak bisimulation. So the modified notion of observational congruence still coincides
with the standard one as far as standard actions are concerned (it is a conservative extension of it). Moreover,  for any action, it still
leads to standard weak bisimulation once the (strengthened) root condition is left. 
As a consequence, we have to slightly modify the axiomatization of the basic calculus considered in Section~\ref{SectAPFSB},
in order to account for the modified way in which the discrete time observational congruence deals with $\delta$ actions:
this just causes an additional $\delta$-specific axiom to be added to standard $Tau$ axioms and does not affect the $Ung$ axioms (the ones we dealt with in Section~\ref{SectAPFSB} due to priority of $\tau$ over $\delta$).

Apart from axiomatizing the basic calculus, one of the technical difficulties concerning the (inductive) procedure of turning discrete time calculus processes into normal form is dealing with unguardedness generated by static operators, i.e., in our case by the hiding operator. During a two months visit of Prof.\ Jos Baeten (who was on sabbatical leave) at University of Bologna, a procedure of this kind was devised in the context of a standard (untimed) process algebra, based on the introduction of a specific axiom. Such a procedure, that is here reapplied in a much more complex setting, is contained in~\cite{concur05,mscs08}.
The completeness result of the axiomatization for the discrete time calculus is then obtained by adopting a syntactical characterization that guarantees finite-stateness, which is similar, apart from treatment of the $P+^t Q$ operator, to that of~\cite{tocl,concur05,mscs08}. Notice that such a characterization includes the possibility of expressing unguarded recursion, which is dealt with by resorting to the basic calculus axiomatization in Section~\ref{SectAPFSB}.

Finally, in Section~\ref{SectRW} we discuss related work and in Section~\ref{SectConc} we report remarks about future work.

Section~\ref{SectAPFSB} is an extended and revised version of~\cite{icalp00} that presents, for the first time,
complete proofs. Presenting such a technical machinery in its entirety is needed for Section~\ref{SectDRT}
that builds its results and proofs on it.

\section{Axiomatizing Prioritized Finite State Behaviors}
\label{SectAPFSB}

In this section we present the complete axiomatization of observational congruence for the
basic calculus: the algebra of
finite-state agents 
used by Milner in~\cite{Mil89} extended with generic $\delta$ prefixing, 
where $\delta$ actions have lower priority than internal $\tau$ actions.


\subsection{A Basic Calculus}
\label{SSectBC}

Prioritized observable actions are denoted by $a,b,c, \dots$.
The denumerable set of all prioritized actions, which includes the silent action 
$\tau$ denoting an internal move, is denoted by $\ms{PAct}$, 
ranged over by $\alpha, \alpha', \dots$ . 
The set of all actions is defined by $\ms{Act} =
\ms{PAct} \cup \{ \delta \}$, ranged over by $\gamma, \gamma', \dots$ .
The denumerable set of term variables is $\ms{Var}$, ranged over
by $X,Y, \dots$ . 
The set $\cale$ of behavior expressions, ranged over by $E, F, G, \dots$ is 
defined by the following syntax.
$$E ::= \nil \mid X \mid \gamma . E \mid E + E \mid \ms{rec}X . E$$
The meaning of the operators is the standard one of~\cite{Mil,Mil89}, where
$\ms{rec}X$ denotes recursion in the usual way. 
We adopt the standard definitions of~\cite{Mil,Mil89} for {\it free} variable,
and {\it open} and {\it closed} term. 
The set of processes, i.e. closed terms, is denoted by $\calp$,
ranged over by $P, Q, R, \dots$.

As in~\cite{Mil89} we take the liberty of identifying expressions which differ only by 
a change of bound variables (hence we do not need to deal with
$\alpha$-conversion explicitly). 
We will write $E \{ F / X \}$ for the result of syntactically substituting
$F$ for each free occurrence of $X$ in $E$, renaming bound variables as necessary.
This is also generalized to sets of variables $\tilde{X} = \{ X_1, \dots , X_n \}$:
$E \{ \tilde{F} / \tilde{X} \}$, where $\tilde{F} = \{ F_1, \dots , F_n \}$ stands for the result of substituting
$F_i$ for each free occurrence of $X_i$ in $E$ for all $i \in \{1, \dots, n\}$, renaming bound variables as necessary.

We adopt the following standard (see~\cite{Mil89}) definitions concerning guardedness of
variables. 

\begin{defi}
A free occurrence of $X$ in $E$ is {\it weakly guarded} if it occurs
within some subexpression of $E$ of the form $\gamma.F$.
It is {\it (strongly) guarded} if we additionally have that
$\gamma \neq \tau$. It is {\it unguarded} if it is not (strongly) guarded.
It is {\it fully unguarded} if it is not weakly guarded.
\end{defi}

We say that $X$ is weakly guarded/(strongly) guarded in $E$ if 
each free occurrence of $X$ in $E$ is weakly guarded/(strongly) guarded, respectively.
Correspondingly, we say that $X$ is unguarded/fully unguarded in $E$ if 
some free occurrence of $X$ in $E$ is unguarded/fully guarded, respectively.
Moreover we say that a recursion $\ms{rec} X . E$ is weakly guarded/(strongly) guarded/unguarded/ fully unguarded, if $X$ is weakly guarded/(strongly) guarded/unguarded/fully unguarded in $E$, respectively.
Finally, we say that an expression $E$ is weakly guarded/(strongly) guarded if
every subexpression of $E$ which is a recursion is weakly guarded/(strongly)
guarded, respectively.
Correspondingly, we say that an expression $E$ is unguarded/fully unguarded if 
some subexpression of $E$ which is a recursion 
is unguarded/fully guarded, respectively.

The operational semantics of the algebra terms is given as a relation
$\arrow{} \subseteq \calp \times  \ms{Act} \times \calp$. We write $P 
\arrow{\gamma} Q$ for $(P,\gamma,Q) \in \arrow{}$, $P \arrow{\gamma}$ for 
$\exists Q 
: (P,\gamma,Q) \in \arrow{}$ and $P \notarrow{\gamma}$ for $\notexists Q
: (P,\gamma,Q) \in \arrow{}$. $\arrow{}$ is defined as the least relation
satisfying the operational rules in Tables~\ref{StRules} and~\ref{DeltaRules}.
Notice that, even if the rules in Table~\ref{DeltaRules} include a negative
premise, the operational semantics is well-defined in that
the inference of transitions can be stratified (see, e.g., \cite{Gro}).

        {\begin{table}[t]

{ \[
\begin{array}{|c|} 
\hline \\
\begin{array}{l}
\gamma. P \arrow{\gamma} P \\[0.3cm]
\end{array} \\
\begin{array}{cccc}
\infr{P \arrow{\alpha} P'}{P + Q \arrow{\alpha} P'} & \hspace{1cm} & 
\infr{Q \arrow{\alpha} Q'}{P + Q \arrow{\alpha} Q'} & \\[0.6cm]
\end{array} \\
\begin{array}{l}
\infr{E\{ \ms{rec}X . E / X \} \arrow{\gamma} P'} 
{\ms{rec}X . E \arrow{\gamma} P'  } \\[0.7cm]
\end{array} \\
\hline
\end{array}\]}

\caption{Standard Rules}\label{StRules}

        \end{table}}

        {\begin{table}[t]

{ \[
\begin{array}{|c|} 
\hline \\
\begin{array}{cccc}
\infr{P \arrow{\delta} P' \hspace{.5cm} Q \notarrow{\tau} }
{P + Q \arrow{\delta} P'} & &
\infr{Q \arrow{\delta} Q' \hspace{.5cm} P \notarrow{\tau} }
{P + Q \arrow{\delta} Q'} & \\[0.6cm]
\end{array} \\
\hline
\end{array}\]}

\caption{Special Rules Expressing Priority.}\label{DeltaRules}

        \end{table}}

As in~\cite{NCCC} we capture the priority of $\tau$ 
actions over $\delta$ actions by cutting transitions which cannot be 
performed directly in semantic models (and not by discarding them
at the level of bisimulation definition as done in~\cite{HL}) so
that we can just apply the ordinary notion of observational 
congruence~\cite{Mil}. 

\begin{prop}[maximal progress]\label{maxprog}
%
Relation $\arrow{}$ is such that: \\[.2cm]
\centerline{
$P \arrow{\tau} \hspace{.1cm}$ implies
$P \narrow{\delta}$}
%
\end{prop}
\begin{proof}
Easily shown by induction on the height
of the inference tree of $\tau$ transitions of terms
$P \in \calp$ (the base case being a $\tau$ prefix).
\end{proof}

\subsection{Weak Bisimulation Equivalence}
\label{SSectWBE}

The equivalence notion we consider over the terms of our prioritized
process algebra is the standard notion of 
observational congruence extended to open terms~\cite{Mil,Mil89}.

As in~\cite{Mil89}
we use $\sarrow{\gamma}\!\!^+$ to denote computations composed of all 
$\sarrow{\gamma}$ transitions whose length is at least one 
and $\sarrow{\gamma}\!\!^*$ to denote computations composed of all 
$\sarrow{\gamma}$ transitions whose length is possibly zero.
Let $\warrow{\gamma}$ denote $\sarrow{\tau}\!\!^*
\sarrow{\gamma} \sarrow{\tau}\!\!^*$. Moreover we define 
$\warrow{\hat{\gamma}} = \warrow{\gamma}$ if $\gamma \neq \tau$ and
$\warrow{\hat{\tau}} = \sarrow{\tau}\!\!^*$. 

\begin{defi}
A relation $\beta \subseteq \calb \times \calb$ is a weak bisimulation if,
whenever 
$(P,Q) \in \beta$:

\begin{itemize}\label{basicweakbis}
\item If $P \arrow{\gamma} P'$ then, for some $Q'$, $Q \warrow{\hat{\gamma}} Q'$
and $(P',Q') \in \beta$.
\item If $Q \arrow{\gamma} Q'$ then, for some $P'$, $P \warrow{\hat{\gamma}} P'$
and $(P',Q') \in \beta$.
\end{itemize}
Two processes $P$, $Q$ are weakly bisimilar,
written $P \approx Q$, iff $(P,Q)$ is included in some weak bisimulation.

\end{defi}

\begin{defi}\label{basicobscongr}
Two processes $P$, $Q$ are observationally congruent, written 
$P \simeq Q$, 
iff:

\begin{itemize}
\item If $P \arrow{\gamma} P'$ then, for some $Q'$, $Q \warrow{\gamma} Q'$
and $P' \approx Q'$.
\item If $Q \arrow{\gamma} Q'$ then, for some $P'$, $P \warrow{\gamma} P'$
and $P' \approx Q'$.
\end{itemize}
\end{defi}
Weak bisimulation $\approx$ and observational congruence $\simeq$ are indeed equivalence relations~\cite{Mil}.

Open terms are dealt with as follows~\cite{Mil}.

\begin{defi}\label{obscongruence}
Two open terms $E$, $F$ are observationally congruent, written 
$E \simeq F$, if, assumed the set of variables $\tilde{X}$ to include all free variables occurring
in $E$ and $F$, the following holds. For all sets $\tilde{P}$ of closed terms we have
$E \{ \tilde{P} / \tilde{X} \} \simeq F \{ \tilde{P} / \tilde{X} \}$.
\end{defi}


\begin{cor}\label{cor}
If $P \simeq Q$ then:
\cws{10}{30}{P \arrow{\tau} \hspace{.5cm} \Leftrightarrow \hspace{.5cm} 
Q \arrow{\tau}}
\end{cor}

The following theorem shows that the presence of priority preserves 
the congruence property of observational congruence w.r.t.\ the 
operators of the algebra.

\begin{thm}
$\simeq$ is a congruence w.r.t. prefix, choice and recursion operators.
\end{thm}

\begin{proof}
As far as the prefix operator is concerned, from $P \simeq Q$ we immediately derive
$\gamma.P \simeq \gamma.Q$

Concerning the choice operator, from $P \simeq Q$ we derive 
$P + R \simeq Q + R$ as follows. Suppose $P + R \arrow{\gamma} P'$,
we have two cases.
\begin{itemize}

\item If $P \arrow{\gamma} P'$ (and $R \notarrow{\tau}$) then (from $P \simeq Q$) 
$Q \warrow{\gamma} Q'$ for some $P' \approx Q'$.
Hence we have $Q + R \warrow{\gamma} Q'$. In particular
if $\gamma = \delta$ this derives from the fact that
$R \notarrow{\tau}$.

\item If $R \arrow{\gamma} P'$ (and $P \notarrow{\tau}$)
then $Q + R \arrow{\gamma} P'$. In particular 
if $\gamma = \delta$ this derives from the following consideration.
From $P \notarrow{\tau}$  we have also that $Q \notarrow{\tau}$
(by Corollary~\ref{cor}).
\end{itemize}

Concerning the recursion operator, from $E \simeq F$ 
we derive 
$\ms{rec}X.E \simeq \ms{rec}X.F$ as follows. We show that
\cws{10}{10}{
\beta = \{ (G \{ \ms{rec}X.E / X \}, G \{ \ms{rec}X.F / X \}) \mid 
G \; {\rm contains \; at \; most} \; X \; {\rm free} \} }
satisfies the condition: \\[.2cm]
\hspace*{1cm} if $G \{ \ms{rec}X.E / X \} \arrow{\gamma} H$ then, 
for some $H',H''$, \\[.1cm]
\hspace*{3cm}
$G \{ \ms{rec}X.F / X \} \warrow{\gamma} H''$ with $H'' \approx H'$ such that $(H,H') \in \beta$, \\[.1cm]
\hspace*{1cm} and symmetrically for a move of $G \{ \ms{rec}X.F / X \}$. \\[.2cm]
This implies that $\beta$ is a weak bisimulation up to $\approx$, see \cite{SM}. 
Moreover
by taking $G \equiv X$ we may conclude that $\ms{rec}X.E \simeq \ms{rec}X.F$.

In the following we begin the proof that $(G \{ \ms{rec}X.E / X \}, 
G \{ \ms{rec}X.F / X \})
\in \beta$ satisfies the condition above by inducing on the height of the
inference tree by which transitions of $G \{ \ms{rec}X.E / X \} \arrow{\alpha} H$ 
are inferred. That is, first we consider standard $\alpha$ transitions (which are inferred from
standard transitions only) and then $\delta$ transitions.

We have the following cases depending on
the structure of $G$.

\begin{itemize}
\item If $G \equiv \nil$ then the condition above trivially holds.
\item If $G \equiv X$ then $G \{ \ms{rec}X.E / X \} \equiv \ms{rec}X.E$ and
$G \{ \ms{rec}X.F / X \} \equiv \ms{rec}X.F$. \\
$\ms{rec}X.E \arrow{\alpha} H$ implies $E\{ \ms{rec}X.E / X \} \arrow{\alpha} H$. 
Hence, by induction, 
for some $H',H''$, 
$E \{ \ms{rec}X.F / X \} \warrow{\alpha} H''$ with $H'' \approx H'$ such that $(H,H') \in \beta$.
Since $E \simeq F$, we have also that  
$F \{ \ms{rec}X.F / X \} \warrow{\alpha} H'''$ with $H''' \approx H''$. Therefore 
$\ms{rec}X.F \warrow{\alpha} 
H'''$ with $H''' \approx H'$ such that $(H,H') \in \beta$. 

\item If $G \equiv \alpha.G'$ then $G \{ \ms{rec}X.E / X \} \equiv 
\alpha.(G' \{ \ms{rec}X.E / X \})$ and $G \{ \ms{rec}X.F / X \} 
\equiv 
 \alpha.(G' \{ \ms{rec}X.F \linebreak
/ X \})$. The result trivially follows.

\item If $G \equiv G' + G''$ then $G \{ \ms{rec}X.E / X \} \equiv 
G' \{ \ms{rec}X.E / X \} + G'' \{ \ms{rec}X.E / X \}$ and
$G \{ \ms{rec}X.F / X \}  \linebreak
\equiv G' \{ \ms{rec}X.F / X \} + 
G'' \{ \ms{rec}X.F / X \}$. \\
Suppose $G' \{ \ms{rec}X.E / X \} + G'' \{ \ms{rec}X.E / X \} 
\arrow{\alpha} H$, we have two cases.
\begin{itemize}
\item If $G' \{ \ms{rec}X.E / X \} \arrow{\alpha} H$ then 
(by induction) $G' \{ \ms{rec}X.F / X \} \warrow{\alpha}$ 
$H''$ with
$H'' \approx H'$ and $(H,H') \in \beta$. Therefore $G' \{ \ms{rec}X.F / X \} 
+$ 
$G'' \{ \ms{rec}X.F / X \} \warrow{\alpha} H''$ with
$H'' \approx H'$ and $(H,H') \in \beta$. 
\item If $G'' \{ \ms{rec}X.E / X \} \arrow{\alpha} H$ then
the result is derived in a similar way.
\end{itemize}

\item If $G \equiv \ms{rec}Y.G'$, with $Y \neq X$, then $G \{ \ms{rec}X.E / X \} \equiv
\ms{rec}Y. G' \{ \ms{rec}X.E / X \}$ and 
$G \{ \ms{rec}X.F$ 
$/ X \} \equiv \ms{rec}Y. G' 
\{ \ms{rec}X.F / X \}$. \\
$\ms{rec}Y. G' \{ \ms{rec}X.E / X \} \arrow{\alpha} H$ implies 
$(G' \{ \ms{rec}X.E / X \}) \{ \ms{rec}Y. G' \{ \ms{rec}X.E / X \} 
/ Y \} \equiv \\
(G' \{ {rec}Y. G' / Y \}) \{ {rec}X.E / X \}
\arrow{\alpha} H$.
By induction we have 
$(G' \{ {rec}Y. G' / Y \})
\{ \ms{rec}X.F / X \} \linebreak
\equiv \! (G' \{ \ms{rec}X.F / X \}) \{ \ms{rec}Y. G' \{ \ms{rec}X.F / X \} / Y \}
\warrow{\alpha} H''$ with
$H'' \approx H'$ and $(H,H')  \in \beta$. 
Therefore $\ms{rec}Y. G' \{ \ms{rec}X.F / X \} \warrow{\alpha}$ 
$H''$ with
$H'' \approx H'$ and $(H,H')  \in \beta$.

\end{itemize}
A completely symmetric inductive proof is performed when we start from a transition of $G \{ \ms{rec}X.F / X \}
\arrow{\alpha} H$ in the condition above.

An identical inductive proof is performed when we start from a transition of $G \{ \ms{rec}X.E / X \}$
$\arrow{\delta} H$ in the condition above (we just have to consider $\delta$ instead of a generic $\alpha$), 
apart when $G \equiv G' + G''$. In this case, to derive the final statement $G' \{ \ms{rec}X.F / X \} + G'' \{ \ms{rec}X.F / X \} \warrow{\delta}$ \linebreak
$H''$ with $H'' \approx H'$ and $(H,H') \in \beta$ we have
to additionally prove that the $\delta$ transition is not pre-empted. In the first subcase this is
obtained as follows. It must be that 
$G'' \{ \ms{rec}X.E / X \} \notarrow{\tau}$, hence, 
since we already showed the condition for relation $\beta$ presented above (and its symmetric one) to hold for any $\alpha$ transition,
we have also that 
$G'' \{ \ms{rec}X.F / X \} \notarrow{\tau}$. The other subcase is analogous.
Finally, a completely symmetric inductive proof is performed when we start from a transition of $G \{ \ms{rec}X.F / X \} \arrow{\delta} H$ in the condition above.
\end{proof}

In the following we show that, by exploiting the priority constraint, we
can rewrite weak bisimulation definition into a less generic form.
In particular we present two reformulations of weak bisimulation
(the second one exploiting guardedness) and a reformulation
of observational congruence.

\begin{prop}\label{reformulation}
A relation $\beta \subseteq \calp \times \calp$ is a weak bisimulation iff,
whenever 
$(P,Q) \in \beta$:

\begin{itemize}
\item If $P \arrow{\alpha} P'$ then, for some $Q'$, $Q \warrow{\hat{\alpha}} Q'$
and $(P',Q') \in \beta$.
\item 
There exists $Q''$ such that $Q \arrow{\tau}\!\!^* \; Q''$ and:
\\[.1cm]
\hspace*{1cm} 
if $P \arrow{\delta} P'$ then, for some $Q'$, 
$Q'' \arrow{\delta}\!\! \arrow{\tau}\!\!^* \;Q'$
and $(P',Q') \in \beta$.
\item If $Q \arrow{\alpha} Q'$ then, for some $P'$, $P \warrow{\hat{\alpha}} P'$
and $(P',Q') \in \beta$.
\item 
There exists $P''$ such that $P \arrow{\tau}\!\!^* \; P''$ and:
\\[.1cm]
\hspace*{1cm} 
if $Q \arrow{\delta} Q'$ then, for some $P'$, 
$P'' \arrow{\delta}\!\! \arrow{\tau}\!\!^* \;P'$
and $(P',Q') \in \beta$.
\end{itemize}
\end{prop}
\begin{proof}
If $\beta$ satisfies such a stricter condition then obviously $\beta$ is a weak bisimulation.
Conversely, suppose that $\beta$ is a weak bisimulation; we now prove that it satisfies the condition above. If $P \arrow{\delta}$ then
it must be that $Q \arrow{\tau}\!\!^* Q'' \arrow{\delta}
\arrow{\tau}\!\!^*$ for some $Q''$. We now prove that such a $Q''$ satisfies the related constraint in the condition above.
Since $Q \arrow{\tau}\!\!^* Q''$, it must be $(P, Q'') \in \beta$: because of the priority constraints, since $P \arrow{\delta}$, the sequence of $\tau$ transitions of $Q$ can only by matched by zero length moves of $P$.
From $(P, Q'') \in \beta$ the constraint above is directly
derived by observing that, since $Q'' \arrow{\delta}$, $Q''$ cannot perform $\tau$ transitions.
Symmetrically for $Q \arrow{\delta}$.
\end{proof}

%
%


\begin{prop}\label{necessary}
If a relation $\beta \subseteq \calp \times \calp$ is a weak bisimulation then,
whenever 
$(P,Q) \in \beta$:

\begin{itemize}
\item If $P \arrow{\alpha} P'$ then, for some $Q'$, $Q \warrow{\hat{\alpha}} Q'$
and $(P',Q') \in \beta$.
\item If $P \arrow{\delta} P'$ then, either $Q \arrow{\tau}$, or,
for some $Q'$, $Q \arrow{\delta}\!\! \arrow{\tau}\!\!^* \;Q'$
and $(P',Q') \in \beta$.
\item If $Q \arrow{\alpha} Q'$ then, for some $P'$, $P \warrow{\hat{\alpha}} P'$
and $(P',Q') \in \beta$.
\item If $Q \arrow{\delta} Q'$ then, either $P \arrow{\tau}$, or,
for some $P'$, $P \arrow{\delta}\!\! \arrow{\tau}\!\!^* \;P'$
and $(P',Q') \in \beta$.
\end{itemize}
If, in addition, $\beta$ is a relation over guarded processes of $\calp$, then
$\beta$ is a weak bisimulation iff it satisfies the condition above.

\end{prop}

\begin{proof}
We start by proving the first statement of the proposition.
Suppose $P \arrow{\delta} P'$, it is just sufficient to observe that: if
$Q \notarrow{\tau}$ then it must be $Q'' \equiv Q$ in the condition of
Proposition~\ref{reformulation}. Symmetrically for $Q \arrow{\delta} Q'$.

Concerning the second statement, if $\beta$ is a weak bisimulation then the 
condition above is satisfied as for the first statement.
Conversely, suppose that $\beta$ satisfies the condition above; 
we now prove that it satisfies the condition of Proposition~\ref{reformulation}. 
If $P \arrow{\delta}$ then either $Q \notarrow{\tau}$ and we are done
(in this case $Q'' \equiv Q$ in the condition of
Proposition~\ref{reformulation} hence the two conditions coincide), 
or $Q \arrow{\tau} Q'$ for some $Q'$. In the latter case it must be $(P,Q') \in
\beta$ (because of the priority constraints, $P$ can only match the $\tau$ move of $Q'$ by a zero length move). Then the same argument is repeated for the pair
$(P,Q')$ and so on\dots until a $\arrow{\tau}\!\!^* Q^n$ derivative of $Q$ is reached 
such that $Q^n \notarrow{\tau}$ and we are done. The guardedness constraint guarantees that such a derivative is always reached: the labeled transition system generated by a guarded process $\calp$ is such that every $\tau$ path eventually reaches a state with no outgoing $\tau$ transitions (since it is finite-state and no $\tau$ loops can be generated with just strongly guarded recursion).
Symmetrically for $Q \arrow{\delta}$.
\end{proof}

A simple counterexample that shows why the guardedness constraint for processes
is needed to obtain the ``iff'' in Proposition~\ref{necessary} is the following: 
$\beta = \{rec X. \tau . X, \delta.\nil\}$ satisfies the condition of the theorem but obviously 
it is not a 
weak bisimulation.

\begin{prop}\label{refcongr}
Two processes $P$, $Q$ are observationally congruent ($P \simeq Q$), 
iff:

\begin{itemize}
\item If $P \arrow{\alpha} P'$ then, for some $Q'$, $Q \warrow{\alpha} Q'$
and $P' \approx Q'$.
\item If $P \arrow{\delta} P'$ then, for some $Q'$, 
$Q \arrow{\delta}\!\! \arrow{\tau}\!\!^* \;Q'$
and $P' \approx Q'$.
\item If $Q \arrow{\alpha} Q'$ then, for some $P'$, $P \warrow{\alpha} P'$
and $P' \approx Q'$.
\item If $Q \arrow{\delta} Q'$ then, for some $P'$, 
$P \arrow{\delta}\!\! \arrow{\tau}\!\!^* \;P'$
and $P' \approx Q'$.
\end{itemize}
\end{prop}
\begin{proof}
If $P,Q$ satisfy such a stricter condition then obviously $P \simeq
Q$.
Conversely, suppose that $P \simeq
Q$; we now prove that they satisfy the condition above.
Suppose $P \arrow{\delta} P'$, due to the priority constraints,
we must have that $Q \notarrow{\tau}$ (otherwise, by contradiction,
being $P \simeq Q$, also $P$ should perfrom a $\tau$ move), therefore
the condition above directly derives from the definition of observational congruence.
Symmetrically for $Q \arrow{\delta} Q'$.
\end{proof}

\subsection{Axiomatization}
\label{SSectAxiomBC}

We now present an axiomatization of $\simeq$ which is complete
over processes $P \in \calp$ of our algebra.

As we already explained in the introduction, the  
law of ordinary CCS which makes it possible escape $\tau$ divergence:
$$\ms{rec} X. (\tau.X + E) = \ms{rec} X. \tau. E$$
is not sound in a calculus with this kind of priority. 
Consider for instance the divergent term
$F \equiv \ms{rec} X. (\tau.X + \delta . \nil)$. 
Because of priority of ``$\tau$''
actions over ``$\delta$'' actions the operational semantics of  
$F$ is isomorphic to that of $\ms{rec} X. \tau.X$. Hence $F$ is an 
infinitely looping term which can never escape from divergence by executing 
the action ``$\delta$''. If the cited law were sound, 
we would obtain $F = \ms{rec} X. \tau. \delta . \nil$ and this is 
certainly not the case. 

%
        {\begin{table}[t]

{ \[
\begin{array}{|c|} 
\hline 
\begin{array}{c}
\hspace {.5cm} \infr{P \arrow{\alpha} P'}{\ms{pri}(P) \arrow{\alpha} P' } \hspace {.5cm} \\[0.6cm]
\end{array} \\
\hline
\end{array}\]}

\caption{Rule for Auxiliary Pri Operator.}\label{PriRule}

        \end{table}}

In general the behavior of $E'$ such that $\ms{rec} X. (\tau.X + E)=\ms{rec} X. \tau. E' $ is obtained from that of $E$ by removing all 
``$\delta$'' actions (and subsequent behaviors) performable in $E$.  
We denote such $E'$, representing the ``prioritized'' behavior of $E$, 
with $\ms{pri}(E)$. The operational semantics of the auxiliary
operator $\ms{pri}$ is simply that in Table \ref{PriRule}. 
As we discussed in the introduction, this auxiliary operator is crucial for being able
to axiomatize the priority of $\tau$ actions over $\delta$ actions in that it allows us to represent
a {\it scope} to which priority is applied.


%
The axiomatization of ``$\simeq$'' we propose is made over the set of terms
$\cale_{pri}$, generated by extending the syntax
to include the new operator $\ms{pri}(E)$. 

We start by noting that the congruence property of ``$\simeq$'' trivially extends to the new 
operator.
			
\begin{thm}
$\simeq$ is a congruence w.r.t.\ the new operator $\ms{pri}(E)$.
\end{thm}
\begin{proof}
Let us suppose $P \simeq Q$.
If $\ms{pri}(P) \arrow{\gamma} P'$ then, according to the semantics of Table~\ref{PriRule}, it must be 
$\gamma \neq \delta$ and $P \arrow{\gamma} P'$. Hence $Q \warrow{\gamma} Q'$
with $P' \approx Q'$. Thus, being $\gamma \neq \delta$, $\ms{pri}(Q) \warrow{\gamma} Q'$.
Symmetrically for transitions of $\ms{pri}(Q)$.
\end{proof}

We adopt the following notion of {\it serial} variable, which is used
in the axiomatization.

\begin{defi}
$X$ is {\it serial} in $E \in \cale_{pri}$ if every subexpression of $E$ 
which contains a free occurrence of $X$, apart from $X$ itself, is of the form $\gamma . F$,
$F' + F''$ or $\ms{rec}Y.F$ for some variable $Y$. 

\end{defi}

        { \begin{table}[t]

{\[\begin{array}{|l|}
\hline
\begin{array}{lrcl}
(A1) \;\;\;\;\, & E + F & = & F + E \hspace{2.3cm} \\[0cm]
(A2) & (E + F) + G & = & E + (F + G) \\[0cm]
(A3) & E + E & = & E \\[0cm]
(A4) & E + \nil & = & E \\[.0cm]
\end{array} \\ 
\hline
\begin{array}{lrcl}
(Tau1) \; & \gamma . \tau . E & = & \gamma . E \\[0cm]
(Tau2) & E + \tau . E & = & \tau . E \\[0cm]
(Tau3) & \gamma . ( E + \tau . F ) + \gamma . F & = & \gamma . ( E + \tau . F )
\\[.0cm]
\end{array} \\ 
\hline
\begin{array}{lrcl}
(Pri1) \; & \ms{pri}(\nil) & = & \nil \\[0cm]
(Pri2) & \ms{pri}(\alpha . E) & = & \alpha . E \\[0cm]
(Pri3) & \ms{pri}(\delta . E) & = & \nil \\[0cm]
(Pri4) & \ms{pri}(E + F) & = & \ms{pri}(E) + \ms{pri}(F) \\[0cm]
(Pri5) & \ms{pri}(\ms{pri}(E)) & = & \ms{pri}(E) \\[0cm]
(Pri6) & \tau . E + F & = & \tau . E + \ms{pri}(F) \\[.0cm]
\end{array} \\
\hline
\begin{array}{lrcll}
(Rec1) & \ms{rec}X . E & = & E \{ \ms{rec}X . E / X \} & \\[0cm]
(Rec2) & F = E \{ F / X \} & \Rightarrow & F = \ms{rec}X . E &
\; \rm{provided\;that}\;X\;\rm{is\;serial\;and} \\[0cm]
& & & & \tabspace{\; \rm{provided}\,} \; \rm{strongly\;guarded\;in}\; E 
\\[.0cm]
\end{array} \\ 
\hline
\begin{array}{lrcl}
(Ung1) & \ms{rec}X . (X + E) & = & \ms{rec}X . E \\[0cm]
(Ung2) & \ms{rec}X . (\tau . X + E) & = & \ms{rec}X . 
( \tau . \ms{pri} (E) ) \\[0cm]
(Ung3) & \ms{rec}X . (\tau . (X + E) + F) & = & \ms{rec}X . (\tau . X
+ E + F) \\[0cm]
(Ung4) & \ms{rec}X . (\tau . (\ms{pri}(X) + E) + F) & = & \ms{rec}X . (\tau . X
+ E + F) \\[.0cm]
\end{array} \\
\hline
\end{array}\]}

\caption{Axiom system $\cala$}\label{Axioms}

        \end{table}}

The axiom system $\cala$ is formed by the axioms presented in 
Table~\ref{Axioms}.
The axiom $(Pri6)$ expresses the priority of $\tau$ actions over 
$\delta$ actions.
Note that from $(Pri6)$ we can derive $\tau . E + \delta . E = \tau . E$
by applying $(Pri3)$.
The axioms $(Rec1)$, $(Rec2)$ handle
strongly guarded recursion in the standard way~\cite{Mil89}.
The axioms $(Ung1)$ and $(Ung2)$ are used to turn unguarded terms into 
strongly guarded ones similarly as in~\cite{Mil89}. The axiom
$(Ung3)$ and the new axiom $(Ung4)$ are used to transform weakly
guarded recursions into the form required by the axiom $(Ung2)$, so
that they can be turned into strongly guarded ones. In particular the
role of axiom $(Ung4)$ is to remove unnecessary occurrences of 
terms $\ms{pri}(X)$ in weakly guarded recursions.
In the following, when transforming terms via axiom applications, we will often omit mentioning 
axioms  $(A1) - (A4)$, which are just used to rearrange arguments and eliminate duplicates and $\nil$ in sums.

\subsection{Soundness}

We start by showing that the axiom system $\cala$ is sound.

\begin{thm}~\label{soundness}
Given $E, F \in \cale_{pri}$, if $\cala \vdash E = F$ then $E \simeq F$.
\end{thm}

\begin{proof}
We prove the soundness of the new laws $(Ung2)$ and $(Ung4)$. 
For the other laws $(Ung1)$ and $(Ung3)$ the proof is a similar 
adaptation of the standard one. 
The soundness of the laws $(Rec1)$ and $(Rec2)$ is shown 
as in~\cite{Mil} (version corrected according to~\cite{SM} concerning the structure of the weak bisimulation up to $\approx$ to be considered) by assuming transition labels to also encompass the $\delta$ action. 


The soundness of the new equations $(Pri1)-(Pri6)$ easily derives
from the fact that the semantic models of left-hand and right-hand 
terms are isomorphic. 

\vspace{.5cm}


\noindent {\bf Soundness of $(Ung2)$}$\;$
In order to prove the soundness of $(Ung2)$ we show that: \\[.2cm]
\hspace*{.5cm} $\beta = \{ \, (G \{ \ms{rec}X.(\tau.X + E) / X \}, G \{ \ms{rec}X.
( \tau .  \ms{pri} (E) ) / X \}) \mid 
G \; {\rm contains \; at}$ \\[.1cm]
\hspace*{.5cm} $\tabspace{\beta =,} 
{\rm most} \; X \; {\rm free} \, \} \, \cup
\{ \,
(
\ms{rec}X.(\tau.X + E) , 
\ms{pri} (E) \{ \ms{rec}X. ( \tau .  \ms{pri} (E) ) / X \}
) 
\, \}$ \\[.2cm]
is a weak bisimulation. From this result it straightforwardly follows
that $\ms{rec}X.(\tau.X + E) \simeq \ms{rec}X. ( \tau .  \ms{pri} (E) )$.

We start the proof that $\beta$ is a weak bisimulation by considering all the pairs
$(G \{ \ms{rec}X.(\tau.X + E) / X \}, G \{ \ms{rec}X. ( \tau .  \ms{pri} (E) ) 
/ X \})$ such that $G$ contains at most $X$ free. In particular we prove
that $\beta$ satisfies the following stronger condition \\[.2cm]
\hspace*{1cm} if $G \{ \ms{rec}X.(\tau.X + E) / X \} \arrow{\gamma} H$ then, 
for some $H'$, \\[.1cm]
\hspace*{5cm}
$G \{ \ms{rec}X.( \tau . \ms{pri} (E) ) / X \} \warrow{\gamma} H'$ with $(H,H') \in 
\beta$, \\[.1cm]
\hspace*{1cm} and symmetrically for a move of $G \{ \ms{rec}X.
( \tau . \ms{pri} (E) ) / X \}$ \\
We begin proving it by induction on the depth of the
inference by which transitions $G \{ \ms{rec}X.(\tau.X + E) / X \} \arrow{\alpha} H$ 
are inferred. 
We have the following cases depending on
the structure of $G$.

\begin{itemize}
\item If $G \equiv \nil$, then the condition above obviously holds.
\item If $G \equiv X$, then $G \{ \ms{rec}X.(\tau.X + E) / X \} \equiv 
\ms{rec}X.(\tau.X + E)$ and $G \{ \ms{rec}X. ( \tau .$ 
$  \ms{pri} (E) ) 
/ X \} \equiv \ms{rec}X. ( \tau .  \ms{pri} (E) )$. \\
Suppose $\ms{rec}X.(\tau.X + E) \arrow{\alpha} H$, we have two cases:
\begin{itemize}
\item If $\ms{rec}X.(\tau.X + E) \arrow{\tau} H \equiv \ms{rec}X.(\tau.X + E)$, then
$\ms{rec}X. ( \tau . \ms{pri} (E) ) \arrow{\tau} H' \equiv \ms{pri} (E)$ \linebreak
$\{ \ms{rec}X. ( \tau .  \ms{pri} (E) )
/ X ) \}$ with $(H,H')$.
\item If $\ms{rec}X.(\tau.X + E) \arrow{\alpha} H$ with
$E \{ \ms{rec}X.(\tau.X + E) / X \} \arrow{\alpha} H$, then, by induction,
we derive that, for some $H'$, 
$E \{ \ms{rec}X.( \tau . \ms{pri} (E) ) / X \} \warrow{\alpha} H'$ with $(H,H') \in 
\beta$. Hence, $\ms{rec}X.( \tau . \ms{pri} (E) ) \arrow{\tau} \warrow{\alpha}
H'$.
\end{itemize}

\item If $G \equiv \alpha.G'$, then $G \{ \ms{rec}X.(\tau.X + E) / X \} \equiv 
\alpha . (G' \{ \ms{rec}X.(\tau.X + E) / X \})$ and
$G \{ \ms{rec}X.(\tau.$ \linebreak $\ms{pri} (E) )
/ X \} \equiv 
\alpha . (G' \{ \ms{rec}X.(\tau. \ms{pri} (E) ) / X \})$. The result trivially follows.

\item If $G \equiv G' + G''$, then $G \{ \ms{rec}X.(\tau.X + E) / X \} \equiv
G' \{ \ms{rec}X.(\tau.X + E) / X \} + G'' \{ \ms{rec}X.(\tau.X + E) / X \}$ and
$G \{ \ms{rec}X.(\tau. \ms{pri} (E) ) / X \} \equiv 
G' \{ \ms{rec}X.(\tau. \ms{pri} (E) ) 
/ X \} + 
G'' \{ \ms{rec}X.(\tau. \ms{pri} (E) ) / X \}$. \\
Suppose $G' \{ \ms{rec}X.(\tau.X + E) / X \} + G'' \{ \ms{rec}X.(\tau.X + E) / X \} 
\arrow{\alpha} H$, we have two cases.
\begin{itemize}
\item If $G' \{ \ms{rec}X.(\tau.X + E) / X \} \arrow{\alpha} H$, then
(by induction) for some $H'$, 
$G' \{ \ms{rec}X. (\tau. \ms{pri} (E) )$ \linebreak 
$/ X \} 
\warrow{\alpha} H'$ and
$(H,H') \in \beta$. Therefore $G' \{ \ms{rec}X.(\tau. 
\ms{pri} (E) ) / X \} + 
G'' \{ \ms{rec}X.(\tau. \ms{pri} (E) )$ \linebreak 
$/ X \} \warrow{\alpha} H'$.
\item If $G'' \{ \ms{rec}X.(\tau.X + E) / X \} \arrow{\alpha} H$, then
the result is derived in a similar way.
\end{itemize}

\item If $G \equiv \ms{rec}Y.G'$ then $G \{ \ms{rec}X.(\tau.X + E) / X \} 
\equiv \ms{rec}Y. G' \{ \ms{rec}X.(\tau.X + E) / X \}$ and \\
$G \{ \ms{rec}X.( \tau.
\ms{pri} (E) ) / X \} \equiv 
\ms{rec}Y. G' \{ \ms{rec}X.( \tau. \ms{pri} (E) ) / X \}$. \\
$\ms{rec}Y. G' \{ \ms{rec}X.(\tau.X + E) / X \} \arrow{\alpha} H$ implies 
$(G' \{ \ms{rec}X.(\tau.X + E) / X \})
\{ \ms{rec}Y. 
G' \{ \ms{rec}X.(\tau.X + E) / X \} / Y \} \equiv
(G' \{ {rec}Y. G' / Y \}) \{ \ms{rec}X.(\tau.X + E) / X \})
\arrow{\alpha}$ 
$H$.
By induction we have 
$(G' \{ {rec}Y. G' / Y \}) \{ \ms{rec}X.( \tau. \ms{pri} (E) ) / X \}
\equiv 
(G' \{ \ms{rec}X.( \tau. \ms{pri} (E) ) / X \})
\{ \ms{rec}Y. G' \{ \ms{rec}X. 
( \tau. $ \linebreak
$\ms{pri} (E) ) / X \}
\warrow{\alpha} H'$ with $(H, 
H') \in \beta$. 
Therefore $\ms{rec}Y. G' \{ \ms{rec}X. 
( \tau. \ms{pri} (E) ) / X \} 
\warrow{\alpha} H'$
with
$(H, H') 
\in \beta$. \\

\end{itemize}
A symmetric inductive proof is performed when we start from a transition 
$G \{ \ms{rec}X.(\tau.X + E) / X \} \arrow{\alpha} H$ in the condition above. The only exception to 
symmetry is the case of $G \equiv X$. The proof of this case follows. \\
If $\ms{rec}X.(\tau. \ms{pri} (E) ) \arrow{\alpha} H$, we have
$\alpha = \tau$ and $H \equiv \ms{pri} (E) \{ \ms{rec}X. 
(\tau. \ms{pri} (E) )
/ X \}$. In this situation $\ms{rec}X.(\tau.X + E) \arrow{\tau}
\ms{rec}X.(\tau.X + E)$.

An identical inductive proof is performed when we start from a transition $G \{ \ms{rec}X.(\tau.X + E) / X \}
\arrow{\delta} H$ in the condition above (we just have to consider $\delta$ instead of a generic $\alpha$), 
apart from the cases $G \equiv X$ and $G \equiv G' + G''$.
In the case $G \equiv X$, due to the priority constraints, no $\delta$ transitions 
can be performed, so the condition obviously holds. In the case $G \equiv G' + G''$, 
to derive the final statement $G' \{ \ms{rec}X.(\tau. \ms{pri} (E) ) / X \} + 
G'' \{ \ms{rec}X.(\tau. \ms{pri} (E) ) / X \} \warrow{\delta} H'$ with $(H,H') \in \beta$, we have
to additionally prove that the $\delta$ transition is not pre-empted. In the first subcase this is
obtained as follows. It must be that 
$G'' \{ \ms{rec}X. (\tau.X + E) / X \} \notarrow{\tau}$, hence, 
since we already showed the condition for relation $\beta$ presented above (and its symmetric one) to hold for any $\alpha$ transition,
we have 
also that $G'' \{ \ms{rec}X. (\tau. \ms{pri} (E) )
/ X \} \notarrow{\tau}$.
The other subcase is analogous.

We conclude the proof that $\beta$ is a weak bisimulation by considering the pair:
$(\ms{rec}X.(\tau.X + E) , \ms{pri} (E) \{ \ms{rec}X. ( \tau .  \ms{pri} (E) ) / X \} )$.

Suppose $\ms{rec}X.(\tau.X + E) \arrow{\gamma} H$, first of all we note
that $\gamma \neq \delta$ for the priority constraints. We have the following
two cases.
\begin{itemize}
\item If $\ms{rec}X.(\tau.X + E) \arrow{\tau} H \equiv \ms{rec}X.(\tau.X + E)$, then
$\ms{pri} (E) \{ \ms{rec}X. ( \tau .  \ms{pri} (E) ) / X \}$ just makes no move. Denoting the latter term with $H'$, we have $(H,H')  \in \beta$.
\item If $\ms{rec}X.(\tau.X + E) \arrow{\alpha} H$ with $E \{ \ms{rec}X.(\tau.X + E) / X \} \arrow{\alpha} H$, 
then, by the first part of the proof,
we derive that, for some $H'$, 
$E \{ \ms{rec}X.( \tau . \ms{pri} (E) ) / X \} 
\warrow{\alpha} H'$ with $(H,H') \in 
\beta$. Hence, $\ms{pri} (E)
\{ \ms{rec}X. ( \tau .  \ms{pri} (E) ) / X \}
\warrow{\alpha}
H'$.
\end{itemize}

Vice-versa if $\ms{pri} (E) \{ \ms{rec}X. ( \tau .  \ms{pri} (E) ) / X \}
\arrow{\gamma} H$ (it must be $\gamma \neq \delta$ for the definition
of operator $\ms{pri}$) we have the following case.
\begin{itemize}
\item If $\ms{pri} (E) \{ \ms{rec}X. ( \tau .  \ms{pri} (E) ) / X \}
\arrow{\alpha} H$ with $E \{ \ms{rec}X. ( \tau .  \ms{pri} (E) ) / X \} \arrow{\alpha} H$, 
then, by the first part of the proof,
we derive that, for some $H'$, 
$E \{ \ms{rec}X.(\tau. X 
+ E) / X \} \warrow{\alpha} H'$ with $(H,H') \in 
\beta$. Hence, $\ms{rec}X.(\tau. X + E) \warrow{\alpha} H'$.
%
%

\end{itemize}

\vspace{.5cm}

\noindent {\bf Soundness of $(Ung4)$}$\;$
The proof that $(Ung4)$ is sound is similar. We show that: \\[.2cm]
$\beta = 
\{ (G \{ \ms{rec}X.(\tau.(\ms{pri}(X) + E) + F) / X \}, G \{ \ms{rec}X.
( \tau . X + E + F ) / X \}) \mid$ \\[.1cm]
$\tabspace{\beta =,} G \; {\rm contains \; at \; most} \; X \; {\rm free} \}
\cup 
\{ ( \ms{pri}(\ms{rec}X.(\tau.(\ms{pri}(X) + E) + F)) \, +$ \\[.1cm]
$\tabspace{\beta =,} 
E \{ \ms{rec}X.(\tau.(\ms{pri}(X) + E) + F) / X \} , 
\ms{rec}X.( \tau . X + E + F ) ) \}$ \\[.2cm]
is a weak bisimulation. From this result it straightforwardly follows
that $\ms{rec}X.(\tau.(\ms{pri}(X) + E) + F) \simeq 
\ms{rec}X.( \tau . X + E + F )$.

Again we start the proof that $\beta$ is a weak bisimulation by 
considering all the pairs
$(G \{ \ms{rec}X.(\tau. (\ms{pri}(X) + E) + F) / X \}, G \{ \ms{rec}X.
( \tau . X + E + F ) / X \})$
such that $G$ contains at most $X$ free. As we did for $(Ung2)$, 
we prove that $\beta$ satisfies the following stronger 
condition
\\[.2cm]
\hspace*{1cm} if $G \{ \ms{rec}X.(\tau.(\ms{pri}(X) + E) + F) / X \} 
\arrow{\gamma} H$ then, 
for some $H'$, \\[.1cm]
\hspace*{4.5cm}
$G \{ \ms{rec}X.( \tau . X + E + F ) / X \} 
\warrow{\gamma} H'$ with $(H,H') \in 
\beta$, \\[.1cm]
\hspace*{1cm} and symmetrically for a move of $G \{ \ms{rec}X.
( \tau . X + E + F ) / X \}$. \\
This is done, first, by induction on the depth of the
inference by which transitions of $G \{ \ms{rec}X.(\tau.
(\ms{pri}(X) + E) + F) / X \} 
\arrow{\alpha} H$ 
are inferred. 
We have the following cases depending on
the structure of $G$.

In the cases $G \equiv \nil$, $G \equiv \gamma . G'$, $G \equiv 
G' + G''$ and $G \equiv \ms{rec}Y.G'$ the proof is analogous to $(Ung2)$. \\
In the case $G \equiv X$ we have that 
$G \{ \ms{rec}X.(\tau.(\ms{pri}(X) + E) + F) / X \} \equiv 
\ms{rec}X.(\tau. 
(\ms{pri}(X) + E) + F)$ 
and 
$G \{ \ms{rec}X. ( \tau . X + E + F ) / X \} \equiv 
\ms{rec}X. ( \tau . X + E + F )$.
Suppose $\ms{rec}X. (\tau.(\ms{pri}(X) + E) + F) \arrow{\alpha} H$, 
we have two cases:
\begin{itemize}
\item if $\ms{rec}X. (\tau.(\ms{pri}(X) + E) + F) \arrow{\tau} 
H \equiv \ms{pri}(\ms{rec}X.(\tau.(\ms{pri}(X) + E) + F)) + E
\{ \ms{rec}X. 
(\tau.(\ms{pri}(X) + E) + F) / X 
\}$, then
$\ms{rec}X. ( \tau . X + E + F ) \arrow{\tau} 
H' \equiv \ms{rec}X. ( \tau . X + E + F )$  with $(H,H') \in  \beta$;

\item if $\ms{rec}X. (\tau.(\ms{pri}(X) + E) + F) \arrow{\alpha} H$ with
$F \{  \ms{rec}X. (\tau.(\ms{pri}(X) + E) + F) / X \} 
\arrow{\alpha} H$, then, by induction we
derive that, for some $H'$, $F \{ \ms{rec}X.( \tau . X + E + F ) / X \} 
\warrow{\alpha} H'$ with $(H,H') \in  \beta$. Hence,
$\ms{rec}X. ( \tau . X + E + F ) \warrow{\alpha}
H'$ with $(H,H') \in  \beta$.
\end{itemize}

A symmetric inductive proof is performed when we start from a transition of $G \{ \ms{rec}X.
( \tau . X + E + F ) / X \}
\arrow{\alpha} H$ in the condition above. The only exception to symmetry is the case of $G \equiv X$. The proof of this case follows. \\
Suppose $\ms{rec}X. (\tau. X + E + F ) \arrow{\alpha} H$,
we have three cases:
\begin{itemize}
\item if $\ms{rec}X. ( \tau . X + E + F ) \arrow{\tau} 
H \equiv \ms{rec}X. ( \tau . X + E + F )$, then 
$\ms{rec}X. (\tau.(\ms{pri}(X) + E) + F) \arrow{\tau}$\\
$ 
H' \equiv \{ ( \ms{pri}(\ms{rec}X.
(\tau.(\ms{pri}(X) + E) + F)) + E
\{ \ms{rec}X.(\tau.(\ms{pri}(X) + E) + F) / X \}$ with $(H,H') \! \in \! \beta$;

\item if $\ms{rec}X. ( \tau . X + E + F ) \arrow{\alpha} H$,
with
$E \{ \ms{rec}X. ( \tau . X + E + F ) / X \} \arrow{\alpha} H$, then, by induction we
derive that, for some $H'$, $E \{ \ms{rec}X. 
(\tau.(\ms{pri}(X) + E) + F) / X \} 
\warrow{\alpha} H'$ with $(H,H') \in  \beta$. Hence,
$\ms{rec}X. (\tau.(\ms{pri}(X) + E) + F) \arrow{\tau} \warrow{\alpha}
H'$ with $(H,H') \in  \beta$.

\item if $\ms{rec}X. ( \tau . X + E + F ) \arrow{\alpha} H$,
with
$F \{ \ms{rec}X. ( \tau . X + E + F ) / X \} \arrow{\alpha} H$, then, by induction we
derive that, for some $H'$, $F \{ \ms{rec}X. 
(\tau.(\ms{pri}(X) + E) + F) / X \} 
\warrow{\alpha} H'$ with $(H,H') \in  \beta$. Hence,
$\ms{rec}X. (\tau.(\ms{pri}(X) + E) + F) \warrow{\alpha}
H'$ with $(H,H') \in  \beta$.

\end{itemize}

An identical pair of inductive proofs is performed when we consider $\delta$ transitions
in the condition above (we just have to consider $\delta$ instead of a generic $\alpha$), 
apart from the cases $G \equiv X$ and $G \equiv G' + G''$.
In the case $G \equiv X$, due to the priority constraints, no $\delta$ transitions 
can be performed, so the condition obviously holds.
In the case $G \equiv G' + G''$ we have to take into account the priority constraints 
by proceeding exactly as for $(Ung2)$. \\

We conclude the proof that $\beta$ is a weak bisimulation by considering the 
pair:\\
$( \ms{pri}(\ms{rec}X.(\tau.
(\ms{pri}(X) + E) + F)) + E
\{ \ms{rec}X.(\tau.(\ms{pri}(X) + E) + F) / X \} , 
\ms{rec}X.( \tau . X 
 + E + F ) )$.

Suppose $\ms{pri}(\ms{rec}X.(\tau.(\ms{pri}(X) + E) + F)) + E
\{ \ms{rec}X.(\tau.(\ms{pri}(X) + E) + F) / X \} 
\arrow{\gamma} H$,
first of all we note
that $\gamma \neq \delta$ for the priority constraints.
We have the following three cases.
\begin{itemize}
\item If $\ms{pri}(\ms{rec}X.(\tau.(\ms{pri}(X) + E) + F)) + E
\{ \ms{rec}X.(\tau.(\ms{pri}(X) + E) + F) / X \} \arrow{\tau}$ 
$H \equiv \ms{pri}(\ms{rec}X.(\tau. \linebreak
(\ms{pri}(X)
+ E) + F)) + E
\{ \ms{rec}X.(\tau.(\ms{pri}(X) + E) + F) / X \}$, then 
$\ms{rec}X.( \tau . X + E + F )$ just makes no move. Denoting the latter term with $H'$, we have $(H,H')  \in \beta$.
\item If $\ms{pri}(\ms{rec}X.(\tau.(\ms{pri}(X) + E) + F)) + E
\{ \ms{rec}X.(\tau.(\ms{pri}(X) + E) + F) / X \} \arrow{\alpha} H$ with 
$F \{ \ms{rec}X.$ $(\tau.(\ms{pri}(X) + E) + F) / X \}
\arrow{\alpha} H$, then, by the first part of the proof we
derive that, for some $H'$, $F \{ \ms{rec}X.( \tau . X + E + F ) / X \} 
\warrow{\alpha} H'$ with $(H,H') \in  \beta$. Hence,
$\ms{rec}X. ( \tau . X + E + F ) \warrow{\alpha}
H'$ with $(H,H') \in  \beta$. 

\item If $\ms{pri}(\ms{rec}X.(\tau.(\ms{pri}(X) + E) + F)) + E
\{ \ms{rec}X.(\tau.(\ms{pri}(X) + E) + F) / X \} \arrow{\alpha} H$ with 
$E \{ \ms{rec}X.$ $(\tau.(\ms{pri}(X) + E) + F) / X \}
\arrow{\alpha} H$, then, by the first part of the proof we
derive that, for some $H'$, $E \{ \ms{rec}X.( \tau . X + E + F ) / X \} 
\warrow{\alpha} H'$ with $(H,H') \in  \beta$. Hence,
$\ms{rec}X. ( \tau . X + E + F ) \warrow{\alpha}
H'$ with $(H,H') \in  \beta$.

%
%
\end{itemize}

Vice-versa if $\ms{rec}X.( \tau . X + E + F ) \arrow{\gamma} H$
(it must be $\gamma \neq \delta$ for the priority constraints)
we have the following three cases.
\begin{itemize}
\item If $\ms{rec}X.( \tau . X + E + F ) \arrow{\tau} 
H \equiv \ms{rec}X.( \tau . X + E + F )$, then 
$\ms{pri}(\ms{rec}X.(\tau.(\ms{pri}(X) + E) + F)) + E
\{ \ms{rec}X.(\tau.(\ms{pri}(X) + E) + F) / X \}$ just makes no move.
Denoting the latter term with $H'$, we have $(H,H')  \in \beta$.

\item if $\ms{rec}X. ( \tau . X + E + F ) \arrow{\alpha} H$,
with
$E \{ \ms{rec}X. ( \tau . X + E + F ) / X \} \arrow{\alpha} H$, then, by the first part of the proof we
derive that, for some $H'$, $E \{ \ms{rec}X. 
(\tau.(\ms{pri}(X) + E) + F) / X \} 
\warrow{\alpha} H'$ with $(H,H') \in  \beta$. Hence,
$\ms{pri}(\ms{rec}X.(\tau. 
(\ms{pri}(X) + E) + F)) + E
\{ \ms{rec}X.(\tau.(\ms{pri}(X) + E) + F) / X \} \warrow{\alpha} H'$ with $(H,H') \in  \beta$.

\item if $\ms{rec}X. ( \tau . X + E + F ) \arrow{\alpha} H$,
with
$F \{ \ms{rec}X. ( \tau . X + E + F ) / X \} \arrow{\alpha} H$, then, by the first part of the proof we
derive that, for some $H'$, $F \{ \ms{rec}X. 
(\tau.(\ms{pri}(X) + E) + F) / X \} 
\warrow{\alpha} H'$ with $(H,H') \in  \beta$. Hence,
$\ms{pri}(\ms{rec}X.(\tau. 
(\ms{pri}(X) + E) + F)) + E
\{ \ms{rec}X.(\tau.(\ms{pri}(X) + E) + F) / X \} \warrow{\alpha}
H'$ with $(H,H') \in  \beta$.
%
%
\qedhere
\end{itemize}
\end{proof}

\subsection{Completeness}

We now show that the axiom system $\cala$ is complete over
processes of $\calp$. In order to do this we follow the lines
of~\cite{Mil89}, so we deal with systems of recursive equations.

We start by introducing the machinery necessary for proving completeness
over guarded expressions $E \in \cale$. Afterwards we will show
how the axioms $(Ung)$ can be used to turn an unguarded processes 
of $\calp$ into a guarded process of $\calp$. Notice that the new operator
$\ms{pri}(E)$ is used only in the intermediate steps of the
second procedure.

\begin{defi}
An {\it equation set} with {\it formal variables} 
$\tilde{X} = \{ X_1, \dots , X_n \}$ and {\it free variables}
$\tilde{W} = \{ W_1, \dots , W_m \}$, where $\tilde{X}$ and $\tilde{W}$ are disjoint, 
is a set $S = \{ X_i = H_i \mid 1 \leq i \leq n \, \wedge \, H_i \in \cale \}$ of
equations such that the expressions $\tilde{H} = \{ H_1, \dots , H_n \}$ have free
variables in $\tilde{X} \cup \tilde{W}$. In the following we will also represent $S$ by using 
the shorthand $S : \tilde{X} = \tilde{H}$. We call $X_1$ the {\it distinguished} variable
of $S$.

We say that an expression $E$ {\it provably satisfies} $S$ if there are
expressions $\tilde{E} = \{ E_1, \dots , E_n \}$ with free variables in 
$\tilde{W}$ such that $E_1 \equiv E$ and for $1 \leq i \leq n$ we have 
$\cala \vdash E_i = H_i \{ \tilde{E} / \tilde{X} \}$. We can also 
simply write: $\cala \vdash \tilde{E} = \tilde{H} \{ \tilde{E} / \tilde{X} \}$.

Finally, we say that $S$ is {\it closed} if $W = \emptyset$. 
\end{defi}
In order to define (similarly as in~\cite{Mil89}) guardedness over
equation sets, we consider the relation 
$\arrowd{ung}{S} \subseteq \tilde{X} \times \tilde{X}$, defined as follows: \\[.3cm]
$\hspace*{2cm} 
\begin{array}{lll}
X_i \arrowd{ung}{S} X_j \hspace*{1cm} & {\rm iff} \hspace*{1cm} & H_i \rhd X_j \\[.1cm]
\end{array}$\\[.1cm]
where $E \rhd X$ denotes that expression $E$ contains a free unguarded occurrence of $X$.

\begin{defi}\label{defGES}
An equation set $S$ with formal variables 
$\tilde{X} \, = 
\{ X_1, \dots , X_n \}$ is {\it guarded} if there is no 
cycle $X_i \sarrow{ung}\!\!\!_S^{\,\,+} X_i$.

\end{defi}

The equation sets that are actually dealt with in the proof of completeness
of~\cite{Mil89} belong to the subclass of {\it standard 
equation sets}. Here we consider the subclass of {\it prioritized} ones.

\begin{defi}
An equation set $S = \{ X_i = H_i \mid 1 \leq i \leq n \}$,
with formal variables 
$\tilde{X} = \{ X_1, \dots , X_n \}$ and free variables
$\tilde{W} = \{ W_1, \dots , W_m \}$, is {\it standard} if 
each expression $H_i$ ($1 \leq i \leq n$) is of the 
form:~\footnote{We assume $\sum_{j \in J} E \equiv \nil$ if $J =
\emptyset$.}
\cws{0}{0}{
H_i \equiv \sum_{j \in J_i} \gamma_{i,j} . X_{f(i,j)} + \sum_{k \in K_i} 
W_{g(i,k)}
}
Moreover, if it also holds that:
\cws{5}{10}{
\exists j \in J_i : \gamma_{i,j} = \tau \; \Rightarrow \; \, \notexists
j \in J_i : \gamma_{i,j} = \delta \; }
we say that $S$ is {\it prioritized}.

\end{defi}

As in~\cite{Mil89}, for a standard equation set $S$ we define the relations 
$\arrowd{}{S} \subseteq \tilde{X} \times \ms{Act} \times \tilde{X}$ and $\rhd_S 
\subseteq
\tilde{X} \times \tilde{W}$ as follows: \\[.3cm]
$\hspace*{2cm} 
\begin{array}{lll}
X_i \arrowd{\gamma}{S} X_j \hspace*{1cm} & {\rm iff} \hspace*{.5cm} &
\gamma . X_j \; {\rm occurs \; in} \; H_i \\[.1cm]
X_i \rhd_S W \hspace*{1cm} & {\rm iff} \hspace*{.5cm} &
W \; {\rm occurs \; in} \; H_i
\end{array}$\\[.1cm]

Notice that, from Definition~\ref{defGES}, we have that a 
standard equation set $S$ with formal variables 
$\tilde{X} \, =
\{ X_1, \dots , X_n \}$ is guarded if there is no 
cycle $X_i \sarrow{\tau}\!\!\!_S^{\,\,+} X_i$.

\begin{defi}
A free variable $W$ is guarded in the standard equation set
$S = \{ X_i = H_i \mid 1 \leq i \leq n \}$ if it is not
the case that 
$X_1 \sarrow{\tau}\!\!\!_S^{\,\,*} X_i \rhd_S W$.

\end{defi}

We obviously have that every standard equation set can be reduced to a prioritized 
one.

\begin{thm}\label{reduction}
For every standard equation set $S$ there is a prioritized standard
equation set $S'$ such that every expression $E$ that provably satisfies
$S$, provably satisfies $S'$ as well. Moreover, if $S$ is guarded then 
$S'$ is also guarded.
\end{thm}

\begin{proof}
$S' = \{ X_i = H'_i \mid 1 \leq i \leq n \}$ is obtained from $S = \{ X_i = H_i \mid 1 \leq i \leq n \}$ by eliminating 
the occurrences of $\delta . X_i$ (for any $i$) in the equations of $S$
which include terms $\tau . X_j$ (for any $j$). The expressions 
$\tilde{E} = \{ E_1, \dots , E_n \}$ used to show that $E$ that provably satisfies
$S$ (with $E_1 \equiv E$) can also be used to prove that $E$ provably satisfies
$S'$: for every $i$ such that $1 \leq i \leq n$ we have that $\cala \vdash
E_i = H_i \{ \tilde{E} / \tilde{X} \} = H'_i \{ \tilde{E} / \tilde{X} \}$
by using laws $(Pri6)$ and $(Pri3)$.
If $S$ is guarded then $S'$ is easily proven to be guarded by contradiction:
if there is a cycle $X_i \sarrow{\tau}\!\!\!_{S'}^{\,\,+} X_i$ for some $i$ then
the same cycle must be also in $S$.
\end{proof}

The following theorem guarantees that from a guarded 
expression $E \in \cale$ we can derive a (prioritized) standard guarded equation 
set which is provably satisfied by $E$.

\begin{thm}[representability]\label{repr}
Every guarded expression $E \in \cale$ with free variables 
$\tilde{W}$ provably satisfies a standard guarded 
equation set $S$ with free variables in $\tilde{W}$.
Moreover, if a free variable $W$ is guarded in $E$
then $W$ is guarded in $S$.
\end{thm}

\begin{proof}
The proof, by induction on the structure of $E$, is exactly as in~\cite{Mil89}.
\end{proof}

Once established the structure of prioritized standard guarded equation sets $S$, 
completeness over guarded expressions is a consequence of the uniqueness of the solution of such equation stets and of the possibility to merge them (after saturation, see Definition \ref{satur}) if they have equivalent solutions,
as in~\cite{Mil89}. 

The following theorem shows that every guarded equation set
(not necessarily standard) has a unique solution up to provable equality.

\begin{thm}[one and only one solution]\label{onesol}
If $S$ is a guarded equation set with free variables in $\tilde{W}$,
then there is a guarded expression $E \in \cale$ with free variables in $\tilde{W}$
that provably satisfies $S$. Moreover,
if $F$ with free variables in $\tilde{W}$ provably satisfies $S$,
then $\cala \vdash E = F$.

\end{thm}

\begin{proof}
The proof is exactly as in~\cite{Mil89}. We just additionally show that the expression $E \in \cale$ 
obtained with the inductive procedure in~\cite{Mil89} is guarded (not proven in~\cite{Mil89}).

The inductive procedure in~\cite{Mil89} constructs $E$ as follows.
Given equation set $S = \{ X_i = H_i \mid 1 \leq i \leq n \}$, the expression $E$ that provably satisfies $S$ is obtained by induction on the number of equations $n$. The base case is $n=1$: $E$ is simply $rec X_1 . H_1$. In the inductive case, given equation set $S = \{ X_i = H_i \mid 1 \leq i \leq n+1 \}$, $E$ is inductively defined as the 
expression that provably satisfies the equation set $\{ X_i = H'_i \mid 1 \leq i \leq n \}$, where $H'_i = H_i \{ rec X_{n+1} . H_{n+1} / X_{n+1} \}$.

It is immediate to show that $E$ obtained from such a procedure is guarded by contradiction.
Suppose that $E$ contains a subterm $\ms{rec}X.E'$ such that $E'$ includes a free unguarded occurrence of $X$. Now consider the recursion operators that appear in $\ms{rec}X.E'$ and which include the free unguarded occurrence of $X$ in their scope (among them there is the recursion operator $\ms{rec}X.E'$ itself). Based on the inductive procedure above, such recursion operators (considered syntactically from outside to inside) would be in correspondence with a sequence of consecutive $X_i \arrowd{ung}{S} X_j$ steps in the equation set $S$: $X_i$ corresponds to the syntactical occurrence of a $rec X_i$ operator and $X_j$ either to the syntactical occurrence of $rec X_j$ or, if $X_j \equiv X$, to the syntactical occurrence of $X$. Such a sequence would 
lead from $X$ to itself, hence $S$ would not be guarded.
\end{proof}


In the following, given a guarded equation set $S$, we will refer to 
the solution $E$ determined in the proof of the theorem above as the
{\it standard solution} of $S$.

\begin{defi}\label{satur}
A {\it saturated} standard equation set $S$ with formal variables 
$\tilde{X}$ is a standard equation set such that for all $X \in \tilde{X}$: \\[.3cm]
$\hspace*{2cm} 
\begin{array}{llll}
(i) \hspace*{.5cm} &
X \sarrow{\tau}\!\!\!_S^{\,\,*} \; \arrowd{\alpha}{S} \, 
\sarrow{\tau}\!\!\!_S^{\,\,*} \; X' & 
\hspace*{.5cm} \Rightarrow \hspace*{.5cm} & 
X \arrowd{\alpha}{S} X' \\[.1cm]
(ii) \hspace*{.5cm} &
X \sarrow{\tau}\!\!\!_S^{\,\,*} \; \arrowd{\delta}{S} \, 
\sarrow{\tau}\!\!\!_S^{\,\,*} \; X' & 
\hspace*{.5cm} \Rightarrow \hspace*{.5cm} & 
X \arrowd{\delta}{S} X' \\[.1cm]
(iii) \hspace*{.5cm} &
X \sarrow{\tau}\!\!\!_S^{\,\,*} \; \rhd_S \; W &
\hspace*{.5cm} \Rightarrow \hspace*{.5cm} & 
X \rhd_S W
\end{array}$\\[.1cm]
%
\end{defi}

\begin{lem}\label{satlemma}
Let $E \in \cale$ provably satisfy $S$, standard and guarded. Then
there is a saturated, standard and guarded equation set $S'$
provably satisfied by $E$.
\end{lem}

\begin{proof}
From~\cite{Mil89} we know that there is a saturated standard equation set $S'$
provably satisfied by $E$: the proof is based on the axioms $(A1)-(A4)$ and $(Tau1)-(Tau3)$ that are unchanged with respect to the standard machinery. 
\end{proof}

The possibility of saturating standard and guarded equation sets $S$
leads to the following theorem. 

%

\begin{thm}[mergeability]
Let process $P \in \calp$ provably satisfy $S$, and process $Q \in \calp$ provably satisfy
$T$, where both $S$ and $T$ are prioritized, standard, guarded and closed sets of equations, and
let $P \simeq Q$. Then there is a prioritized, standard, guarded and closed equation set $U$ 
provably satisfied by both $P$ and $Q$.
\end{thm}

\begin{proof}
A standard guarded (and closed) equation set $U'$ provably satisfied by both $P$ and $Q$ can be derived from $S$ and $T$ (considered just as standard
guarded equation sets) by means of the procedure in the proof of the related theorem in~\cite{Mil89}.
In particular, as in~\cite{Mil89}, we can assume standard and guarded equation sets $S$ and $T$ to have been preliminarily saturated because of Lemma~\ref{satlemma}. The relation between the variables of saturated $S$ and $T$, which in~\cite{Mil89} is deduced from $P \simeq Q$, still holds here (in spite of priority).  This is because $S$ and $T$ are assumed to be prioritized standard guarded equation sets, hence no $\delta$ prefix that does not actually correspond 
to a weak transition $\warrow{\delta}$ occurs in the saturated standard equation sets obtained from $S$ and $T$. From the existance of $U'$ we derive the existance of the
the prioritized standard guarded equation set $U$ by means of Theorem~\ref{reduction}.
\end{proof}

Hence we have proved completeness over guarded processes of $\calp$ . 

\begin{thm}\label{compguarded}
If $P$ and $Q$ are guarded processes of $\calp$ and $P \simeq Q$ then 
$\cala \vdash P = Q$.

\end{thm}

Now we show that each unguarded process can be turned into a guarded
process of $\calp$, so that we obtain completeness also over 
unguarded processes. 
We start with a technical lemma which, in analogy to~\cite{Mil89},
will be used in the proof of this result.

\begin{lem}\label{makeung}
If $X$ occurs free and unguarded in $E \in \cale$, then $\cala \vdash E = X + E$.
\end{lem}
\begin{proof}
The proof is exactly as in~\cite{Mil89}.
\end{proof}


\begin{thm}\label{unguard}
For each process $P \in \calp$ there exists a guarded $P' \in \calp$ such that
$\cala \vdash P = P'$.
\end{thm}

\begin{proof}
We show, by structural induction, that given an expression 
$E \in \cale$, it is possible to find an expression $F \in \cale_{pri}$
such that:
\begin{enumerate}
\item if $\ms{pri}(G)$ is a subexpression of $F$ then $G \equiv X$ for
some free variable $X$;
\item for any free variable $X$, $\ms{pri}(X)$ is weakly guarded in $F$, i.e. 
each occurrence of $\ms{pri}(X)$ is within some 
subexpression of $F$ of the form $\gamma.G$;
\item a summation cannot have both $\ms{pri}(X)$ and $Y$ as arguments, for any
(possibly coincident) variables $X$ and $Y$;
\item for any variable $X$, each subterm $\ms{rec}X . G$ of $F$  
is (strongly) guarded in $F$, i.e. each occurrence of $\ms{rec}X . G$ 
is within some subexpression of $F$ of the form $\gamma.H$, with $\gamma 
\neq \tau$;
\item $F$ is guarded;
\item all variables occurring free in $F$ occur free also in $E$ 
\item $\cala \vdash E = F$.
\end{enumerate}

Notice that a consequence of property $4$ is that
each unguarded occurrence of any free variable $X$ of $F$ does not lie 
within the scope of a subexpression $\ms{rec}Y . G$ of $F$;

Showing this result proves the theorem, in that if $E \in \cale$ is a process of 
$\calp$, i.e. a closed term, we have (by the properties of $F$ above) that $F$ 
is also a process of $\calp$, it is guarded, and $\cala \vdash E = F$.

The result above is derived by structural induction on the syntax of an 
expression $E \in \cale$ as follows.

\begin{itemize}
\item If $E \equiv \nil$, then $F \equiv \nil$.
\item If $E \equiv X$, for some variable $X$, then $F \equiv X$.
\item If $E \equiv \gamma. E'$ then $F \equiv \gamma . F'$, where
$F'$ is the term obtained from $E'$ via the induction hypothesis.
\item If $E \equiv E' + E''$, then $F \equiv F' + F''$, where 
$F'$ and $F''$ are the terms obtained from $E'$ and $E''$, respectively, via the 
induction hypothesis.
\item If $E \equiv \ms{rec}X . E'$, then $F$ is evaluated as follows. \\
First of all we derive from the term $F'$, obtained from $E'$ 
via the induction hypothesis, a term $F''$ such that:
\begin{itemize}
\item $F''$ satisfies the properties $1-5$ above;
\item $X$ is guarded in $F''$;
\item all variables occurring free in $F''$ occur free also in $E'$;
\item $\cala \vdash \ms{rec}X . E' = \ms{rec}X . F''$;
\end{itemize}
We start by eliminating all fully unguarded occurrences of $X$ 
in $F'$, via the axiom $(Ung1)$, so that we obtain a term $G$ such that
$X$ is weakly guarded in $G$ and $\ms{rec}X . E' = \ms{rec}X . G$. 
The axiom $(Ung1)$ is sufficient because $F'$ satisfies the property $2$ above,
so we do not need to deal with priority. Notice that $G$ still satisfies properties $1-5$ above. \\
Then we need to remove weakly guarded occurrences of $X$ in $G$ (if 
there are not, we just take $F'' \equiv G$). 
In order
to do this we begin by evaluating a term $G'$ such that:
\cws{11}{15}{G' \equiv \tau.X + G''} 
where $X$ is guarded in $G''$ and $\ms{rec}X . G = \ms{rec}X . G'$. \\
We do this by employing the following iterative procedure, where initially
we let $H \equiv G$.

\begin{itemize}

\item We start each iteration with $H$ such that $X$ is weakly 
guarded in $H$ and $H$ satisfies the properties $1$, $3$, $4$ and $5$ above. 
If there still exists $H' \not\equiv X$ with $X$ occurring unguarded in $H'$
such that $H$ has the form $\tau . H' + H''$ we continue with the
procedure. Otherwise $H$ can be turned into a term $G' \equiv \tau.X + G''$ 
with the property above, by simply 
using $\tau.X + \tau.X = \tau.X$ and we are finished. 

\item Since $H$ satisfies property $3$, we have the following cases for the 
structure of $H'$ (after applying the idempotency law $(A3)$ to $X$ or 
$\ms{pri}(X)$ if needed).
\begin{enumerate}

\item $H'$ has the form $X + H'''$, where $X$ is weakly guarded in $H'''$. 
In this case we consider the term $H'''' \equiv \tau . X + H''' + H''$ and
we have $\ms{rec}X . H = \ms{rec}X . H''''$ by applying the axiom $(Ung3)$.

\item $H'$ has the form $\ms{pri}(X) + H'''$, where $X$ is weakly guarded in $H'''$.
In this case we consider again the term $H'''' \equiv \tau . X + H''' + H''$ and
we have $\ms{rec}X . H = \ms{rec}X . H''''$ by applying the axiom $(Ung4)$.


\item $X$ is weakly guarded in $H'$ . We have two sub-cases.

\begin{itemize}
\item[(i)] If $\ms{pri}(X)$ occurs in $H'$, then we do the following. \\
We consider a new variable $Y'$ for each variable $Y$ such that 
$\ms{pri}(Y)$ appears in $H'$. Let $H'''$ be the unique expression
of $\cale$ such that if we replace $\ms{pri}(Y)$ for each new variable $Y'$ 
inside $H'''$, we obtain the term $H'$. \\
Since $X$ occurs unguarded in $H'$, then $X'$ occurs unguarded in $H'''$ and
from Lemma~\ref{makeung} we have that $H''' = X' + H'''$. 
Hence, since substitution preserves equality, we have $H' = \ms{pri}(X) + H'$
and we continue as in the case $2$.

\item[(ii)] Otherwise we do the following. \\
We consider a new variable $Y'$ for each variable $Y$ such that 
$\ms{pri}(Y)$ appears in $H'$. Let $H'''$ be the unique expression
of $\cale$ such that if we replace $\ms{pri}(Y)$ for each new variable $Y'$ 
inside $H'''$, we obtain the term $H'$. \\
Since $X$ occurs unguarded in $H'$, then $X$ occurs unguarded also in $H'''$ and
from
Lemma~\ref{makeung} we have that $H''' = X + H'''$. 
Hence, since substitution preserves equality, we have $H' = X + H'$
and we continue as in the case $1$.

\end{itemize}

\end{enumerate}

\item Now a new iteration is performed starting from the term $H''''$ that we have 
obtained. Since there is at least one unguarded
occurrence of $X$ such that the length of the $\tau$ path that leads to
that occurrence is shorter in $H''''$ than in $H$, we are guaranteed that the
iterative procedure will eventually terminate.

\end{itemize}
Now we are in a position to derive the term $F''$ with the properties we 
described above. From $G' \equiv \tau.X + G''$ we have $\ms{rec}X . G' = 
\ms{rec}X . \tau. \ms{pri}(G'')$ by applying the axiom $(Ung2)$. 
$X$ is guarded in the term $\tau. \ms{pri}(G'')$ and 
such term satisfies properties $2$, $4$ and $5$. 

Due to the property $4$, the term $G''$ has the following structure:
\cws{10}{5}{G'' \equiv \sum_i \alpha_i . H_i +
\sum_j \delta . H'_j + \sum_k Y_k + \sum_h \ms{pri}(Y'_h) .}
where the variables $Y_k$ and $Y'_h$ are free and do not coincide with $X$
because $X$ is guarded in $G''$.
By applying the axioms $(Pri1)$, $(Pri2)$, $(Pri3)$, $(Pri4)$
and $(Pri5)$, we obtain a term $F'' = \tau. \ms{pri}(G'')$,
where:
\cws{10}{5}{F'' \equiv \tau . ( \, \sum_i \alpha_i . H_i +
\sum_k \ms{pri}(Y_k) + \sum_h \ms{pri}(Y'_h) \,) .}
We have, thus, obtained an $F''$ that satisfies also properties $1$ and $3$.

Finally, once obtained (possibly by applying the procedure above for weak guarded recursion elimination, if needed) a term $F''$ with the properties listed above, we have to transform the term $\ms{rec}X . F''$ so to 
obtain a term $F$ that satisfies the seven properties of the induction statement. \\
$\ms{rec}X . F''$ already satisfies the properties $2$ and $3$
because $F''$ satisfies them. Moreover $\ms{rec}X . F''$ already satisfies
property $5$, i.e. it is guarded, because $X$ is guarded in $F''$ and 
$F''$ already satisfies that property. \\
In order to satisfy property $1$ we have to remove the occurrences
of $\ms{pri}(X)$ because now $X$ is no longer a free variable. \\
We unfold $\ms{rec}X . F''$ via the axiom $(Rec1)$ and we obtain 
$F''' \equiv F'' \{ \ms{rec}X . F'' 
/ X \} = \ms{rec}X . F''$.
Notice that $F'''$ may include subterms of the form $\ms{pri}(\ms{rec}X . $ 
$F'')$. \\
We deal with such subterms in the following way. 
We have $\ms{pri}(\ms{rec}X . F'') = \ms{pri}(F'' \{ \ms{rec}X . 
F'' \linebreak
/ X \})$,
by applying law $(Rec1)$. \\
Since $X$ is guarded in $F''$ and $F''$ satisfies the property $4$, 
the term $F'' \{ \ms{rec}X .$ 
$F'' / X \}$
has the following structure:
\cws{10}{5}{F'' \{ \ms{rec}X . F'' / X \} \equiv \sum_i \alpha_i . H_i +
\sum_j \delta . H'_j + \sum_k Y_k + \sum_h \ms{pri}(Y'_h) .}
where the variables $Y_k$ and $Y'_h$ are free and obviously do not coincide 
with $X$.
By applying the axioms $(Pri1)$, $(Pri2)$, $(Pri3)$, $(Pri4)$
and $(Pri5)$, we obtain a term $T = \ms{pri}(F'' \{ \ms{rec}X .
F'' / X \})$,
where:
\cws{10}{5}{T \equiv \sum_i \alpha_i . H_i +
\sum_k \ms{pri}(Y_k) + \sum_h \ms{pri}(Y'_h) .}
Now let us consider a new variable $X'$. 
Let $T'$ be the unique term, not having
$\ms{pri}(\ms{rec}X . F'')$ as a subterm, such that 
$T' \{ \ms{pri}(\ms{rec}X . F'') / X' \} \equiv T$.
Since $X'$ is guarded in $T'$ (because $X$ is guarded in $F''$) and 
serial in $T'$ (thanks to the transformation of 
$\ms{pri}(F'' \{ \ms{rec}X . F'' / X \})$ into $T$),
from $\ms{pri}(\ms{rec}X . F'') = T' \{ \ms{pri}(\ms{rec}X . F'') / X' \}$
we derive, by applying the law $(Rec2)$, $\ms{pri}(\ms{rec}X . F'') 
= \ms{rec}X' . T'$. \\
Therefore now we can replace all occurrences of $\ms{pri}(\ms{rec}X . F'')$
in $F'''$ with $\ms{rec}X' . T'$, and we obtain a term $F'''' = F'''$ 
not having
$\ms{pri}(\ms{rec}X . F'')$ as a subterm. In order to satisfy the
property $1$ we still have to remove the occurrences of $\ms{rec}X . F''$
in $F''''$, so that we get rid of the occurrences of 
$\ms{pri}(X)$ appearing inside $F''$. \\
In order to do this, we consider another new variable $X''$. 
Let $T''$ be the unique term, not having
$\ms{rec}X . F''$ as a subterm, such that 
$T'' \{ \ms{rec}X . F'' / X'' \} 
\equiv F''''$. 
Since $X''$ is guarded in $T''$ (because each occurrence of the expression 
$\ms{rec}X . F''$ is guarded both in $T'$ and in $F'''$,
since $X$ is guarded in $F''$) and 
serial in $T''$ (thanks to the substitution of
$\ms{pri}(\ms{rec}X . F'')$ with $\ms{rec}X' . T'$),
from $\ms{rec}X . F'' = F'''' \equiv T'' \{ \ms{rec}X . F'' / X'' \}$
we derive, by applying the law $(Rec2)$, $\ms{rec}X . F'' =
\ms{rec}X'' . T''$. \\
Now we have that $\ms{rec}X'' . T''$ satisfies the property $1$, because
such term no longer uses the variable $X$, and does not include 
occurrences of $\ms{pri}(X')$ or $\ms{pri}(X'')$. \\
Moreover, since $X''$ is guarded in $T''$ and each occurrence of $\ms{rec}X'. 
T'''$, for any $T'''$, is guarded in $T''$ (because each occurrence of
$\ms{pri}(X)$ is guarded in $F''$), in order to obtain property $4$ 
it is sufficient to unfold $\ms{rec}X'' . T''$, by applying the law $(Rec1)$. 
In this way we finally obtain $F \equiv T'' \{ \ms{rec}X'' . T'' 
/ X'' \} 
= \ms{rec}X'' . T''$ that satisfies the seven properties above.
\end{itemize}
\end{proof}

\begin{exa}
We here show an example of transformation of an unguarded process $P \in \calp$ into a guarded $P' \in \calp$, following the inductive procedure described in the proof of Theorem~\ref{unguard}. We take $P$ to be:
$$\ms{rec}Y. (\ms{rec}X. (\tau.(X+ b.X + \delta. \nil) + Y) + Y)$$
We start from subexpression $\tau.(X+ b.X + \delta. \nil) + Y$, which is not modified by the inductive transformation. 

When $\ms{rec}X. (\tau.(X + b.X + \delta. \nil)+Y)$ is considered,
we first remove fully unguarded occurrences of $X$ from $\tau.(X+ b.X + \delta. \nil) + Y$ (there are not). Then, we remove weakly unguarded occurrences of $X$: one iteration yielding $\tau.X + b.X + \delta. \nil +Y$  (case $(1)$ for the structure of $H'$, which in our example is $X + b.X + \delta. \nil$) is sufficient. Such an iteration is correct because, thanks to $(Ung3)$ we have  $\ms{rec}X. (\tau.(X + b.X + \delta. \nil)+Y) = \ms{rec}X. (\tau.X + b.X + \delta. \nil +Y)$. We can now apply $(Ung2)$ to get $\ms{rec}X. (\tau.\ms{pri}(b.X + \delta. \nil +Y))$
and $(Pri)$ axioms to get $\ms{rec}X. (\tau.(b.X +\ms{pri}(Y)))$. We terminate this level of the structural induction by 
unfolding recursion via axiom $(Rec1)$, thus getting $\tau.(b.\ms{rec}X.\tau.(b.X +\ms{pri}(Y)) +\ms{pri}(Y))$.\footnote{In the proof of Theorem~\ref{unguard} the actual term derived at the end of the induction step is $\tau.(b.\ms{rec}X''.\tau.(b.X'' +\ms{pri}(Y)) +\ms{pri}(Y))$ with variable $X''$ taking the place of $X$.}

Finally, we consider $\ms{rec}Y. (F'+Y)$ with $F' \equiv \tau.(b.\ms{rec}X. (\tau.(b.X +\ms{pri}(Y))) +\ms{pri}(Y))$.
We first remove fully unguarded occurrences of $Y$ in $\tau.(b.\ms{rec}X.\tau.(b.X +\ms{pri}(Y)) +\ms{pri}(Y))+Y$
obtaining $\tau.(b.\ms{rec}X.\tau.(b.X +\ms{pri}(Y)) +\ms{pri}(Y))$. Such a transformation is correct because, thanks to $(Ung1)$, the two expression are equated when used as arguments of $\ms{rec}Y.\_$.
Then, we remove weakly unguarded occurrences of $Y$: one iteration yielding $\tau.Y + b.\ms{rec}X. \tau.(b.X +\ms{pri}(Y))$  (case $(2)$ for the structure of $H'$, which in our example is $b.\ms{rec}X.\tau.(b.X +\ms{pri}(Y)) + \ms{pri}(Y)\!$ ) is sufficient. Such an iteration is correct because, thanks to $(Ung4)$, the two expression are equated when used as arguments of $\ms{rec}Y.\_$. 
We can now apply $(Ung2)$ to get $\ms{rec}Y.\tau.\ms{pri}(b.\ms{rec}X. \tau.(b.X +\ms{pri}(Y)))$
and $(Pri)$ axioms to get $\ms{rec}Y.F''$, with $F'' \equiv \tau.b.\ms{rec}X.\tau.(b.X +\ms{pri}(Y))$.
In order to get rid of $\ms{pri}(Y)$ subterms 
we now unfold recursion via axiom $(Rec1)$, thus getting 
$\tau.b.\ms{rec}X.\tau.(b.X +\ms{pri}(\ms{rec}Y.F''))$. We now aim at replacing $\ms{pri}(\ms{rec}Y.F'')$ subterms. This is done by 
using the fact that the same transformation via axiom $(Rec1)$ can be done inside a $\ms{pri}(\_)$ operator, i.e.\
$\ms{pri}(\ms{rec}Y.F'') = \ms{pri}(\tau.b.\ms{rec}X.\tau.(b.X +\ms{pri}(\ms{rec}Y.F'')))$. By using $(Pri)$ axioms on the latter term we obtain $\tau.b.\ms{rec}X.\tau.(b.X +\ms{pri}(\ms{rec}Y.F''))$. Considering a new variable $Y'$ such a term can be written as $\tau.b.\ms{rec}X.\tau.(b.X + Y') \{ \ms{pri}(\ms{rec}Y.F'') / Y' \}$.
We have therefore that $\ms{pri}(\ms{rec}Y.F'')=\tau.b.\ms{rec}X.\tau.(b.X + Y') \{ \ms{pri}(\ms{rec}Y.F'') / Y' \}$ with $Y'$ being guarded and serial in $\tau.b.\ms{rec}X.\tau.(b.X + Y')$. Hence, by applying axiom $(Rec2)$, we have $\ms{pri}(\ms{rec}Y.F'')=
\ms{rec}Y'.\tau.b.\ms{rec}X.\tau.(b.X + Y')$.
We can, thus, achieve our  goal of replacing such a term inside $\tau.b.\ms{rec}X.\tau.(b.X +\ms{pri}(\ms{rec}Y.F''))$,  obtaining $\tau.b.\ms{rec}X.\tau.(b.X +\ms{rec}Y'.\tau.b.\ms{rec}X.\tau.(b.X + Y'))$.
\end{exa}

From Theorem~\ref{compguarded} and Theorem~\ref{unguard} 
we derive the completeness of $\cala$ over processes of $\calp$.

\begin{thm}
Let $P, Q \in \calp$. If $P \simeq Q$ then 
$\cala \vdash P = Q$.
\end{thm}

Note that all the axioms of $\cala$ are actually used in the proof of completeness.
In particular in the proof of completeness over guarded expressions 
(Theorem~\ref{compguarded}) we employ
the standard axioms $(A1)-(A4)$, $(Tau1)-(Tau3)$ and $(Rec1),(Rec2)$ as in~\cite{Mil89},
plus the new axioms $(Pri3)$ and $(Pri6)$. All these axioms are necessary even if
we restrict ourselves to consider completeness over guarded processes only.
Moreover, proving that a process of $\calp$ can always be turned into
a guarded process (Theorem~\ref{unguard}) requires the use of the remaining 
axioms $(Pri1), (Pri2), (Pri4), (Pri5)$ and $(Ung1)-(Ung4)$. 
This supports the claim that our axiomatization is irredundant.

\section{Discrete Time}
\label{SectDRT}

       \begin{table}[t]

\centerline{\small
$\begin{array}{|clcl|}
\hline 
& & & \\[-.3cm]
\infr{P \arrow{\alpha} P'}
{P \pco{S} Q \arrow{\alpha} P' \pco{S} Q} \; \, \ms{\alpha \! \notin \! S} 
\negspace{\; \, \ms{\alpha \! \notin \! S}}
& \tabspace{\; \, \ms{\alpha \! \notin \! S}} \hspace{.4cm} & 
\infr{Q \arrow{\alpha} Q'} 
{P \pco{S} Q \arrow{\alpha} P \pco{S} Q'} 
\; \, \ms{\alpha \! \notin \! S} 
\negspace{\; \, \ms{\alpha \! \notin \! S}}
& \tabspace{\; \, \ms{\alpha \! \notin \! S}}
\\[0.7cm]
\multicolumn{3}{|c}{
\infr{P \arrow{a} P' \, \hspace{.5cm} \, Q \arrow{a} Q'}
{P \pco{S} Q \arrow{a} P' \pco{S} Q'}
\; \, a \in S
\negspace{\; \, a \in S}
} & \\[0.7cm]
\infr{P \arrow{a} P'}{P / L \arrow{\tau} P' / L} 
\hspace{.2cm} a \in L \hspace{.8cm} 
\negspace{\hspace{.2cm} a \in L \hspace{.8cm}}
& & \infr{P \arrow{\alpha} P'}{P / L \arrow{\alpha} P' / L} 
\hspace{.2cm} \alpha \notin L 
\negspace{\hspace{.2cm} a \notin L} 
& 
\\[0.5cm]
\hline
\end{array}$}

\caption{Standard Rules for Parallel Composition and Hiding}\label{TabRulesStandStatic}

    \end{table}

        \begin{table}[t]

\[\begin{array}{|clc|} 
\hline 
&& \\[-.3cm]
\infr{P \arrow{\delta} P' \, \hspace{.5cm} \, Q \arrow{\delta} Q'} 
{P \pco{S} Q \arrow{\delta} P' \pco{S} Q'} 
& \hspace{.2cm} &\infr{P \arrow{\delta} P' \, \hspace{.5cm} \, \notexists a \in L \ldotp
P \arrow{a}}{P / L \arrow{\delta} P' / L}
\\[0.6cm]
\hline 
\end{array}
\]

\caption{Rules for Discrete Time Transitions}\label{TabRulesDRT}

        \end{table}


We now 
interpret unprioritized $\delta$ actions of the basic calculus as representing time delays in the context of {\it discrete time}, see~\cite{HenR}. We first show that we can extend the basic calculus of Section~\ref{SectAPFSB} with {\it static} operators, like CSP~\cite{Hoa} parallel composition and hiding, preserving the congruence property of standard observational congruence.
We then consider a full discrete time calculus and provide a complete axiomatization.


In the following
we will make use of the CSP~\cite{Hoa} parallel composition
operator ``$P \pco{S} Q$'', where standard actions with type in $S$ are required
to synchronize, while the other standard actions are executed independently from $P$ 
and $Q$. Such an operator will be used in combination with a hiding
operator ``$P / L$'', which turns all the standard actions of $P$ whose
type is in $L$ into $\tau$ and does not affect the other standard actions.
See Table~\ref{TabRulesStandStatic} for the standard operational rules of such operators. 
Along the lines of~\cite{HenR},
the behavior of parallel composition concerning ``tick'' transitions
is defined by the operational rule in Table \ref{TabRulesDRT},
which states that time is allowed to elapse in $P \pco{S} Q$ only if {\it both}
processes $P$ and $Q$ may explicitely make it pass via $\delta$
transitions. Moreover, due to the maximal progress assumption,
the generation of $\tau$ actions enacted by the hiding operator
must cause all alternative $\delta$ actions to be pre-empted, as
formalized by the rule in Table \ref{TabRulesDRT}.

\begin{thm}\label{EDRTcongruence1}
$\approx$ and $\simeq$ are congruences with respect to both $\_ \pco{S} \_$ and $\_ /L$ operators.
\end{thm}
\begin{proof}
We will show congruence of $\approx$ for both $\_ \pco{S} \_$ and $\_ /L$
operators and then congruence of $\simeq$.
In the proof we will often exploit the fact that standard transitions
$\alpha$ are inferred by standard transitions only.

Let $P \approx Q$. 
In the following we will use the reformulation of
weak bisimulation introduced in 
Proposition \ref{reformulation}.

$P \pco{S} R \approx Q \pco{S} R$ is shown by considering the weak bisimulation
$$\beta = \{ (P \pco{S} R, Q \pco{S} R) \mid P \approx Q 
\wedge R \in \calp_{static}
\}$$
with $\calp_{static}$ being the set of basic processes extended with $\_ \pco{S} \_$ and $\_ /L$
operators.
The proof for $\alpha$ transitions is standard. Concerning $\delta$ transitions, since $P \approx Q$, we have 
that there exists $Q''$ such that $Q \arrow{\tau}\!\!^* \; Q''$
and for all $P'$ we have: if $P \arrow{\delta} P'$ then, for some $Q'$, 
$Q'' \arrow{\delta} \arrow{\tau}\!\!^* \; Q'$ and $P' \approx Q'$. 
Therefore $Q \pco{S} R \arrow{\tau}\!\!^* \; Q'' \pco{S} R$
and for all $P'$, $R'$
we have: if $P \pco{S} R \arrow{\delta} P' \pco{S} R'$ then (from $R \arrow{\delta} R'$), for the above $Q'$, 
$Q'' \pco{S} R \arrow{\delta} \arrow{\tau}\!\!^* \; Q' \pco{S} R'$ and $(P' \pco{S} R',Q' \pco{S} R') \in \beta$.
%

$P /L \approx Q /L$ is shown by, similarly, considering the weak bisimulation
$$\beta = \{ (P /L, Q /L) \mid P \approx Q \}$$
Under the condition that $P$ does not perform any $a$ transition with $a \in L$ and $P$ performs some $\delta$ transitions (otherwise there is nothing to prove concerning $\delta$ transitions of $P/L$), the proof is performed similar to the case of $\pco{S}$.
We have in addition to rely on the fact that $Q''$ cannot do $a$ transitions
with $a \in L$ because otherwise $P$ (that cannot perform $\tau$ transitions) should perform a corresponding $a$ transition.

Concerning congruence of $\simeq$, we show that $P \simeq Q$ implies $P \pco{S} R \simeq Q \pco{S} R$ by resorting
to the reformulation of Proposition \ref{refcongr}. The proof is as that done for the pair $(P \pco{S} R, Q \pco{S} R)$ in the proof of congruence of $\approx$, with $Q''$ being just $Q$, and by reaching the pair $P' \pco{S} R' \approx Q' \pco{S} R'$.
A similar reasoning is done for $P /L \simeq Q /L$.
\end{proof}


\subsection{Discrete Time Calculus}
\label{SSectDRTC}

        \begin{table}[t]

 \[
\begin{array}{|c|} 
\hline \\[-.2cm]
\begin{array}{l}
\alpha^t. P \arrow{\alpha} P \\[0.3cm]
\end{array} \\
\begin{array}{cccc}
\infr{P \arrow{\alpha} P'}{P +^t Q \arrow{\alpha} P'} & \hspace{1cm} & 
\infr{Q \arrow{\alpha} Q'}{P +^t Q \arrow{\alpha} Q'} & \\[0.6cm]
\end{array} \\
\hline
\end{array}\]

\caption{Standard Rules for Prioritized Action Transitions}\label{TabRulesStand}

        \end{table}

%
%
%
%
%

        \begin{table}[t]

\[\begin{array}{|c|} 
\hline \\[-.3cm]
\snil \arrow{\delta} \snil \hspace{1cm}
\delta^t . P \arrow{\delta} P \hspace{1cm} 
a^t . P \arrow{\delta} a^t . P
\\[0.2cm]
\infr{P \arrow{\delta} P' \, \hspace{.5cm} \, Q \arrow{\delta} Q'}
{P \tch Q \arrow{\delta} P' \tch Q'}
\\[0.5cm]
\hline 
\end{array}
\]

\caption{Rules for Discrete Time Transitions}\label{TabRulesDRT2}

        \end{table}


As we mentioned in the introduction, in the discrete time setting, besides the static operators
$\_ \pco{S} \_$ and $\_ /L$ , it is
convenient (from a
modeling viewpoint, as we detail below) to adopt the specialized ``timed'' 
action prefix $\gamma^t.P$
and ``timed'' choice operator 
$P +^t Q$
used in~\cite{HenR}. Such operators are different from the prefix $\gamma.P$
and choice $P + Q$ that we considered in the basic calculus of Section \ref{SectAPFSB}
(here we use the ``$t$'' superscript to distinguish them)
in that they allow time to evolve via explicit execution of ``$\delta$'' 
transitions as for the semantics of parallel composition:
\begin{itemize}
\item
The new prefix operator $\gamma^t . P$ behaves like $\gamma . P$
plus the additional behaviour defined by 
``$a^t . P \arrow{\delta} a^t . P$''. That is, if $\gamma$ is not $\delta$ or $\tau$, it allows time to elapse via 
the explicit execution of a ``tick'' transition, see Tables \ref {TabRulesStand} and \ref {TabRulesDRT2}.
When specifying systems it is convenient to adopt such a prefix operator,
because it allows visible actions to be arbitrarily 
delayed, so that, e.g., in $a^t . P \pco{\{a\}} Q$, the action $a$ does 
not cause a time deadlock in the case it is not
immediately executable.

\item
The new choice operator $P \tch Q$ behaves like $P + Q$ as far as standard $\alpha$ transitions are concerned, see 
Table \ref {TabRulesStand}. Concering ``tick'' transitions it behaves, instead, 
according to the rule in Table \ref {TabRulesDRT2}. That is, similarly as for parallel composition operator, it allows one of $P$ and $Q$ to 
let 
time pass only if the other one may let time pass and is defined in such a way
that time passage does not resolve the choice (yielding time determinism, see below). 
When specifying systems it is convenient to adopt such a choice operator
because it allows new prefixes $a^t . P$ 
to be 
used without causing the delays $\delta$ preceding the execution of the $a$
to solve the choice.
\end{itemize}
We also use a terminated $\snil$ process that allows time to to elapse ($\nil$ instead deadlocks time).

As explained in the introduction, the idea is that such specialized ``timed'' prefix $\gamma^t.P$
and choice $P +^t Q$ together with $\nil$, $\snil$, 
recursion, parallel composition and hiding form a {\it specification level} calculus (like that of~\cite{HenR} apart from the action synchronization model) that is 
used to specify timed systems. On the contrary,
$\gamma.P$ and $P + Q$ are just used, as auxiliary operators,
to express {\it normal forms} for them, and to produce an axiomatization
by using (a variant of) the machinery presented in Section~\ref{SectAPFSB}.

The set $\cale_{DT}$ of Discrete Time Calculus expressions, ranged over by $E, F, G, \dots$, is 
defined by the following syntax:
\cws{11}{15}{E ::= \nil \mid \snil \mid X \mid \gamma^t. E \mid E \tch E
\mid E / L 
\mid E \pco{S} E \mid \ms{rec}X . E}
where $L, S \subseteq \ms{PAct} - \{ \tau \}$.
The corresponding set of processes, i.e. closed terms, is denoted by $\calp_{DT}$,
ranged over by $P, Q, R, \dots$. 
The operational semantics of the calculus is defined as the least subset of
$\calp_{DT} \times \ms{Act} \times \calp_{DT}$ 
satisfying the standard
operational rules of Tables~\ref{TabRulesStand} and ~\ref{TabRulesStandStatic},
and the timed ones of 
Tables~\ref{TabRulesDRT2} and~\ref{TabRulesDRT}.
In the following we will assume $\gamma^t.E$ prefixes to guard variables as for $\gamma.E$ prefixes in the basic calculus, i.e.\
we say that a variable is guarded by $\gamma$ if it occurs in the scope of a $\gamma^t.E$ prefix.

As far as properties of the generated labeled transition systems are concerned, we first notice that, as for the basic calculus, {\it maximal progress} (Proposition \ref{maxprog}) holds also for $\calp_{DT}$ processes (with a similar inductive proof on the inference tree of $\tau$ transitions, with base cases: $\tau$ transitions inferred by hiding of a visible transition and $\tau$ transitions obtained directly by a $\tau^t$ prefix). Moreover the effect of
adopting such a specification level time calculus, is, like in~\cite{HenR}, 
{\it time determinism}.

\begin{defi}\label{deftimedet}
A process $P$ is called {\it time-deterministic}, if, for all the states in the semantics of $P$, i.e.\ terms reachable by $P$, at most one outgoing $\delta$ transition can be inferred by the operational rules.
\end{defi}

\begin{prop}\label{proptimedeterministic}
Let $P \in \calp_{DT}$. Then $P$ is {\it time-deterministic}. 
\end{prop}
\begin{proof}
We show that, for any  $P \in \calp_{DT}$, we have:
$P \arrow{\delta} P_1 \wedge P \arrow{\delta} P_2$ implies $P_1=P_2$ and such two transitions are obtained with exactly the same inference (tree). 
This is proven
by induction on the maximum of the
inference depth of the two transitions.

We have the following cases depending on
the structure of $P$.

\begin{itemize}
\item If $P \equiv \nil$, $P \equiv \snil$ or $P \equiv\gamma^t. P'$, the condition above obviously holds.

\item If $P \equiv P' / L$, we have that $P' / L \arrow{\delta} P_1 \wedge P' / L \arrow{\delta} P_2$ implies: 
$P_1 = P'_1 / L$, $P_2 = P'_2 / L$ and $P'  \arrow{\delta} P'_1 \wedge P'  \arrow{\delta} P'_2$.  
Therefore, by applying the induction hypothesis to such a pair of transitions (having a strictly smaller maximum inference depth w.r.t.\ that of the pair of transitions considered for $P' / L$), we have: $P'_1 =  P'_2$ and the two transitions are obtained with exactly the same inference. Hence we have $P_1 =  P_2$ and that the inference of the considered $P' / L$ transitions, extending the inference determined by induction, is the same.

\item If $P \equiv \ms{rec}X . E$, we have that $\ms{rec}X . E \arrow{\delta} P_1 \wedge \ms{rec}X . E \arrow{\delta} P_2$ implies: $E \{ \ms{rec}X . E / X \}$ $\arrow{\delta} P_1 \wedge E \{ \ms{rec}X . E / X \} \arrow{\delta} P_2$.
Therefore, by applying the induction hypothesis to such a pair of transitions (having a strictly smaller maximum inference depth w.r.t.\ that of the pair of transitions considered for $\ms{rec}X . E$), we directly have $P_1 =  P_2$ and that the two transitions are obtained with exactly the same inference.

\item If $P \equiv P' \tch P''$, we have that $P' \tch P'' \arrow{\delta} P_1 \wedge P' \tch P'' \arrow{\delta} P_2$ implies:
$P_1 = P'_1 \tch P''_1$, $P_2 = P'_2 \tch P''_2$, 
$P'  \arrow{\delta} P'_1 \wedge P'  \arrow{\delta} P'_2$ and 
$P''  \arrow{\delta} P''_1 \wedge P''  \arrow{\delta} P''_2$.
Therefore, by applying the induction hypothesis to both such pairs of transitions (both having a strictly smaller maximum inference depth w.r.t.\ that of the pair of transitions considered for $P' \tch P''$), we have: $P'_1 =  P'_2$, $P''_1 =  P''_2$ and, 
for each pair, 
it holds that the two transitions are obtained with exactly the same inference. 
Hence we have $P_1 =  P_2$ and that the inference of the considered $P' \tch P''$ transitions, extending the inferences (of
the $P'$ and $P''$ transitions) determined by induction, is the same.

\item If $P \equiv P' \pco{S} P''$, we have that $P' \pco{S} P'' \arrow{\delta} P_1 \wedge P' \pco{S} P'' \arrow{\delta} P_2$ implies
(with exactly the same reasoning, with $\pco{S}$ replacing $\tch$, as that of the previous item): $P_1 =  P_2$ and such two transitions are obtained with the same inference . 
\qedhere
\end{itemize}
\end{proof}

Finally, we notice that, for the specification discrete time calculus, 
fully unguarded recursions cannot perform $\delta$ transitions.

\begin{prop}\label{nosigma}
Let $\ms{rec}X.E \in \calp_{DT}$. If $X$ occurs fully unguarded in $E$ 
then $\ms{rec}X.E \narrow{\delta}$.

\end{prop}
\begin{proof}
Let us suppose, by contradiction, that $\ms{rec}X.E \arrow{\delta}$. 
We now show that, if we remove such a transition from the (timed) transition relation
obtained as the semantics of terms, then the obtained transition relation
still satisfies the operational rules, thus violating minimality.
We just need to observe that, since $X$ occurs fully unguarded in $E \in \cale_{DT}$, any possible inference tree of the $\ms{rec}X.E \arrow{\delta}$ transition necessarily includes 
the transition $\ms{rec}X.E \arrow{\delta}$ itself as a premise. This is because: $(i)$ every operator, apart from $\snil$ 
and prefix, generates a $\delta$ move only if {\it all} of its arguments perform a $\delta$ move, and $(ii)$ we know that $X$ indeed occurs in $E$ and it occurs, 
being fully unguarded, not in the scope of prefix. 
Therefore, if we remove $\ms{rec}X.E \arrow{\delta}$ from the transition relation 
the operational inference rules are still satisfied. 
%
%
%
%
%
\end{proof}

Regarding equivalence, unfortunately,
the introduction of the new choice operator $P \tch Q$
makes standard observational congruence 
(with $\delta$ being considered as a standard visible action)
no longer
a congruence: e.g., $\delta . \tau . a.\nil \simeq \delta. a.\nil$ but
$\delta.b.\nil \; \tch \; \delta . \tau . a.\nil \not\simeq \delta.b.\nil \; \tch \; \delta . a.\nil$. This is because 
$\delta. b. \nil \; \tch \; \delta . \tau . a . \nil$ is isomorphic to $\delta . (b.\nil + \tau .a.\nil)$, while
$\delta. b. \nil \; \tch \; \delta . a . \nil$ is isomorphic to $\delta . (b.\nil + a.\nil)$. The problem is that,
since in $P \tch Q$ the execution of $\delta$ transitions does
not cause the choice to be resolved, in $P$ (as well as in $Q$) it is 
incorrect to abstract from $\tau$ transitions that occur before the 
execution of a standard $a$ action (that cause the choice to be resolved).
In other words, a finer notion of observational congruence, called {\it discrete time observational congruence}, must be considered that is defined similarly as in~\cite{CLM}:
 the ``root'' condition of the
equivalence (where standard transitions are matched as in standard observational congruence) 
can be ``left'' only by executing standard transitions (and not by executing $\delta$ transitions).

\begin{defi}\label{newequiv}
A relation $\beta \subseteq \calp_{DT} \times \calp_{DT}$ is a rooted time weak bisimulation if,
whenever $(P,Q) \in \beta$:

\begin{itemize}
\item If $P \arrow{\alpha} P'$  then, for some $Q'$, $Q \warrow{\alpha} Q'$
and $P' \approx Q'$.
\item If $P \arrow{\delta} P'$ then, for some $Q'$, $Q \arrow{\delta} Q'$
and $(P',Q') \in \beta$.
\item If $Q \arrow{\alpha} Q'$ then, for some $P'$, $P \warrow{\alpha} P'$
and $P' \approx Q'$.
\item If $Q \arrow{\delta} Q'$ then, for some $P'$, $P \arrow{\delta} P'$
and $(P',Q') \in \beta$.
\end{itemize}
Two processes $P$, $Q$ are time observationally congruent,
written $P \simeq_{T} Q$, iff $(P,Q)$ is included in some rooted time
weak bisimulation.
We consider $\simeq_{T}$ as being defined also on open terms by means of free variable substitution, as in Definition \ref{obscongruence}.

\end{defi}


Notice that, even if finer than the observational congruence defined
in 
Definition \ref{basicobscongr}
(where we just considered time delays as standard
actions), time observational congruence is still a {\it conservative
extension} of Milner's observational congruence: for transition systems
without timed transitions it reduces to this equivalence. In particular 
only the root condition is more restrictive: 
the notion of weak bisimulation $\approx$ remains unchanged. 


Due to the time determinism of the 
generated labeled transition systems, it is also possible to simplify the form of weak bisimulation definition:
we can provide reformulations as those presented in Proposition \ref{reformulation} and \ref{necessary}
where 
{\it tails of $\tau$ transitions can be disregarded}
when matching weak timed moves.

\begin{prop}\label{propreftimedet}
In the case a relation $\beta$ over time-deterministic processes is considered, a simplified form of the reformulations of weak bisimulation in Propositions \ref{reformulation} and \ref{necessary} holds true, where 
matching $\sarrow{\delta} \sarrow{\tau}\!\!^*$ moves are replaced by matching $\sarrow{\delta}$ moves.
\end{prop}
\begin{proof}

We now show that $\beta$ satisfying the reformulations in Propositions \ref{reformulation} and \ref{necessary} also satisfies
the simplified form (the converse is obvious). Assumed $(P,Q) \in \beta$, from $P \sarrow{\delta} P'$ and a matching $Q \sarrow{\tau}\!\!^* \, Q''$ such that $Q'' \sarrow{\delta} Q''' \sarrow{\tau}\!\!^* \, Q'$ with $(P',Q') \in \beta$, we show that we also have $(P',Q''') \in \beta$.
Since $Q'' \sarrow{\delta} Q'''$ we have that, due to time determinism, $P \sarrow{\delta} P'  \sarrow{\tau}\!\!^* \, P''$ with 
$(P'',Q''') \in \beta$. Therefore, from $Q''' \sarrow{\tau}\!\!^* \, Q'$ with $(P',Q') \in \beta$ (showing that $Q'''$ can match moves made by $P'$) and $P'  \sarrow{\tau}\!\!^* \, P''$ 
with $(P'',Q''') \in \beta$ (showing that $P'$ can match moves made by $Q'''$), we can conclude that $(P',Q''') \in \beta$.
\end{proof}

We will show in Theorem \ref{EDRTcongruence2} in the context of a larger signature that $\simeq_{T}$ is a congruence with respect to all calculus operators, including recursion.

\subsection{Axiomatizing the Discrete Time Calculus}\label{DRTAxiomat}

Producing a complete axiomatization for the discrete time calculus requires it to be 
extended so to be able to 
express normal forms of processes (terms 
of the basic calculus) and manage and derive them: e.g.\ the
$pri(P)$ operator and left and synchronization merge operators needed to axiomatize parallel composition, along the lines of~\cite{Ac}.
More precisely, we use as normal forms the terms of the basic calculus that are 
time deterministic (Definition \ref{deftimedet}) and get a complete axiomatization via a variant of the machinery presented  in Section \ref{SectAPFSB}.


\subsubsection{Extending the Discrete Time Calculus}\label{MSTExt}

        \begin{table}[t]

\centerline{\small
$\begin{array}{|clcl|} 
\hline 
& & & \\[-.3cm]
\infr{P \arrow{a} P'}{\ms{vis}(P) \arrow{a} P'} 
& &
\infr{P \arrow{\tau} P'' \, \wedge \, vis(P'') \arrow{a} P'}{\ms{vis}(P) \arrow{a} P'} & 
\\[0.7cm]
\infr{P \arrow{\alpha} P'}
{P \lme{S} Q \arrow{\alpha} P' \pco{S} Q} \; \, \ms{\alpha \! \notin \! S} 
\negspace{\; \, \ms{\alpha \! \notin \! S}}
&  \tabspace{\; \, \ms{\alpha \! \notin \! S}} \hspace{.4cm} &
\infr{vis(P) \arrow{a} P' \, \wedge \, vis(Q) \arrow{a} Q'}
{P \sme{S} Q \arrow{a} P' \pco{S} Q'}
\; \, a \in S
\negspace{\; \, a \in S} 
& \hspace{1cm}
\\[0.7cm]
\multicolumn{4}{|c|}{
\infr{P \arrow{\delta} P' \, \wedge \, Q \arrow{\delta} Q'}
{P \sme{S} Q \arrow{\delta} P' \pco{S} Q'}
}
\\[0.5cm]
\hline 
\end{array}$}

\caption{Rules for Auxiliary Operators}\label{TabRulesAux}

        \end{table}

We now formally introduce the auxiliary operators needed to build
the axiomatization, whose semantics is presented in Table~\ref{TabRulesAux}.
The operators 
``$P \lme{S} Q$'' and ``$P \sme{S} Q$'' are timed extensions  
of the left merge and synchronization merge operators of~\cite{Ac}, where 
the definition of the operational rule 
for ``$P \sme{S} Q$'' allows actions $\tau$ to be skipped
so to get a congruence. In particular, here this is expressed by using the novel
operator ``$\ms{vis}(P)$'' that we introduce in this paper: it 
turns $\sarrow{\tau}\!\!^* \sarrow{a}$  weak transitions of $P$ labeled by visible 
(non$-\tau$) standard actions $a$ into strong transitions.
Similarly to the standard setting of~\cite{Ac}, the above auxiliary operators will be used
for axiomatizing parallel composition. In doing this, ``$\ms{vis}(P)$''
will play an important role in that it allows us to check for the absence
of executable delays and/or $\tau$ actions
(see axiom $(SM7)$ of Table~\ref{TimedAxioms}). 
As we will see, differently from the standard setting of~\cite{Ac}, we need here to perform such a 
check because of priority of $\tau$ over $\delta$.

We define the Extended Discrete Time calculus to be the process algebra
obtained by extending the Discrete Time calculus with the auxiliary 
operators above, the operators of the basic calculus of Section \ref{SectAPFSB} and the $pri(\_)$ operator.
The set $\cale_{EDT}$ of behavior expressions,
ranged over by $E, F, G, \dots$, is 
defined by the following syntax:\\[.2cm]
$E ::= \nil \mid \snil \mid X \mid \gamma^t. E \mid E \tch E
\mid E / L 
\mid E \pco{S} E \mid \ms{rec}X . E \mid \\[.1cm]
\tabspace{E ::=} 
\gamma. E \mid E + E \mid pri(E) \mid \\[.1cm]
\tabspace{E ::=} 
\ms{vis}(E) \mid E \lme{S} E \mid E \sme{S} E \\[.2cm]$
where $L, S \subseteq \ms{PAct} - \{ \tau \}$.
The set of processes, i.e.\ closed terms, is 
denoted by $\calp_{EDT}$,
ranged over by $P, Q, R, \dots$. 
The operational semantics of processes is defined as the least subset of
$\calp_{EDT} \times \ms{Act} \times \calp_{EDT}$ 
satisfying all operational rules already presented.

As far as properties of the generated labeled transition system are concerned, we obviously still have that {\it maximal progress} (Proposition \ref{maxprog}) holds. On the contrary, differently from the Discrete Time calculus specification calculus, in general the {\it time-deterministic} property does not hold (as, e.g., in $\delta.P+\delta.Q$).


%
%
%


\begin{thm}\label{EDRTcongruence2}
$\simeq_{T}$ is a congruence for the extended discrete time 
calculus with respect to all of its operators, including recursion.

\begin{proof}
We show congruence of $\simeq_{T}$ with respect to all
the operators of the extended discrete time 
calculus, considering
the recursion operator at the end. 
In the proof we will often exploit the fact that standard transitions
$\alpha$ are inferred by standard transitions only.
Let $P \simeq_{T} Q$. 

$\gamma.P \simeq_{T} \gamma.Q$, in the case $\gamma = \delta$
(otherwise, the proof is trivial), is shown by considering the 
rooted time weak bisimulation 
$\beta=  \{(\delta.P, \delta.Q)\} \, \cup \simeq_{T}$. Similarly, $P + R \simeq_{T} Q + R$ is shown by considering the rooted  time weak bisimulation
$\beta= \{(P+R,Q+R)\} \, \cup \simeq_{T}$.

$pri(P) \simeq_{T} pri(Q)$ is immediate by just considering the
rooted  time weak bisimulation $\{(pri(P),pri(Q))\}$: $pri(P)$ cannot perform $\delta$ transitions and standard $\alpha$ transitions are matched by non-zero length standard weak transitions.
$vis(P) \simeq_{T} vis(Q)$ is similarly shown by considering the 
rooted  time weak bisimulation $\{(vis(P),vis(Q))\}$:
$vis(P)$ can perform neither $\delta$ transitions nor $\tau$ transitions
and each standard $a$ transition is matched by an $a$ transition (the 
transition leading directly to the term $Q''$ such that $Q \arrow{\tau}^*
\arrow{a} Q'' \arrow{\tau}^* Q'$ is the weak transition matching the $\arrow{\tau}^*
\arrow{a}$ transition of $P$) possibly
followed by a sequence of $\tau$ transitions (the sequence $Q'' \arrow{\tau}^* Q'$).

$\gamma^t.P \simeq_{T} \gamma^t.Q$, in the case $\gamma = a$ for some $a \in \ms{PAct}-\{\tau\}$ (otherwise, this case reduces to the previous case), is shown by considering the rooted  time weak bisimulation 
$\{(a^t.P, a^t.Q)\}$.

$P +^t R \simeq_{T} Q +^t R$ is shown by considering the rooted  time weak bisimulation
$$\beta = \{ (P +^t R, Q +^t R) \mid P \simeq_{T} Q 
\wedge R \in \calp_{EDT}
\}$$
The proof for $\alpha$ transitions is trivial. 
Concerning $\delta$ transitions, since $P \simeq_{T} Q$, we have 
that for all $P'$: if $P \arrow{\delta} P'$ then, for some $Q'$, 
$Q \arrow{\delta} Q'$ and $P' \simeq_{T} Q'$. 
Therefore for all $P'$, $R'$
we have: if $P +^t R \arrow{\delta} P' +^t R'$ then (from $R \arrow{\delta} R'$), for the above $Q'$, 
$Q +^t R \arrow{\delta} Q' +^t R'$ and $(P' +^t R',Q' +^t R') \in \beta$.


$P \pco{S} R \simeq_{T} Q \pco{S} R$ and $P /L \simeq_{T} Q /L$
are shown as for the proof of congruence of $\approx$ with the
simplification that (as for ``$+^t$'' above) the possibility of silently
reaching a $Q''$ state is not considered and by resorting to 
congruence of $\approx$ in the case of standard $\alpha$ moves.

%
%

$P \lme{S} R \simeq_{T} Q \lme{S} R$ is immediate by just considering the
rooted  time weak bisimulation $\{(P \lme{S} R,Q \lme{S} R)\}$: $P \lme{S} R$ cannot perform $\delta$ transitions and standard $\alpha$ transitions not in $S$ (that must originate in $P$) are matched by non-zero length standard weak transitions not in $S$ (congruence of $\approx$ with respect to $\pco{S}$ is then exploited).

%
%
$P \sme{S} R \simeq_{T} Q \sme{S} R$ is shown by considering the
rooted  time weak bisimulation 
$\beta= \{(P \sme{S} R,Q \sme{S} R)\} \, \cup \simeq_{T}$.
Concerning standard transitions, $P \sme{S} R$ can just perform 
(non-$\tau$) $a$ transitions with $a \in S$: each $a$ transition originated
from $P$ is matched (as for the $vis()$ operator) by an $a$ transition (the 
transition leading directly to the term $Q''$ such that $Q \arrow{\tau}^*
\arrow{a} Q'' \arrow{\tau}^* Q'$ is the weak transition matching the $\arrow{\tau}^* \arrow{a}$ transition of $P$) possibly
followed by a sequence of $\tau$ transitions (the sequence $Q'' \arrow{\tau}^* Q'$); congruence of $\approx$ with respect to ``$\pco{S}$'' is then exploited on reached terms.
Concerning $\delta$ transitions, since $P \simeq_{T} Q$, we have that,
for all $P'$, $R'$: if $P \sme{S} R \arrow{\delta} P' \pco{S} R'$ then (from $R \arrow{\delta} R'$), there exists $Q'$ such that
$Q \sme{S} R \arrow{\delta} Q' \pco{S} R'$ and $P' \pco{S} R' \simeq_{T} Q' \pco{S} R'$ because of congruence of $\simeq_{T}$ with respect to ``$\pco{S}$''.

Concerning the recursion operator, from $E \simeq_{T} F$ we derive 
$\ms{rec}X.E \simeq_{T} \ms{rec}X.F$ as follows. We show that
\cws{13}{13}{
\beta = \{ (G \{ \ms{rec}X.E / X \}, G \{ \ms{rec}X.F / X \}) \mid 
G \; {\rm contains \; at \; most} \; X \; {\rm free} \} }
satisfies the conditions: 
\begin{itemize}
\item
 If $G \{ \ms{rec}X.E / X \} \arrow{\alpha} H$ then, 
for some $H',H''$, \\[.1cm]
\hspace*{1.5cm}
$G \{ \ms{rec}X.F / X \} \warrow{\alpha} H''$ with $H'' \approx H'$ such that $(H,H') \in \beta$, \\[.1cm]
and symmetrically for a move of $G \{ \ms{rec}X.F / X \}$.
\item If $G \{ \ms{rec}X.E / X \} \arrow{\delta} H$ then, 
for some $H',H''$, \\[.1cm]
\hspace*{1.5cm}
$G \{ \ms{rec}X.F / X \} \arrow{\delta} H''$ with $H'' \simeq_{T} H'$ such that $(H,H') \in \beta$, \\[.1cm]
and symmetrically for a move of $G \{ \ms{rec}X.F / X \}$. 
\end{itemize}
This implies that $\beta$ is a weak bisimulation up to $\approx$, see \cite{SM}. 
As a consequence (by using $\beta \subseteq \; \approx$ in the first item)
the relation $\beta \simeq_{T}$ is a 
rooted discrete time weak bisimulation, hence $\beta \subseteq
\; \simeq_{T}$. Thus,
by taking $G \equiv X$ we may conclude that $\ms{rec}X.E \simeq_{T} \ms{rec}X.F$.

The proof that $(G \{ \ms{rec}X.E / X \}, 
G \{ \ms{rec}X.F / X \})
\in \beta$ satisfies the conditions above is by induction on
the height of the inference tree
by which transitions of $G \{ \ms{rec}X.E / X \} \arrow{\gamma} H$ 
are inferred, first for standard $\alpha$ transitions (which are inferred from
standard transitions only) and then for $\delta$ transitions.
In both cases the induction proof is performed by considering several cases depending on the structure of $G$: its topmost operator.
For $\alpha$ transitions 
the induction hypothesis is exploited in a similar
way as we did for the proof of congruence of $\simeq_{T}$ for that operator.
Thus, here we just detail the cases for $\delta$ transitions.

$G$ cannot be of the form $pri(G')$, $vis(G')$ or $G' \lme{S} G''$ because
in these cases it would not be able to perform $\delta$ transitions.
If $G$ is of the form $\gamma. G'$, $\gamma^t. G'$, $G' + G''$, then the proof is trivial.  

If $G \equiv G' \pco{S} G''$ then, from $(G' \pco{S} G'')\{ \ms{rec}X.E / X \} \arrow{\delta} H_1 \pco{S} H_2$, we have, from the induction
hypothesis on $G'$ and $G''$, that $(G' \pco{S} G'')\{ \ms{rec}X.F / X \} \arrow{\delta} H''_1\pco{S} H''_2$
with $H''_1\pco{S} H''_2 \linebreak
\simeq_{T} H'_1\pco{S} H'_2$ (due to congruence of $\simeq_{T}$ with respect to ``$\pco{S}$'') such that $(H_1 \pco{S} H_2,H'_1 \pco{S} H'_2) \in \beta$.

If $G \equiv G' / L$ then we must have that $G'\{ \ms{rec}X.E / X \}$ does not perform any $a$ transition with $a \in L$, otherwise it would not be
possible for $(G'/L)\{ \ms{rec}X.E / X \}$ to perform a $\delta$ transition to $H/L$.
Since $G'\{ \ms{rec}X.F / X \}$ cannot do $a$ transitions
with $a \in L$ because otherwise $G'\{ \ms{rec}X.E / X \}$ (that cannot perform $\tau$ transitions) should perform a corresponding $a$ transition, we have, from the induction hypothesis on $G'$, that $(G' / L)\{ \ms{rec}X.F / X \} \arrow{\delta} H''/L$ with $H''/L \simeq_{T} H'/L$ (due to congruence of $\simeq_{T}$ with respect to ``$/ L$'') such that $(H/L,H'/L) \in \beta$.

If $G$ is of the form $G' \tch G''$ or $G' \sme{S} G''$ then the proof is
carried out like in the case of $\pco{S}$. The cases 
$G \equiv \ms{rec}Y.G'$ and $G \equiv X$ 
are dealt with as in the standard way.
\end{proof}
\end{thm}

\subsubsection{Axiom System}\label{DRTAxioms}

        { \begin{table}[t]

{\[\begin{array}{|l|}
\hline
\begin{array}{lrcll}
(Tau1') \; & \alpha . \tau . E & = & \alpha . E &\\[0cm]
(Tau3') & \alpha . ( E + \tau . F ) + \alpha . F & = & \alpha . ( E + \tau . F ) &\\[0cm]
(Tau4) & \alpha . F \{ \delta. \tau . E / X \} & = & \alpha . F \{ \delta . E / X \} &
\; \rm{provided\;that}\;X\;\rm{is\;serial\;in}\; F\\[.0cm]
\end{array} 
\\ 
\hline
\end{array}\]}

\caption{Axioms of $\cala_{DT}$ for unguarded basic terms}\label{BasicAxioms}

        \end{table}}

        { \begin{table}[t]

{\[\begin{array}{|l|}
\hline
\begin{array}{lrcl}
(Ter) \;\;\, & \snil & = & \ms{rec}X . \delta. X  \\
\end{array} \\ 
\hline
\begin{array}{lrcll}
(TPre1) \;\;\, & \alpha^t . E & = & \ms{rec}X . (\delta . X +  \alpha . E) &
\hspace{.3cm}
\; \rm{provided\;that}\;X\;\rm{is\;not\;free\;in}\; E
\\[0cm]
(TPre2) & \delta^t . E & = & \delta.E \\
\end{array} \\ 
\hline 
\begin{array}{lrcll}
(TCh1) \; & E \tch F & = & F \tch E & \\[0cm]
(TCh2) & (E +^t F) +^t G & = & E +^t (F +^t G) \\[0cm]
(TCh3) \; & E \tch \snil & = & E & \\[0cm]
(TCh4) & pri(E \tch F) & = & pri(E) + pri(F) & \\[0cm]
(TCh5) & pri(E) \tch F & = & pri(E) + pri(F) & \\[0cm]
(TCh6) & (\delta . E) \tch (\delta . F) & = & \delta . ( E \tch
F) &  \\[0cm]
(TCh7) & (E \!+ \!F) \tch G & = & (E \tch G) + (F \tch G)  
& \\
\end{array} \\ 
\hline 
\end{array}\]}
\caption{Axioms of $\cala_{DT}$ for timed operators}\label{TimedAxioms}

        \end{table}}

        { \begin{table}[t]

{\[\begin{array}{|l|}
\hline
\begin{array}{lrcll}
(Hi1) \;\;\, & \nil / L & = & \nil & \\[0cm]
(Hi2) & (\gamma . E) / L & = & \gamma . (E / L) & \gamma \notin L   \\[0cm]
(Hi3) & (a . E) / L & = & \tau . (E / L) \; \; \; & a \in L  \\[0cm]
(Hi4) & (E + F) / L & = & E / L + F / L \; & \\
\end{array} \\ 
\hline
\begin{array}{lrcll}
(RecHi) \; \; \, & (\ms{rec}X . E) / L & = & \ms{rec}X . (E/L) \hspace{.3cm}& 
\; \rm{provided\;that}\;X\;\rm{is\;serial\;in}\; E
\\[.1cm]
\end{array} \\ 
\hline 
\begin{array}{lrcl}
(Vis1) \;\;\;\, & \ms{vis}(\nil) & = & \nil \\[0cm]
(Vis2) \;\;\;\, & \ms{vis}(a.E) & = & a.E \\[0cm]
(Vis3) \;\;\;\, & \ms{vis}(E+F) & = & \ms{vis}(E) + \ms{vis}(F) \\
\end{array} \\ 
\hline 
\begin{array}{lrcl}
(Par) \;\;\, & E \pco{S} F & = & E \lme{S} F + 
F \lme{S} E + E \sme{S} F \\
\end{array} \\ 
\hline 
\begin{array}{lrcll}
(LM1) \; & \nil \lme{S} E & = & \nil & \\[0cm]
(LM2) & (\gamma . E) \lme{S} F & = & \nil & \gamma \in S \cup \{ \delta \}
\\[0cm]
(LM3) & (\alpha . E) \lme{S} F & = & \alpha . ( E \pco{S} F) & \alpha \notin S \\[0cm]
(LM4) & (E + F) \lme{S} G & = & E \lme{S} G + F \lme{S} G & \\
\end{array} \\ 
\hline 
\begin{array}{lrcll}
(SM1) \; & E \sme{S} F & = & F \sme{S} E & \\[0cm]
(SM2) & \nil \sme{S} E & = & \nil & \\[0cm]
(SM3) & (\gamma . E) \sme{S} (\gamma' . F) & = & \nil & \hspace{-1.5cm}
(\gamma \notin S \cup \{ \delta \} \; \vee \; \gamma \neq \gamma') \wedge \tau \notin \{
\gamma, \gamma' \} \\[0cm]
(SM4) & (\tau . E) \sme{S} F & = & pri(E \sme{S} F) \\[0cm]
(SM5) & (\gamma . E) \sme{S} (\gamma . F) & = & \gamma . ( E \pco{S}
F) \hspace{.3cm} & \gamma \in S \cup \{ \delta \} \\[0cm]
(SM6) & (pri(E) \! + \! pri(F)) \! \sme{S} G  & = & pri(E) \sme{S} G + pri(F) \sme{S} G  \hspace{-5cm} & \\[0cm]
(SM7) & (\delta.E \! + \! vis(F)) \! \sme{S} G  & = & \delta.E \sme{S} G + vis(F) \sme{S} G  
\hspace{-5cm} & \\
\end{array} \\
\hline
\end{array}\]}

\caption{Axioms of $\cala_{DT}$ for parallel composition, hiding and related auxiliary operators}\label{StaticAxioms}

        \end{table}}


We now present the axiom system $\cala_{DT}$. The idea is that the axiom system
must be able: $(i)$ to turn $\calp_{DT}$ terms of the specification calculus into normal form, i.e.\ basic processes in $\calp$ that are 
time-deterministic, and
$(ii)$ to equate normal forms when they are equivalent according to $\simeq_{T}$.
 
$\cala_{DT}$ is composed of: 
\begin{itemize}
\item the axioms 
in Tables~\ref{Axioms} and~\ref{BasicAxioms}, related to axiomatizing 
basic processes, 
where the axioms in Table~\ref{BasicAxioms} that are primed
replace the corresponding axioms in Table~\ref{Axioms} (reflecting the modification in the equivalence, i.e.\ we are axiomatizing $\simeq_{T}$ instead of $\simeq$);
\item the axioms in Table~\ref{TimedAxioms} related to axiomatizing timed operators (make it possible to eliminate them, so to obtain normal forms); and 
\item the axioms in Table~\ref{StaticAxioms} related to axiomatizing parallel composition and hiding operators (eliminating such operators via auxiliary operators and dynamically generated unguardedness via the $(RecHi)$ axiom taken from~\cite{concur05,mscs08}, so to obtain normal forms).

\end{itemize}



  

Concerning Table~\ref{BasicAxioms}, due 
to the more restrictive root of equivalence,
axioms $(Tau1)$ and $(Tau3)$ are restricted to $\gamma = \alpha \in \ms{PAct}$, instead of a general $\gamma$ action including $\sigma$, and a specific $\tau$ elimination axiom $(Tau4)$ for time is added, 
i.e.\ a restricted version of old axiom $(Tau1)$ in the case $\gamma = \sigma$ that now we have excluded. Notice
that variable replacement in axiom $(Tau4)$ is needed in order to deal with the case that $\delta.\tau.E$ occurs inside a recursion
in $F$. 
Moreover the ``serial'' condition in $(Tau4)$ is needed because, e.g., $\alpha. (\delta.\tau.E +^t G)$ is not equivalent to $\alpha. (\delta.E +^t G)$.

Concerning Table~\ref{TimedAxioms}, we just include axioms for timed operators 
that are actually needed to prove our completeness result for the discrete time calculus. As a matter of fact, we 
could have considered other axioms as, e.g., $\ms{rec}X . (X +^t E)  = \ms{rec}X .pri(E)$, which allows 
fully unguarded recursion to be removed in the case of timed choice (such an axiom is sound in that $\ms{rec}X . (X +^t E)$
cannot perform $\delta$ transitions, as it can be shown with a proof similar to that of Proposition \ref{nosigma}).

Concerning Table~\ref{StaticAxioms}, the axioms are standard, apart from the usage of the auxiliary 
operator $\ms{vis}(\_)$ to deal with distributivity of synchronization merge:  we cannot just distribute with 
the standard axiom 
$(E + F) \sme{S} G = E \sme{S} G + F \sme{S} G$ 
because, e.g., $(\tau.E + \delta.F) \sme{S} G$ is not equivalent to $\tau.E \sme{S} G + \delta.F \sme{S} G$, due
to priority of $\tau$ over $\delta$.
%
Moreover, as explained in the introduction, we make use of the axiom $(RecHi)$ 
taken from~\cite{concur05,mscs08} for eliminating unguardedness dynamically generated by the hiding operator. Notice that hiding here also has a prioritization effect: this makes it essential to resort to the axiomatization presented in Section 
\ref{SectAPFSB} that enacts priority of $\tau$ over $\delta$. Also in the case of 
Table~\ref{StaticAxioms} we just include axioms that we need in order to prove our completeness result. For this reason auxiliary operator axioms such as, e.g.,  $\ms{vis}(\tau.E)=\ms{vis}(E)$ and $\ms{vis}(\delta.E)=\nil$, are omitted.


The axiom system is sound with respect to $\simeq_{T}$ extended to open terms.

\begin{thm}~\label{soundness2}
Given $E, F \in \cale_{EDT}$, if $\cala_{DT} \vdash E = F$ then $E \simeq_{T} F$.
\end{thm}

\begin{proof}
The soundness of the new axiom $(Tau4)$ is proved by just showing that\\[.1cm]
$\beta \!=\! \{ (G\{\delta.\tau.E / X\}, G\{\delta.E / X\}) \! \mid \! 
G \; {\rm contains \; at \; most} \; X \; {\rm free \; and} \; X \; {\rm is \; serial} 
{\rm \; in \; G} \} \cup \{ (\tau.E, E) \}$\\[.1cm]
is a weak bisimulation. 
This holds because, for corresponding transitions, 
$G\{\delta.\tau.E / X\}$ and $G\{\delta.E / X\}$, with $X$ serial in $G$,
either reach related terms $(G'\{\delta.\tau.E / X\}$ and $G'\{\delta.E / X\})$, respectively,
for some $G'$ such that $G'$ contains  at  most $X$ free  and $X$ is  serial  in $G'$,
or reach related terms $(\tau.E, E)$.

The soundness of the $(RecHi)$ axiom is shown as in~\cite{mscs08}. 
The proof of soundness for the other axioms is standard, see
~\cite{Ac}. In particular, concerning axiom $(Par)$, notice that the semantics of synchronization merge  ``$\sme{S}$'' follows a standard approach, in that: inferring $\sarrow{a}$ transitions using the new $\ms{vis}(\_)$ operator is like directly inferring them from $\sarrow{\tau}\!\!^* \sarrow{a}$ transitions.
\end{proof}


\subsubsection{Completeness for Time-Deterministic Basic Processes}\label{DRTBCompl}

%

Completeness of the $\cala_{DT}$ axiomatization is proven by resorting to equation sets. In particular, we have to introduce the  subclass of {\it time-deterministic} ones.

\begin{defi}
Let $S = \{ X_i = H_i \mid 1 \leq i \leq n \}$, with formal variables 
$\tilde{X} = \{ X_1, \dots , X_n \}$, where $X_1$ is the distinguished variable of $S$, and free variables
$\tilde{W} = \{ W_1, \dots , W_m \}$ be a standard equation set and let us suppose each expression $H_i$ ($1 \leq i \leq n$) to be denoted by
\cws{0}{0}{
H_i \equiv \sum_{j \in J_i} \gamma_{i,j} . X_{f(i,j)} + \sum_{k \in K_i} 
W_{g(i,k)}.
}
$S$ is {\it time-deterministic} if it holds that, for all $i$, we have
\cws{10}{10}{
\forall j,j' \!\in\! J_i \ldotp \gamma_{i,j} = \gamma_{i,j'} = \delta \; \Rightarrow \; \, j = j'
}
$S$ is {\it well-rooted} if it holds that, for all $i$, we have
\begin{itemize}
\item If $X_i \in \tilde{X}_R$ then $\forall j \in J_i \ldotp \gamma_{i,j} \neq \delta
\Rightarrow X_{f(i,j)} \notin \tilde{X}_R$
\item If $X_i \notin \tilde{X}_R$ then $\forall j \in J_i \ldotp X_{f(i,j)} \notin \tilde{X}_R$
\end{itemize}
where $\tilde{X}_R$ is the set of {\it root variables} of $S$
defined by $\tilde{X}_R = \{ X_i \mid X_1 \sarrow{\delta}\!\!\!_S^{\,\,*} X_i \}$.

\end{defi}

\begin{prop}\label{rootprop}
Let expression $E \in \cale$ provably satisfy a standard equation set
$S$. Then there exists a well-rooted standard equation set $S'$ provably satisfied by $E$.
Moreover, if $S$ is time-deterministic, prioritized, guarded and closed, 
then $S'$ is time-deterministic, prioritized, guarded and closed.
\end{prop}

\begin{proof}
We first consider new variables $X'_i$, one for each variable $X_i \in \tilde{X}_R$, and new equations $X'_i = H'_i$ where $H'_i$ is obtained from the term $H_i$ such that $X_i = H_i$ by replacing each occurrence of a variable $X_j \in \tilde{X}_R$ with $X'_j$. The idea is that, once the root is left, $X'_i$ variables are used to represent the same behavior as that of $X_i$ variables.
Then we modify: the equations $X_i = H_i$ for the variables $X_i \in \tilde{X}_R$ by replacing each occurrence of $\gamma.X_j$ in $H_i$, for any $\gamma \neq \delta$ 
and $X_j \in \tilde{X}_R$, with $\gamma.X'_j$; the equations $X_i = H_i$ for the variables $X_i \notin \tilde{X}_R$ by replacing each occurrence of a variable $X_j \in \tilde{X}_R$ in $H_i$ with $X'_j$. It is immediate to verify that the new equation system has the same root variables 
$\tilde{X}_R$ as the original one (because $\delta$ prefixed variables in the equations
of root variables are not modified) and that, by construction, the two statements in the
definition of ``well-rooted'' standard equation systems hold true.
\end{proof}

Theorems \ref{repr} and \ref{onesol} concerning representability and one and only one solution still hold with exactly the same proofs (the change from $\simeq$ to $\simeq_{T}$ does not affect the transforming expressions into equation sets and vice-versa). Here we just have to additionally show 
that the standard solution of a time-deterministic guarded equation set is time-deterministic and that any time-deterministic guarded basic process can be represented by a time-deterministic equation set. 

\begin{thm}[time-deterministic solution]\label{soldet}
If $S$ is a time-deterministic, prioritized, standard, guarded and closed equation set,
then there is a time-deterministic guarded process $P \in \calp$ 
which provably satisfies $S$. 
\end{thm}

\begin{proof}
The inductive procedure for building $P$ is that presented in the proof of Theorem \ref{onesol}, in turn taken from~\cite{Mil89}. In addition here we show that such $P$ is time-deterministic.
We first observe that, since $S$ is standard and closed, 
the equation of any of its variables $X$ is of the form $X = \sum_{j \in J} \gamma_j . X_j$, for some variables $X_j$, with $j \in J$, all having a defining equation in $S$.

We now show that every process $P'$ reachable from $P$ is an expansion of some variable $X$ of $S$: in short, an $X$-expansion. In general an expression $E$ is an $X$-expansion if the following holds: if $X = \sum_{j \in J} \gamma_j . X_j$ is the defining equation for $X$ in $S$, $E \equiv \ms{rec}X. \sum_{j \in J} \gamma_j . E_j$ for some expressions $\{ E_j \mid j \in J\}$ such that, for each $j \in J$, either $E_j \equiv X_j$ or $E_j$ is, itself, an $X_j$-expansion.
It is immediate to show that the inductive procedure for building $P$ yields an $X$-expansion, with 
$X$ being the distinguished (first) variable of $S$. Moreover, 
taken any process $P'$ that is an $X$-expansion for some $X$ in $S$, with $X = \sum_{j \in J} \gamma_j . X_j \in S$ and
 $P' \equiv \ms{rec}X. \sum_{j \in J} \gamma_j . E_j$, we have that process $P''$ reached from $P'$ with the $\gamma_j$ transition
is an $X_j$-expansion.
%
%
%
%
%

Therefore, since $S$ is time-deterministic, then $P$ is time-deterministic.
\end{proof}

\begin{thm}[time-deterministic representability]\label{reprtimedet}
Every time-deterministic guarded process $P \in \calp$ provably satisfies a 
time-deterministic, prioritized, standard, guarded and closed equation set $S$.
\end{thm}

\begin{proof}
%
Let $P_1 \dots P_n$ be the states of the transition system of $P \equiv P_1$. 
By applying axiom $(Rec1)$ and the $(Pri)$ axioms, for each 
$i \in \{ 1 \dots n \}$, there exist $m_i$, $\{\gamma^i_j\}_{j \leq m_i}$
(denoting actions), 
$\{k^i_j\}_{j \leq m_i}$ (denoting natural numbers) s.t.\ we can derive 
$P_i = \sum_{j \leq m_i} \gamma^i_j . P_{k^i_j} $,
where ``$\arrow{\gamma^i_j} P_{k^i_j}$'', with $j \leq m_i$, are the outgoing transitions of $P_i$ 
(no outgoing transitions corresponds to the sum being $\nil$).
Hence we can characterize the behavior of
$P_1$ by means of a time-deterministic, prioritized, standard, guarded and closed equation set  
$S = \{ X_i = H_i \mid 1 \leq i \leq n \}$ where 
$H_i \equiv \sum_{j \leq m_i} \gamma^i_j . X_{k^i_j} $ and
$P_1 \dots P_n$ are a solution of the equation set. 
Such an equation set is 
guarded because the arguments of
the sums in the equations are the outgoing transitions of the states of $P$ and
$P$ is a guarded basic processes,
hence every cycle in its transition
system contains at least one non-$\tau$ action.
\end{proof}

\begin{defi}
An $\alpha$-{\it saturated} standard equation set $S$ with formal variables 
$\tilde{X}$ is a standard equation set such that, for all $X \in \tilde{X}$, items $(i)$ and $(iii)$ of Definition \ref{satur} hold.
 \end{defi}

\begin{lem}\label{satlemma2}
Let expression $E \in \cale$ provably satisfy $S$, time-deterministic, prioritized, standard and guarded. Then there is an $\alpha$-saturated time-deterministic, prioritized, standard and guarded equation set $S'$
provably satisfied by $E$. Moreover $S'$ is well-rooted and closed if $S$ is well-rooted and closed. 

\begin{proof}
Since the saturation involves only standard $\alpha$ actions, 
we can derive $S'$ by following the same procedure as that in~\cite{Mil89} 
to saturate $X \sarrow{\tau}\!\!\!_S^{\,\,*} \; \arrowd{\alpha}{S} \, 
\sarrow{\tau}\!\!\!_S^{\,\,*} \; X'$ in standard equation sets. Hence, the procedure involves 
the axioms $(A1)-(A4)$ and $(Tau1'),(Tau2),(Tau3')$, i.e. the axioms that correspond
to standard axioms when just dealing with (saturation of) standard actions, plus the $(Pri)$ axioms that are needed in order to remove unwanted instances of $\delta$ prefixes produced by using $(Tau2)$.
\end{proof}

\end{lem}

We now introduce the novel concept of $\tau$-{\it saturation} of (time-deterministic, prioritized, standard and closed) equation sets $S$. $\tau$-{\it saturation} transforms $\delta.X$ terms occurring in the definition of a non-root variable of $S$ into $\delta.\tau.X$ terms. Such a transformation preserves the equation set solutions due to axiom $(Tau4)$, which makes it possible to apply it inside recursions. $\tau$-{\it saturation} will play a fundamental role in the equation set mergeability Theorem \ref{merge} in that it allows us to eliminate, via axiom $(Tau1')$, $\tau$ prefixes occurring inside terms used to replace $X$ in $\delta.\tau.X$.

\begin{defi}
Let $S = \{ X_i = H_i \mid 1 \leq i \leq n \}$ be a 
time-deterministic, prioritized, standard and closed equation set  with formal variables 
$\tilde{X} = \{ X_1, \dots , X_n \}$. Moreover
assume expressions 
$H_i$ ($1 \leq i \leq n$) to be denoted by:
\cws{0}{0}{
H_i \equiv \sum_{j \in J_i} \gamma_{i,j} . X_{f(i,j)}
}

$\tau$-{\it saturation} of $S$  
yields the equation set $S' = \{ X_i = H'_i \mid 1 \leq i \leq n \}$
with formal variables $\tilde{X}$ and with
\cws{0}{0}{
H'_i \equiv \sum_{j \in J_i} \gamma_{i,j} . G_{i,j}
}
where: $G_{i,j} \equiv \tau.X_{f(i,j)}$ if $X_i \notin \tilde{X}_R$ and $\gamma_{i,j} = \delta$ ; $G_{i,j} \equiv X_{f(i,j)}$ otherwise.
 
\end{defi}

\begin{lem}\label{tausatlemma}
Let process $P \in \calp$ provably satisfy 
$S$ 
time-deterministic, prioritized, standard, guarded and closed.
Then $P$ provably satisfies $S'$ obtained by $\tau$-saturating $S$.

\begin{proof}
We assume, without loss of generality, that the root variables of $S$ and $S'$ are the first ones
in the index variable ordering (if that is not the case we simply re-order variables in the same
way inside both $S$ and $S'$ and, obviously, since the initial variable is not involved in such
a re-ordering, any term that provably satisfies an equation set still satisfies a re-ordered one
and vice-versa). 

Let $Q$ be the standard solution of $S$ and $Q'$ be the standard solution of $S'$.
Since $S'$ is obtained by $\tau$-saturating $S$, $S$ and $S'$ just differ for the fact that, in the definition
of non-root variables, $\delta.R$ summands are replaced by $\delta.\tau.R$ summands. As a consequence,
according to the inductive procedure in~\cite{Mil} for deriving the standard solution from a guarded standard equation set (here reported in the proof of Theorem~\ref{onesol}),
we have that $Q$ and $Q'$ just differ for the fact that: in subterms $\ms{rec}X.E$, where $X$ is a non-root variable of $S$, $\delta.R$ summands of the sum $E$ are replaced by $\delta.\tau.R$ summands. It is easy to see, by contradiction, that such subterms $\ms{rec}X.E$ must be in the scope
of an $\alpha$ prefix operator. If that was not the case, i.e.\ if $\ms{rec}X.E$ was only in the scope of $\delta$ prefixes (or not in the scope of any prefix), it would imply that $X$ is a root variable: since $Q$, $Q'$ are built just by variable replacement, in $S$, $S'$ variable $X$ would be reachable from the initial variable by performing $\delta$ steps only (traversing the variables $Y$ such that $\ms{rec}X.E$ is in the scope of $\ms{rec}Y.\_$, in the outside-inside order). 

Therefore, by multiple applications of the axiom $(Tau4)$ we have that $Q = Q'$ and, since $P=Q$, for unique solution of guarded equation sets, we are finished.
\end{proof}
\end{lem}



\begin{thm}[mergeability]\label{merge}
Let 
process $P \in \calp$ provably satisfy $S$, and 
process $Q \in \calp$ provably satisfy
$T$, where both $S$ and $T$ are  
time-deterministic, prioritized, standard, guarded and closed sets of equations, and
let $P \simeq Q$. 
Then there is a 
time-deterministic, prioritized, standard, guarded and closed equation set $U$ 
provably satisfied by both $P$ and $Q$.
\end{thm}

\begin{proof}
We may suppose that $S$ is the
time-deterministic prioritized standard guarded and closed equation set
$S : \tilde{X} = \tilde{H}$ and $T$ is the 
time-deterministic prioritized standard guarded  and closed equation set $T : \tilde{Y} = \tilde{J}$ where $\tilde{X} = \{ X_1, \dots , X_m \}$ and
$\tilde{Y} = \{ Y_1, \dots , Y_n \}$ are disjoint sets of formal variables.
We can assume that both 
$S$ and $T$ are both well-rooted and $\alpha$-saturated because of Proposition~\ref{rootprop} and  Lemma~\ref{satlemma2}.

Since $P \simeq Q$ and $S$ and $T$ are $\alpha$-saturated, from Propositions~\ref{necessary}
and~\ref{propreftimedet}
we derive that there exists a relation $\beta \subseteq \tilde{X} \times \tilde{Y}$ 
such that:
\begin{enumerate}
\item[1.] Whenever $(X,Y) \in \beta$:

\begin{enumerate}
\item[(i)] If $X \arrowd{\alpha}{S} X'$ then, either $(A)$ $\alpha = \tau$ and $(X',Y) \in \beta$,
or $(B)$ for some $Y'$, $Y \arrowd{\alpha}{T} Y'$
and $(X',Y') \in \beta$.
\item[(ii)] If $X \arrowd{\delta}{S} X'$ then, either $(A)$ $Y \arrowd{\tau}{T}$, or $(B)$
for some $Y'$, $Y \arrowd{\delta}{T} Y'$
and $(X',Y') \in \beta$.
\item[(iii)] If $Y \arrowd{\alpha}{T} Y'$ then, either $(A)$ $\alpha = \tau$ and $(X,Y') \in \beta$,
or $(B)$ for some $X'$, $X \arrowd{\alpha}{S} X'$
and $(X',Y') \in \beta$.
\item[(iv)] If $Y \arrowd{\delta}{T} Y'$ then, either $(A)$ $X \arrowd{\tau}{S}$, or $(B)$
for some $X'$, $X \arrowd{\delta}{S} X'$
and $(X',Y') \in \beta$.
\end{enumerate}

\item[2.] $(X_1,Y_1) \in \beta$, and when $(X,Y) \in \tilde{X}_R \times \tilde{Y}_R$
then cases (i)$(A)$,
(ii)$(A)$,(iii)$(A)$ and (iv)$(A)$ do not apply.

\end{enumerate}

Notice that, due to condition $2$ and to the fact that $S$ and $T$ are both well-rooted, we 
can assume $\beta$ to be such that
$\beta \subseteq \tilde{X}_R \times \tilde{Y}_R \cup(\tilde{X}-\tilde{X}_R) \times (\tilde{Y}-\tilde{Y}_R)$.


We now build the new equation set $U : \tilde{Z} = \tilde{K}$, where 
$\tilde{Z}$ is the set of variables $\tilde{Z} = \{ Z_{i,j} \mid (X_i,Y_j) \in \beta \}$
and $\tilde{K}$ is the set of expressions $\tilde{K} = \{ K_{i,j} \mid (X_i,Y_j) \in \beta \}$
defined as follows. Each expression $K_{i,j}$ is a sum containing the terms:

\begin{enumerate}
\item[(i)] $\alpha . Z_{k,l}$, whenever $X_i \arrowd{\alpha}{S} X_k$ and $Y_j \arrowd{\alpha}{T} Y_l$ and
$(X_k,Y_l) \in \beta$.
\item[(ii)] $\tau . Z_{k,j}$, whenever $X_i \arrowd{\tau}{S} X_k$ and there is no $l$ 
such that: $Y_j \arrowd{\tau}{T} Y_l$ and $(X_k,Y_l) \in \beta$. As a consequence we must have that 
$(X_k,Y_j) \in \beta$.
\item[(iii)] $\tau . Z_{i,l}$, whenever $Y_j \arrowd{\tau}{T} Y_l$ and there is no $k$ 
such that: $X_i \arrowd{\tau}{S} X_k$ and $(X_k,Y_l) \in \beta$. As a consequence we must have that
$(X_i,Y_l) \in \beta$.
\item[(iv)] $\delta . Z_{k,l}$, whenever $X_i \arrowd{\delta}{S} X_k$ and 
$Y_j \arrowd{\delta}{T} Y_l$ and $(X_k,Y_l) \in \beta$.
\end{enumerate}

Notice that $U$ is clearly closed, standard and prioritized and it is also
guarded: as in the standard case, by contradiction, any $\tau$-cycle $Z_{i,j} \arrowd{\tau}{U} Z_{i,j}$ would imply either a $\tau$-cycle $X_i \arrowd{\tau}{S} X_i$ or a $\tau$-cycle $Y_j \arrowd{\tau}{T} Y_j$.
Moreover, notice that at most one term of the kind (iv) can occur due to the fact that $S$ and $T$ are time-deterministic, hence also $U$ is time-deterministic. 
Finally, notice that, due to condition $2$ over the relation $\beta$
and to the fact that $S$ and $T$ are both well-rooted, we have that $\tilde{Z}_R = \{ Z_{i,j} \mid (X_i,Y_j) \in \beta \cap (\tilde{X}_R \times \tilde{Y}_R) \}$
and that $U$ is well-rooted as well.

We now show that $P$ provably satisfies the equation set $U$ (with $Z_{1,1}$ as distinguished 
variable). To do this, we consider the equation sets $S' : \tilde{X} = \tilde{H'}$, $T' : \tilde{Y} = \tilde{J'}$ and $U' : \tilde{Z} = \tilde{K'}$ obtained by $\tau-$saturating $S$, $T$, and $U$ respectively. 
We also suppose that $\tilde{P} = \{ P_1, \dots , P_m \}$, with $P_1 \equiv P$, and 
$\tilde{Q} = \{ Q_1, \dots , Q_n \}$, with $Q_1 \equiv Q$, both with free variables in 
$\tilde{W}$, are such that $\cala \vdash \tilde{P} =  \tilde{H'} \{ \tilde{P} / \tilde{X} \}$
and $\cala \vdash \tilde{Q} =  \tilde{J'} \{ \tilde{Q} / \tilde{Y} \}$.

We choose expressions $\tilde{G} = \{ G_{i,j} \mid (X_i,Y_j) \in \beta \}$
as
$$
G_{i,j} \equiv
\left\{
\begin{array}{ll} 
\tau . P_i & {\rm if} \; Z_{i,j} \arrowd{\tau}{U} Z_{i,l} \; {\rm for} \; {\rm some}\; l\\[0cm]
P_i & {\rm otherwise} 
\end{array}
\right.
$$
and in the following we show that $\cala \vdash \tilde{G} =  \tilde{K'} \{ \tilde{G} / \tilde{Z} \}$. Notice that for $G_{i,j}$ such that $Z_{i,j} \in \tilde{Z}_R$ we have $G_{i,j} \equiv P_i$ because the expression $K_{i,j}$ cannot include terms of the kind
(iii). In particular $G_{1,1} \equiv P_1 \equiv P$.

Since $U$ is guarded (hence also $U'$ is guarded) we can conclude that any solution of $U$ (being it also a solution for $U'$ which is guarded) is provably equal to $P$. Hence $P$ is a solution for $U$.

For each equation $G_{i,j} =  K'_{i,j} \{ \tilde{G} / \tilde{Z} \}$ we have the following two cases:

\begin{itemize}

\item $Z_{i,j} \narrowd{\tau}{U} Z_{i,l}$ for any $l$, hence $G_{i,j} \equiv P_i$.
In this case we have the two following subcases:

\begin{itemize}

\item $X_i \arrowd{\delta}{S}$. In this case we have (since the standard equation set $S$ is prioritized) 
$X_i \narrowd{\tau}{S}$. Since $K_{i,j}$ does not include terms of the kind (iii), this implies 
$Y_j \narrowd{\tau}{T}$. Therefore the structure of $K_{i,j}$ is such that it may contain only terms of 
kind (i) and (iv), 
where: terms of kind (i) cannot be $\tau$ prefixes and exactly one term of kind (iv) must be present because (since $Y_j \narrowd{\tau}{T}$) the case
(ii)$(A)$ does not apply.
Due to the properties of relation $\beta$ the following holds for the terms included in $K_{i,j}' \{ \tilde{G} / \tilde{Z} \}$. 
The terms of kind (i), which become, by possibly
using axiom $(Tau1')$, terms $a.P_k$ for some $k$, are exactly (with possible repetitions)
the terms for which $X_i \arrowd{a}{S} X_k$.
Moreover, concerning the term of kind (iv), we have the following two cases. If $Z_{i,j} \in \tilde{Z}_R$ then
such a term is
$\delta.P_k$ for some $k$
(since the $\delta$ prefixed variable again belongs to $\tilde{Z}_R$ we are guaranteed it is not replaced by $\tau.P_k$): it is exactly the term for which $X_i 
\arrowd{\delta}{S} X_k$. Otherwise, such a term, due to $\tau-$saturation, is a $\delta$ prefix followed by $\tau$ and  it becomes, by possibly
using axiom $(Tau1')$, term $\delta.\tau.P_k$ for some $k$: it is exactly
the term for which $X_i 
\arrowd{\delta}{S} X_k$. Hence, by using axiom 
$(A3)$ 
to deal with repetitions and by applying the hypothesis 
we derive $K'_{i,j} \{ \tilde{G} / \tilde{Z} \} = H'_i \{ \tilde{P} / \tilde{X} \}
= P_i$.


\item $X_i \narrowd{\delta}{S}$. In this case we have that
also $Z_{i,j} \narrowd{\delta}{U}$, hence the structure of $K_{i,j}$ 
is such that it may contain only terms of kind (i) and (ii). 
Therefore this case works exactly as in the untimed case of~\cite{Mil89}:
due to the properties of relation $\beta$ the following hold.
By using a similar reasoning
as in the first subcase, terms of kind (i) and (ii) 
become, by possibly using 
axiom $(Tau1')$, exactly as prefixes included in the summation 
$H_i' \{ \tilde{P} / \tilde{X} \}$ (with possible repetitions). Hence, by using axioms 
$(A3)$ 
to deal with repetitions and by applying the hypothesis 
we derive $K_{i,j}' \{ \tilde{G} / \tilde{Z} \} = H_i' \{ \tilde{P} / \tilde{X} \}
= P_i$.

\end{itemize}

\item $Z_{i,j} \arrowd{\tau}{U} Z_{i,l}$ for some $l$, hence $G_{i,j} \equiv \tau.P_i$. 
In this case we have that the structure of $K_{i,j}'$ is such that it may contain only terms of 
kind (i), (ii) and (iii).
Due to the properties of relation $\beta$ the following holds for the terms included in $K_{i,j}' \{ \tilde{G} / \tilde{Z} \}$. 
Every term of kind (iii) becomes, by possibly using axiom $(Tau1')$, $\tau.P_i$.
Moreover (by using a similar reasoning
as in the first subcase) the terms of kind (i) and (ii) 
become, by possibly using 
axiom $(Tau1')$, exactly as the non-$\delta$ prefixes included in the summation 
$H_i \{ \tilde{P} / \tilde{X} \}$ (with possible repetitions). Hence, by using axioms 
$(A3)$,
$(Pri3)$ and $(Pri6)$ that allow $\delta$ prefixes to be removed in the presence of a $\tau$ alternative, we derive $K_{i,j} \{ \tilde{G} / \tilde{Z} \} = \tau.P_i +
H_i \{ \tilde{P} / \tilde{X} \}$. By applying the hypothesis and by 
using axiom $(Tau2)$ we have that $\tau.P_i +
H_i \{ \tilde{P} / \tilde{X} \} = \tau.P_i + P_i = \tau.P_i$.

\end{itemize}
In a completely symmetrical way we can also show that $Q$ provably satisfies $U$.
\end{proof}


Hence we have proved completeness over time-deterministic guarded  basic processes. 

\begin{thm}\label{compguardedtd}
Let $P, Q \in \calp$ be time-deterministic guarded processes. If $P \simeq_{T} Q$ then 
$\cala_{DT} \vdash P = Q$.
\end{thm}





Finally, by an analogous of Theorem \ref{unguard}, we get also completeness over unguarded time-deterministic basic processes. 

\begin{lem}\label{timedunguard}
For each time-deterministic process $P \in \calp$ there exists a time-deterministic guarded process $P' \in \calp$ such that $\cala \vdash P = P'$.
\end{lem}

\begin{proof}
Same proof as that of Theorem \ref{unguard} with the additional observation that the transformations performed preserve time-determinism.
\end{proof}

\begin{thm}\label{comptd}
Let $P, Q \in \calp$ be time-deterministic processes. If $P \simeq_{T} Q$ then 
$\cala_{DT} \vdash P = Q$.
\end{thm}

\subsubsection{Completeness for the Discrete Time Calculus}\label{DRTStaticCompl}




We first introduce, along the lines of \cite{concur05,mscs08}, a syntactical characterization that guarantees processes to be finite-state.

\begin{defi}
$\cale_{DT}^{fs}$ is the set of expressions
$E \in \cale_{DT}$ such that:
for any subterm $E'$ of $E$, every
free occurrence of a variable $X$ does not appear in $E'$
in the scope of a static 
operator, i.e. ``$\_ \pco{S} \_$'' or ``$\_ /L$'', and, if it appears in the scope of a ``$\_ +^t \_$'' operator, it is 
guarded inside such an operator
by a standard
action $\alpha$. 
$\calp_{DT}^{fs}$ is the set of closed $\cale_{DT}^{fs}$ expressions.
\end{defi}

\begin{lem}\label{staticremoval}
Let $P',P'' \in \calp$ be time-deterministic guarded processes. $P \equiv P' \lme{S} P''$, 
$P \equiv P' \sme{S} P''$, $P \equiv P' \pco{S} P''$, $P \equiv P' /L$ and $P \equiv P' +^t P''$
can be turned by the axiom system $\cala_{DT}$ into the form
$\sum_{1 \leq i \leq k} \gamma_i . P_i$, 
where $k \geq 0$ ($k = 0$ corresponds to the sum being $\nil$) and there exists at most one $i$, with $1 \leq i \leq k$, such that $\gamma_i = \delta$. 
Moreover, $\{ (\gamma_i,P_i) \mid 1 \leq i \leq k \}=\{ (\gamma,Q) \mid  P \arrow{\gamma} Q\}$.\footnote{Since processes $P_i$ and labels $\gamma_i$, with $1 \leq i \leq k$, are those
such that $P \arrow{\gamma_i} P_i$ this implies that: if $\gamma_h = \delta$ for some $h$, then $\gamma_i \neq \tau$ for all $i$, with $1 \leq i \leq k$.}
\end{lem}

\begin{proof}


By means of axiom $(Rec1)$ and $(Pri)$ axioms it is immediate to show that $P'=P'_{next} \equiv \sum_{1 \leq i \leq n} \gamma'_i . P'_i$ and $P''=P''_{next} \equiv \sum_{1 \leq i \leq m} \gamma''_i . P''_i$,
where $n,m \geq 0$ and there exists at most one $i$, with $1 \leq i \leq n$, such that $\gamma_i = \delta$ and 
at most one $j$, with $1 \leq j \leq m$, such that $\gamma_j = \delta$. 
Moreover, $\{ (\gamma'_i,P'_i) \mid 1 \leq i \leq n \}=\{ (\gamma,Q) \mid  P' \arrow{\gamma} Q\}$ and
$\{ (\gamma''_i,P''_i) \mid 1 \leq i \leq m \}=\{ (\gamma,Q) \mid  P'' \arrow{\gamma} Q\}$.


The cases of left merge, i.e. $P \equiv P' \lme{S} P''$, hiding, i.e. $P \equiv P'/L$,
are proved by using the corresponding axioms (that are standard) to perform the following transformations:
\footnote{We assume prefix to take precedence over $\lme{S}$, 
$\sme{S}$ and $\pco{S}$ operators, when writing terms.}
$$P' \lme{S} P'' = P'_{next} \lme{S} P'' = \sum_{i \leq n} ( \gamma'_i . P'_i \lme{S} P'') =
\sum_{i \leq n, \gamma'_i \notin S \cup \{\delta\}} \gamma'_i . (P'_i \pco{S} P'')$$
$$P'/L = P'_{next}/L = \sum_{i \leq n} (\gamma'_i . P'_i)/L =
\sum_{i \leq n, \gamma'_i \in L} \tau . (P'_i/L) +
\sum_{i \leq n, \gamma'_i \notin L} \gamma'_i . (P'_i/L) 
$$

We now consider the case of synchronization merge, i.e. $P \equiv P' \sme{S} P''$.
We initially have:
$$P' \sme{S} P'' = P'_{next} \sme{S} P''_{next} = 
\sum_{i \leq n} (\gamma'_i . P'_i \sme{S} P''_{next}) =
\sum_{i \leq n} (\sum_{j \leq m} (\gamma'_i . P'_i \sme{S} \gamma''_j . P''_j))
$$
where the second equality is obtained as follows.
First of all we observe that, for any (possibly empty) set of terms $Q_i$ and visible actions $a_i$, with $i \leq h$, we have that
$\sum_{i \leq h} a_i.Q_i = vis(\sum_{i \leq h} a_i.Q_i)$ by applying axioms $(Vis1-3)$. Similarly, for any (possibly empty) set of terms $Q_i$ and standard actions $\alpha_i$, with $i \leq h$, we have that
$\sum_{i \leq h} \alpha_i.Q_i = pri(\sum_{i \leq h} \alpha_i.Q_i)$ by applying axioms $(Pri1-2)$ and $(Pri4)$. We have two cases for the structure of $P'_{next}$. If there exists $k \leq n$ such that $\gamma'_k = \delta$ then
we have:
$\sum_{i \leq n, i \neq k} \gamma'_i . P'_i = vis(\sum_{i \leq n, i \neq k} \gamma'_i . P'_i)$, 
hence $P'_{next} \sme{S} P''_{next} = \delta.P'_k \sme{S} P''_{next} + 
(\sum_{i \leq n, i \neq k} \gamma'_i . P'_i) \sme{S} P''_{next}$ by axiom $(SM7)$; moreover, since, for any $I \subset \{ i | i \leq n, i \neq k \}$, it holds
$\sum_{i \in I} \gamma'_i . P'_i = pri(\sum_{i \in I} \gamma'_i . P'_i)$, we conclude that
$(\sum_{i \leq n, i \neq k} \gamma'_i . P'_i) \sme{S} P''_{next}
= \sum_{i \leq n, i \neq k} (\gamma'_i . P'_i \sme{S} P''_{next})$ by repeatedly applying axiom $(SM6)$. Otherwise, we directly have: for any $I \subset \{ i | i \leq n \}$, it holds $\sum_{i \in I} \gamma'_i . P'_i = pri(\sum_{i \in I} \gamma'_i . P'_i)$, hence
$(\sum_{i \leq n} \gamma'_i . P'_i) \sme{S} P''_{next}
= \sum_{i \leq n} (\gamma'_i . P'_i \sme{S} P''_{next})$ by repeatedly applying axiom $(SM6)$.
The third equality is obtained by applying the same procedure above to terms $P''_{next}$ for each of the $n$ summands.
Therefore 
$$
P'_{next} \sme{S} P''_{next} =
\sum_{i \leq n, \gamma'_i = \tau} (\tau . P'_i \sme{S} P'')  +
\sum_{j \leq m, \gamma''_j = \tau} (P' \sme{S} \tau . P''_j) +
\sum_{i \leq n, j \leq m, \gamma'_i = \gamma''_j \in S \cup \{ \delta \}} \gamma'_i. (P'_i \pco{S} P''_j) 
$$

In the following we show that, for any closed normal forms $P',P''$,
we can turn $P \equiv P' \sme{S} P''$ into $\sum_{1 \leq i \leq n} \gamma_i . P_i $ such that the arguments of the sum correspond to the transitions of $P$ by inducing on the following measure: the maximal length of the sequences of $\tau$ transitions performable by $P'$ plus the maximal length of the sequences of $\tau$ transitions performable by $P''$.
From this result we can conclude that any such $P$ can be turned into the desired form because, since normal forms include only guarded recursion, $P',P''$ cannot include cycles of $\tau$ loops (and are finite-state), hence
the sequences of $\tau$ transitions they can perform are bounded.
\begin{itemize}
\item The base case of the induction corresponds to such a measure being $0$, i.e. both $P'$ and $P''$ cannot
perform $\tau$ transitions. This means that, when transforming $P$ in the sum-form above, the first two sums are not obtained,
hence the assertion obviously holds.
\item
The inductive case is performed by just observing that the summands $\tau . P'_i \sme{S} P''$ and $P' \sme{S} \tau . P''_j$
obtained by transforming $P$ into the sum-form above can be rewritten, by using axiom $(SM4)$, into 
$pri(P'_i \sme{S} P'')$ and $pri(P' \sme{S} P''_j)$, respectively. For such terms
$P'_i \sme{S} P''$ and $P' \sme{S} P''_j$ we can apply the induction hypothesis and turn them into the form $\sum_{1 \leq i \leq n} \gamma'''_i . P'''_i$ 
such that the arguments of the sum correspond to their transitions.
The obtained term $pri(\sum_{1 \leq i \leq n} \gamma'''_i . P'''_i)$ can be 
then turned into $\sum_{1 \leq i \leq n, \gamma'''_i \neq \delta} \gamma'''_i . P'''_i$ by using axioms $(Pri1-4)$.  
It is now easy to observe that, once terms 
$pri(P'_i \sme{S} P'')$ and $pri(P' \sme{S} P''_j)\,$ have been turned into the form
$\sum_{1 \leq i \leq n, \gamma'''_i \neq \delta} \gamma'''_i . P'''_i $ inside the sum-form above,
the arguments of the obtained overall sum correspond to the transitions of $P' \sme{S} P''$.
This is because,
according to the operational rules for synchronization merge, $P' \sme{S} P''$ is endowed, besides visible and $\delta$ transitions obtained by synchronization of visible and $\delta$ transitions of $P'$ and $P''$, with the visible transitions of $P'_i \sme{S} P''$ ($P' \sme{S} P''_j$) whenever $P' \arrow{\tau} P'_i$ ($P'' \arrow{\tau} P''_j$),


\end{itemize}

Let us then consider the case of parallel composition operator, i.e. $P \equiv P' \pco{S} P''$.
We initially have 
$$P' \pco{S} P'' = P'_{next} \lme{S} P'' + P''_{next} \lme{S} P' + P'_{next} \sme{S} P''_{next}$$
We then apply the transformation for $P'_{next} \lme{S} P''$ considered 
in the proof for the case of left merge (and we also apply it to $P''_{next} \lme{S} P'$)
and the transformation for $P'_{next} \sme{S} P''_{next}$ considered
in the proof for the case of synchronization merge. Here, however, instead of dealing with the first and second
sums of the sum form obtained from $P'_{next} \sme{S} P''_{next}$ by means of an inductive transformation, we
just get rid of them as follows.
Since, for any $i$ and $j$, we have 
$\tau . P'_i \sme{S} P'' = pri(P'_i \sme{S} P'')$ and $P' \sme{S} \tau . P''_j = pri(P' \sme{S} P''_j)$
and terms $P'_i \sme{S} P''$ and $P' \sme{S} P''_j$
already occur in the transformation of $P'_{next} \lme{S} P''$ and $P''_{next} \lme{S} P'$ (by additionally applying
axiom $(Par)$ to parallel composition and commutativity via $(SM1)$ ) of the form 
$\tau.((P'_i \sme{S} P'') + P''')$ and $\tau.((P' \sme{S} P''_j) + P''')$, respectively, by using axioms $(Tau2)$ and $(A3)$,
we derive
\cws{8}{6}{P' \pco{S} P'' =
\!\! \sum_{i \leq n, \gamma'_i \notin S \cup {\delta}} \!\gamma'_i . (P'_i \pco{S} P'') +
\!\! \sum_{i \leq m, \gamma''_i \notin S \cup {\delta}} \!\gamma''_i . (P' \pco{S} P''_i) +
\!\! \sum_{i \leq n, j \leq m, \gamma'_i = \gamma''_j \in S \cup {\delta}
} \!\gamma'_i. (P'_i \pco{S} P''_j)}
where, in the second sum, we also have exploited the commutativity of ``$\pco{S}$'' derived by axiom $(SM1)$.
It is immediate to observe that, being
$\arrow{\gamma'_i} P'_i$, with $i \leq n$, 
and 
$\arrow{\gamma''_i} P''_i$, with $i \leq m$, the outgoing transitions of $P'$ and $P''$, respectively, the
arguments of the above sum correspond to the transitions derived for $P$ from the operational rules of parallel composition.

Finally, let us consider the case $P \equiv P' +^t P'' = P'_{next} +^t P''_{next}$.
We have two cases. If for both $P'$ and $P''$ there exist $h$, $k$ such
that $\gamma'_h = \delta$ and $\gamma''_k = \delta$, then we have that $P = 
(\sum_{1 \leq i \leq n, i \neq h} \gamma'_i . P'_i + \delta.P'_h) +^t (\sum_{1 \leq j \leq m, j \neq k} \gamma''_j . P''_j  + \delta.P'_k)$. By using axioms $(TCh7)$ and $(TCh1)$ we derive
$P = 
((\sum_{1 \leq i \leq n, i \neq h} \gamma'_i . P'_i) +^t (\sum_{1 \leq j \leq m, j \neq k} \gamma''_j . P''_j)) +
((\sum_{1 \leq i \leq n, i \neq h} \gamma'_i . P'_i) +^t \delta.P'_k) +
(\delta.P'_h +^t \sum_{1 \leq j \leq m, j \neq k} \gamma''_j . P''_j) +
(\delta.P'_h +^t \delta.P'_k)$. Then, by observing that
$\sum_{1 \leq i \leq n, i \neq h} \gamma'_i . P'_i = pri(\sum_{1 \leq i \leq n, i \neq h} \gamma'_i . P'_i)$ and $\sum_{1 \leq j \leq m, j \neq k} \gamma''_j . P''_j =
pri(\sum_{1 \leq j \leq m, j \neq k}$ $\gamma''_j . P''_j)$ can be derived by using axioms
$(Pri1-2)$ and $(Pri4)$ because, since both $P'$ and $P''$ are time deterministic, all involved actions are standard actions, we derive $P = (\sum_{1 \leq i \leq n, i \neq h} \gamma'_i . P'_i) + (\sum_{1 \leq j \leq m, j \neq k} \gamma''_j . P''_j) +
\delta.(P'_h +^t P'_k)$ by using axioms $(TCh5-6)$, $(Pri3)$ and $(A3)$.
Otherwise, if there exists one between $P'$ and $P''$ (let us suppose it to be $P'$, the other case is symmetric) such that there is no $h$ yielding $\gamma'_h = \delta$, then
we directly have: $P = (\sum_{1 \leq i \leq n} \gamma'_i . P'_i) + (\sum_{1 \leq j \leq m, j \neq k} \gamma''_j . P''_j)$ if there is $k$ yielding $\gamma''_k = \delta$;
$P = (\sum_{1 \leq i \leq n} \gamma'_i . P'_i) + (\sum_{1 \leq j \leq m} \gamma''_j . P''_j)$ otherwise. Both equalities are derived by preliminarily observing that 
$\sum_{1 \leq i \leq n} \gamma'_i . P'_i = pri(\sum_{1 \leq i \leq n} \gamma'_i . P'_i)$ can be derived by using axioms
$(Pri1-2)$ and $(Pri4)$ (because all involved actions are standard actions) and
by subsequently applying axioms $(TCh5)$ and, for the first equality, $(Pri3)$.
\end{proof}

%
%
%
%
%
%
%
%
%
%
%
%
%
%
%
%
%
%
%
%
%
%
%
%
%
%

\begin{thm}\label{unguardfs}
For each process $P \in \calp_{DT}^{fs}$ there exists a time-deterministic 
process $P' \in \calp$ such that
$\cala \vdash P = P'$.

\begin{proof}

We show, by structural induction over the syntax of expressions $E \in \cale_{DT}^{fs}$ 
that $E$ can be turned 
into a 
basic expression $F \in \cale$ 
such that
$\cala \vdash F = E$ and: 
\begin{enumerate}
\item
For any variable $X$: if $X$ occurs free in $F$ then $X$ occurs free in $E$. 
\item
If all free variable $X$ occurrences are guarded in $E$ by a standard action $\alpha$ then the same holds in $F$.
\item 
Every free occurrence of a variable $X$ in $F$ that is in the scope
of a ``$\_+\_$'' operator is 
guarded inside such an operator
by a standard
action $\alpha$.

\item $F\{\nil/X \mid X\,$ occurs free in $F \}$ is time-deterministic.

%
%
\end{enumerate}


The base cases of the induction are $E \equiv \snil$, $E \equiv \nil$ and 
$E \equiv X$: for the first one we just apply axiom $(Ter)$, the second and third ones are
trivial because they are of the
desired form already.
The inductive cases are the following ones:
\begin{itemize}

\item If $E \equiv \gamma^t.E'$ 
then $E$ can be turned into the desired form by directly exploiting
the inductive argument over $E'$ and by applying axioms $(TPre1-2)$.

\item If $E \equiv \ms{rec}X. E'$
then $E$ can be turned into the desired form by directly exploiting
the inductive argument over $E'$ and by considering $F \equiv \ms{rec}X. F'$: since the obtained expression $F'$ satisfies items $3$ and $4$, we have that $G \equiv F \{\nil/Y \mid Y\,$ occurs free in $F \}$ is time-deterministic because, in states of $G$, terms $\ms{rec}X. G'$ (for some $G'$) that are in the scope of a  ``$\_+\_$'' operator are 
guarded inside such an operator
by a standard
action $\alpha$.

\item If $E \equiv E' \pco{S} E''$ 
then we can turn $E$ into the desired form as follows. By exploiting 
the inductive argument over $E'$ and $E''$, and by observing that
$E$ cannot include free variables, we obtain the two closed basic 
terms $F'$ and $F''$ that, by using 
Lemma~\ref{timedunguard} 
can be transformed into a guarded time-deterministic basic processes.
As a consequence we have that the closed term $P_1$  
obtained by replacing both $E'$ and $E''$ inside $E$ with such terms has
a finite transition system and is time-deterministic. 
We can therefore turn $P_1$ into a time-deterministic basic process 
as follows.
%
%
Let $P_1 \dots P_n$ be the (finite) states of the transition system of $P_1$. 
Due to Lemma~\ref{staticremoval}, 
for each 
$i \in \{ 1 \dots n \}$
we can rewrite $P_i$ into the form
$P_i = \sum_{j \leq m_i} \gamma^i_j . P_{k^i_j} $, 
hence we can characterize the behavior of
$P_1$ by means of a time-deterministic prioritized standard guarded equation set  
exactly as we did for the proof of Theorem \ref{reprtimedet}.
Such an equation set is 
guarded because, by Lemma~\ref{staticremoval}, the arguments of
the sums above are the outgoing transitions of the states of $P_1$ and
$P_1$ is the parallel composition of two guarded basic processes,
hence (since it cannot turn visible actions into $\tau$ ones) every cycle in its transition
system contains at least a non-$\tau$ action.
Since guarded equation sets have a unique solution (Theorem \ref{onesol}), we have that there exists a 
time-deterministic (due to Theorem \ref{soldet}) basic process $P'_1$ such that $P'_1 = P_1$.
%

\item If $E \equiv E' /L$ 
then we can turn $E$ into the desired form as follows. Similarly as for parallel composition, by exploiting 
the inductive argument over $E'$, and by observing that
$E$ cannot include free variables, we obtain a closed basic 
term $F'$ that, by using 
Lemma~\ref{timedunguard} 
can be transformed into a guarded time-deterministic basic processes $P$.
We then proceed exactly as in the proof of Proposition~$6.5$ of~\cite{mscs08} to replace
with $\tau.\_$ each $a.\_$ prefix occurring in $P$ such that $a \in L$, thus obtaining a term $P'$ (we can distribute top-down and bottom-up the hiding inside $P/L$ thanks to axiom $(RecHi)$ and to the fact that,
due to axiom $(Hi2)$, $\delta.\_$ prefixes can be dealt with as $a.\_$ prefixes with $a \notin L$). We then turn the weakly guarded basic process $P'$ into a guarded $P''$ by 
applying Lemma~\ref{timedunguard} and we finally turn $P''/L$ into a basic process 
by characterizing its behavior by means of a time-deterministic prioritized standard guarded equation set exactly as we did for parallel composition (notice that now the $\_/L$ has no hiding effect
so, being $P''$ guarded, every cycle in the transition system of $P''/L$ contains at least a non-$\tau$ action).

\item If $E \equiv E' +^t E''$ 
then we can turn $E$ into the desired form as follows. By exploiting 
the inductive argument over $E'$ and $E''$, and by observing that
all free variables in $E$ are guarded by a standard action $\alpha$, we obtain the two basic 
expressions $F'$ and $F''$ for which, due to item $2$, the same property holds.
$F'$ and $F''$, by using the statement in the proof of 
Lemma~\ref{unguard} (on which Lemma~\ref{timedunguard} is based), 
can be transformed into guarded expressions $H'$ and $H''$, respectively, such that the same property holds 
(because weakly unguardedness elimination in Lemma~\ref{unguard} introduces a $\tau$ guard). Moreover the obtained $H'$ and $H''$ are both expressions in $\cale$ because both $F'$ and $F''$ satisfy item $3$, hence, during the induction of Lemma~\ref{unguard}, $pri(\_)$ operators are always generated in front of guards. Finally, notice that, since both $F'$ and $F''$ satisfy item $4$, the same holds also for $H'$ and $H''$, because the transformations performed  in Lemma~\ref{unguard} preserve time-determinism.
%
%
%
We consider $G'$ and $G''$ such that the set of free variables of $G'$ and $G''$ is disjoint from 
the set of free variables of $H'$ and $H''$
and $(G' +^t G'') \{\alpha_X.E_X / X \mid X$ occurs free in $G',G'' \} \equiv H' +^t H''$ (such $G'$, $G''$ exist because 
all free variables in $H'$, $H''$ are guarded by a standard action $\alpha$).  
We have that the process $P_1 \equiv (G' +^t G'') \{\alpha_X.\nil / X \mid X$ occurs free in $G',G'' \}$ has a finite transition system and is time-deterministic (being a $+^t$ of two time-deterministic processes).
We can therefore turn $P_1$ into a time-deterministic basic process 
as follows.
%
%
Let $P_1 \dots P_n$ be the (finite) states of the transition system of $P_1$. 
Due to Lemma~\ref{staticremoval}, 
for each 
$i \in \{ 1 \dots n \}$
we can rewrite $P_i$ into the form
$P_i = \sum_{j \leq m_i} \gamma^i_j . P_{k^i_j} $, 
hence we can characterize the behavior of
$P_1$ by means of a time-deterministic prioritized standard guarded equation set  
exactly as we did for the proof of Theorem \ref{reprtimedet}: it is 
guarded because, by Lemma~\ref{staticremoval}, the arguments of
the sums above are the outgoing transitions of the states of $P_1$ and
$P_1$ is the $+^t$ of two guarded basic processes.
Since guarded equation sets have a unique solution (Theorem \ref{onesol}), we have that there exists a 
time-deterministic (due to Theorem \ref{soldet}) basic process $P'_1$ such that $P'_1 = P_1$.
By applying exactly the same axioms in the same way, we can transform $(G' +^t G'') \{\alpha_X.E_X / X \mid X$ occurs free in $G',G'' \} \equiv H' +^t H''$ into $F \equiv G_1\{\alpha_X.E_X / X \mid X$ occurs free in $G',G'' \}$, with $G_1$ such that $G_1 \{\alpha_X.\nil / X \mid X$ occurs free in $G',G'' \} \equiv P'_1$.  Moreover, we have that $F$ satisfies item $2$ because all
its free variables are inside  $E_X$ terms. Also the other items are satisfied because $P'_1$ is time-deterministic and every $E_X$ term satisfies items $3$ and $4$.
\qedhere
\end{itemize}

%
%
%

\end{proof}

\end{thm}


From Theorem~\ref{unguardfs} and Theorem~\ref{comptd} 
we derive the completeness of $\cala_{DT}$ over processes of $\calp_{DT}^{fs}$.

\begin{thm}
Let $P,Q \in \calp_{DT}^{fs}$. If $P \simeq_{T} Q$ then $\cala_{DT} 
\vdash P = Q$.

\end{thm}

\section{Related Work}
\label{SectRW}

In the following we consider related work on priority mechanisms and on expressing time in process algebra. Here we just comment on the characteristics of our axiomatization using notions introduced in~\cite{ML}. In~\cite{ML}  
the form of Milner's axiomatization~\cite{Mil89,Mil} is discussed by observing that it is not in pure equational
Horn logic: the axioms involve non-equational side-conditions and/or they are schematically infinitary, as, e.g., the folding axiom $(Rec2)$. Our axiomatization is an extension/variant of Milner's one that preserves its characteristics.

\subsection{Local and Global Priority}
\label{SSSectLGP}

In classical prioritized calculi the 
parallel composition operator is usually managed in two ways: either
by implementing {\it local} pre-emption or {\it global} pre-emption 
(see~\cite{CLN2}). 

Assuming local pre-emption means that $\tau$ actions of a sequential process 
(i.e. a process not including a parallel composition operator in its
immediate behavior)
may pre-empt only actions $\delta$ of the same sequential process. 
For instance in $\tau. E \pco{\emptyset} \delta . F$ the action 
$\delta$ of the righthand process is not pre-empted by the action $\tau$ of 
the lefthand process, as instead happens if we assume global pre-emption. 
It is easy to see that an extension to our basic calculus with a parallel composition operator 
implementing local pre-emption makes it possible to preserve congruence
w.r.t. Milner's observational congruence. However, unfortunately, expressing
local pre-emption makes it necessary to introduce location information
in the semantics and does not allow to directly produce an 
axiomatization by means of standard techniques.

If global pre-emption is, instead, assumed, then standard Milner's notion of
observational congruence is not a congruence for the parallel composition operator 
(see~\cite{CLN2}). 
%
This is because, e.g.,  $\tau. \nil$
 is observationally congruent to $\ms{rec}X.\tau.X$, but
$\tau. \nil \pco{\emptyset} \delta.P$, whose semantics is that of $\tau. \delta. P$, is not observationally congruent to $\ms{rec}X.\tau.X \pco{\emptyset} \delta.P$, whose semantics, due to global pre-emption, is
that of $\ms{rec}X.\tau.X$.
In general note that the problem with congruence is related to the
behavior of parallel composition for {\it processes Q which may initially execute neither
``$\tau$'' prefixes, nor ``$\delta$'' prefixes}, among which is $\nil$
(for any such Q, the use of $\tau.Q \simeq \ms{rec}X.(\tau.X + Q)$ with
the context ``$\_ \pco{\emptyset} \delta.P$'' provides
a counterexample to congruence). 
In this case a possibility is to resort to a 
finer notion of observational congruence (which is divergent sensitive in 
certain cases) similar to those presented in~\cite{NCCC} and~\cite{HL}.

An alternative approach is to adopt an explicit priority operator, see e.g.~\cite{ACFI}. In this context~\cite{ACFI} shows 
that a finite complete axiomatization is admitted only if the action set of the considered process algebra is finite. With our simpler
form of priority such a limitation is not needed: as for the axiomatizations in~\cite{Mil89,Mil,concur05,mscs08} we can assume a denumerable action set and we have a finite complete axiomatization (with axioms using action variables).

\subsection{Time}

When priority derives from time (maximal progress assumption), 
i.e.\ when $\delta$ actions represent time delays and standard actions  
are executed in zero time, there is an alternative choice w.r.t.\ that of just adopting global priority, 
as, e.g., in the above mentioned approach~\cite{HL}.
Conceptually, under the timed interpretation, the problem with congruence derives from the fact that the 
parallel composition operator deals with the 
terminated process $\nil$ (and in general with processes which may initially execute
neither $\tau$ actions nor $\delta$ actions) as if it let time pass.
For example $\nil \pco{\emptyset} \delta$ may execute $\delta$ and become
$\nil \pco{\emptyset} \nil$. This is obviously in contrast with the fact
that $\nil$ is weakly bisimilar to $\ms{rec}X.\tau.X$, which is clearly
a process that does not let time pass (in the context of time it 
represents a Zeno process which executes infinite $\tau$ actions in the same 
time point): it originates a so-called {\it time deadlock}.
The solution adopted in~\cite{HenR} and in this paper is, instead, to consider, as processes which can
let time pass, only processes which can actually execute $\delta$ actions.
In this way $\nil$ is interpreted not as a terminated
process which may let time pass, but as a time deadlock. As a consequence 
the behaviour of parallel composition is defined, as in~\cite{HenR}, in such a way that the {\it absence} of 
$\delta$ actions within the
actions executable by a process (which means that the process cannot let
time pass) pre-empts the other process from executing a timed action $\delta$, see Table \ref{TabRulesDRT}.

As we already explained in the introduction, in order to produce an axiomatization
we have to face the problem of standard axiom $\ms{rec} X. (\tau.X + E) = \ms{rec} X. \tau. E$ unsoundness.
In order to overcome this problem, in the above mentioned approach of~\cite{HL} 
the distinguished symbol ``$\perp$'' is introduced, which 
represents an ill-defined term that can be removed from a summation only if 
a silent computation is possible. In this way by considering the rule
$\ms{rec} X. (\tau.X + E) = \ms{rec} X. \tau. (E \; + \perp)$ the resulting
term that escapes divergence can be turned into a ``normal'' term only
if $E$ may execute a silent move. 
This law is surely sound (over terms without ``$\perp$'') 
also in our language, but is not sufficient to
achieve completeness. Since, differently from~\cite{HL}, 
we do not impose conditions about stability in our definition of observational
congruence, we can escape divergence
not only when $E$ includes a silent move but {\it for all possible} terms $E$.
For example in our calculus (but not in~\cite{HL}) the term 
$\ms{rec} X. \tau.X$ is equivalent to $\tau . \nil$ (as in standard CCS), 
so we can escape divergence in $F$ even if $\tau.X$ has not a silent 
alternative inside $\ms{rec} X$.
In our case $\tau$ divergence can always be escaped
by turning $E$ into $\ms{pri}(E)$ and
the strongly guarded terms we obtain
are always ``well-defined'' terms.
Notice that the introduction of this auxiliary operator, representing priority ``scope'', is crucial for being able
to axiomatize the priority of $\tau$ actions over $\delta$ actions
when standard observational congruence is considered.
Since we have to remove $\delta$ actions performable by a term $E$ 
even if $E$ does not include a silent move, we cannot do this by employing a special 
symbol like ``$\perp$'' instead of using an operator. This is because
$\perp$ must somehow be removed at the end of the deletion process
(in~\cite{HL} $\perp$ is eliminated by silent alternatives) 
in order to obtain a ``normal'' term.

Concerning the discrete interpretation of time, our Discrete Time Calculus just differs from~\cite{HenR} for the choice of using CSP~\cite{Hoa} parallel composition and hiding, instead of CCS~\cite{Mil} parallel composition and restriction, as done in~\cite{HenR}. This 
allowed us to provide an axiomatization that is complete for finite-state processes
by applying the technique of~\cite{concur05,mscs08} that requires dynamic generation
of ``$\tau$'' actions to be expressed by a dedicated hiding operator (and not by another operator, like CCS~\cite{Mil} parallel,
that mixes generation of $\tau$ actions with other mechanisms).
Notice that an analogous complete axiomatization of~\cite{HenR} could have been obtained by expressing the CCS-like parallel composition of~\cite{HenR} in terms of parallel composition and hiding by using the generic process algebra TCP+REC introduced 
in~\cite{mscs08}, that was produced during the mentioned collaboration with Prof.\ Jos Baeten. 

Finally, concerning limitations of our approach, we consider the possible extension to discrete-time process algebras with multiple clocks and general clock scoping, as e.g. the calculus in~\cite{NLM}. In this context a complete axiomatization of maximal progress for weakly guarded finite-state processes cannot be achieved by a direct application of our approach. 
This because a $\tau$-circle could involve 
two or more processes, each of which has an (independent) outgoing $\delta$ transition: in this situation we could not directly use our axiomatization to remove the $\tau$-circle while maintaining time-determinism.

\section{Future Work}
\label{SectConc}

We just make some remarks concerning future work. We plan to apply the axiomatization techniques used in this paper also in the context of stochastic time, in particular to provide a complete axiomatization of Markovian observational congruence for Revisited Interactive Markov Chains~\cite{mtcs02}, where, due to maximal progress, $\tau$ actions are prioritized w.r.t.\ Markovian delays. Differently from
original Interactive Markov Chains~\cite{Her}, where the maximal progress assumption is implemented by a global priority mechanism and a $\tau$-divergent sensitive equivalence like that in~\cite{HL} is adopted, Revisited Interactive Markov Chains are based on a parallel composition that requires both
processes to let time pass, as for our discrete time calculus. As a consequence the axiomatization of
the basic calculus introduced in this paper can also form a basis (once $\delta$ prefixes are turned into Markovian delays and 
Markovian observational congruence of ~\cite{mtcs02} is considered) for completely axiomatizing the Markovian calculus of~\cite{mtcs02}.

\section*{Acknowledgments}
\noindent We thank the anonymous reviewers for their careful reading, comments and suggestions.


\end{document}